\documentclass[letter,11pt,final]{article}
\usepackage{typearea}
\typearea{14}
\makeatletter

\usepackage{fullpage}

\usepackage{authblk}
\usepackage{amsmath, amssymb, amsthm, thmtools, amsfonts, bm}
\usepackage{dsfont}
\usepackage{nicefrac}
\usepackage{cite}
\usepackage[numbers,sort]{natbib}
\usepackage{url}
\usepackage[usenames,dvipsnames]{xcolor}
\usepackage{algorithm}
\usepackage{algorithmicx}
\usepackage[noend]{algpseudocode}
\usepackage{epstopdf}
\usepackage[shortlabels]{enumitem}
\usepackage{mleftright}

\usepackage{framed}
\usepackage[framemethod=tikz]{mdframed}
\usepackage{xspace}

\usepackage[colorlinks,citecolor=blue,linkcolor=BrickRed]{hyperref}
\usepackage{cleveref}

\theoremstyle{plain}
\newtheorem{thm}{Theorem}[section]

\newtheorem{cor}[thm]{Corollary}
\newtheorem{prop}[thm]{Proposition}

\newtheorem{fact}[thm]{Fact}
\newtheorem{lem}[thm]{Lemma}
\newtheorem{Def}[thm]{Definition}
\newtheorem{obs}[thm]{Observation}
\newtheorem{cla}[thm]{Claim}

\let\poly\relax
\DeclareMathOperator*{\poly}{poly}

\newcommand{\eps}{\epsilon}%
\newcommand{\E}{\mathbb{E}}%
\newcommand{\Var}{\mathrm{Var}}

\newcommand{\midd}{\;\middle|\;}

\newcommand{\palette}{P}%
\newcommand{\calH}{\mathcal{H}}
\newcommand{\B}{\mathcal{B}}
\newcommand{\e}{\mathcal{E}}
\newcommand{\C}{\mathcal{C}}

\newenvironment{wrapper}[1]
{
	\begin{center}
		\begin{minipage}{\linewidth}
			\begin{mdframed}[hidealllines=true, backgroundcolor=gray!20, leftmargin=0cm,innerleftmargin=0.4cm,innerrightmargin=0.4cm,innertopmargin=0.4cm,innerbottommargin=0.4cm,roundcorner=10pt]
				#1}
			{\end{mdframed}
		\end{minipage}
	\end{center}
}

\title{Online Edge Coloring  Algorithms via the Nibble Method}
\author[1]{Sayan Bhattacharya}
\author[2]{Fabrizio Grandoni}
\author[3]{David Wajc\footnote{Work done while the author was at Carnegie Mellon University.}}
\affil[1]{University of Warwick}
\affil[2]{IDSIA, USI-SUPSI}
\affil[3]{Stanford University}
\date{}

\begin{document}


\begin{titlepage}
\maketitle
\pagenumbering{gobble}

\begin{abstract}
	Nearly thirty years ago, Bar-Noy, Motwani and Naor [IPL'92] conjectured that an online $(1+o(1))\Delta$-edge-coloring algorithm exists for $n$-node graphs of maximum degree $\Delta=\omega(\log n)$. This conjecture remains open in general, though it was recently proven for bipartite graphs under \emph{one-sided vertex arrivals} by Cohen et al.~[FOCS'19]. In a similar vein, we study edge coloring under widely-studied relaxations of the online model.
	
	\smallskip
	
	Our main result is in the \emph{random-order} online model. For this model, known results fall short of the Bar-Noy et al.~conjecture, either in the degree bound [Aggarwal et al.~FOCS'03], or number of colors used [Bahmani et al.~SODA'10]. We achieve the best of both worlds, thus resolving the Bar-Noy et al.~conjecture in the affirmative for this model. 
	
	\smallskip
	
	Our second result is in the adversarial online (and dynamic) model with \emph{recourse}. A recent algorithm of Duan et al.~[SODA'19] yields a $(1+\epsilon)\Delta$-edge-coloring with poly$(\log n/\epsilon)$ recourse. We achieve the same with poly$(1/\epsilon)$ recourse, thus removing all dependence on $n$.
	
	\smallskip
	
	Underlying our results is one common offline algorithm, which we show how to implement in these two online models. Our algorithm, based on the R\"odl Nibble Method, is an adaptation of the distributed algorithm of Dubhashi et al.~[TCS'98]. The Nibble Method has proven successful for distributed edge coloring. We display its usefulness in the context of online algorithms.
\end{abstract}

\newpage

\end{titlepage}
\pagenumbering{arabic}
\section{Introduction}

Edge coloring is the problem of assigning one of $k$ colors to all edges of a simple graph, so that no two incident edges have the same color. The objective is to minimize the number of colors, $k$.
The edge coloring problem goes back to the 19th century and studies of the four-color theorem \cite{tait1880remarks,petersen1898theoreme}. 
In 1916, K\"onig \cite{konig1916graphen}, in what many consider to be the birth of matching theory, proved that any bipartite graph of maximum degree $\Delta$ is colorable using $\Delta$ colors. (Clearly, no fewer colors suffice.) Nearly half a century later, Vizing \cite{vizing1964estimate} proved that any general graph is $(\Delta+1)$-edge-colorable. Vizing's proof is algorithmic, yielding such a coloring in polynomial time. This is likely optimal, as it is NP-hard to determine if a general graph is $\Delta$-edge-colorable  \cite{holyer1981np}. 
Algorithms for the edge coloring problem have been studied in several different  models of computation, including offline, online, distributed, parallel, and dynamic models (see, e.g.,  \cite{cole2001edge,cohen2019tight,su2019towards,chang2018complexity,karloff1987efficient,motwani1994probabilistic,duan2019dynamic,charikar2021improved} and references therein.)
In this work, we study the edge coloring problem in online settings.

\medskip
\noindent {\bf Online edge coloring:} Here, an adversary picks an $n$-node graph $G$ of maximum degree $\Delta$ (the algorithm knows $n$ and $\Delta$, but not $G$),
and then reveals the edges of $G$ one at a time. 
Immediately after the arrival of an edge, the algorithm must irrevocably assign a color to it, with the objective of minimizing the final number of colors used.
This problem was first studied nearly thirty years ago, by \citet*{bar1992greedy}. They showed that the greedy algorithm, which returns a proper $(2\Delta-1)$-edge coloring, is worst-case optimal among online algorithms. This might seem to be the end of the story for this line of research. However, as pointed out by \citet{bar1992greedy}, their lower bound only holds for bounded-degree graphs, with some $\Delta=O(\log n)$. This then led them to conjecture that online $(1+o(1))\Delta$-edge-coloring is possible for graphs with  $\Delta=\omega(\log n)$. This conjecture remains wide open.

Recently, an important progress was made towards proving the Bar-Noy et al.~conjecture: \citet{cohen2019tight} showed how to obtain a $(1+o(1))\Delta$-edge coloring for bipartite graphs in an online setting under \emph{node arrivals} (together with their edges). This is a relaxation of the online edge-arrival model. Thus, this latter result can be seen as an intermediate step towards the ultimate goal of proving the Bar-Noy et al.~conjecture. In a similar spirit, we consider edge coloring in two well-studied relaxations of the online model, that act as intermediate steps towards the Bar-Noy et al.~conjecture, and make substantial progress on the state-of-the-art results in both these settings.

\medskip
\noindent\textbf{(I) Random-order online edge coloring:} 
Here, an adversarially-chosen graph has its edges revealed to the algorithm in {\em uniformly random order}. 
Such random-order arrivals, which capture numerous stochastic arrival models, have been widely studied for many online problems. (See, e.g., \cite{karande2011online,mahdian2011online,korula2018online,meyerson2001online,kesselheim2014primal} and the survey by \citet{gupta2020random} and references therein.) In the context of edge coloring, this model was studied by \cite{aggarwal2003switch,bahmani2012online}.
\citet{aggarwal2003switch} were the first to show that high $\Delta$ suffices for near-ideal coloring in this model, giving  a $(1+o(1))\Delta$-edge-coloring algorithm for \emph{multigraphs} with $\Delta=\omega(n^2)$.
\citet{bahmani2012online} then breached the greedy $2\Delta-1$ barrier for simple graphs with polylogarithmic  $\Delta$,
giving a $1.43\Delta$-edge coloring algorithm for $\Delta=\omega(\log n)$ (improved to $1.26\Delta$ in their journal version). This leads to the following natural open question: 
can one obtain ``the best of both worlds'' w.r.t.~\cite{aggarwal2003switch,bahmani2012online}? That is, can one obtain a  $(1+o(1))\Delta$-edge-coloring for graphs of maximum degree $\Delta=\omega(\log n)$ whose edges are presented in random order? 
Put another way, is the Bar-Noy et al.~conjecture true for random-order edge arrivals? We answer this question in the affirmative.

\begin{wrapper}
\begin{restatable}{thm}{thmrandorder}\label{thm:ro}
	For some absolute constant $\gamma \in (0, 1)$, there exists an online algorithm that, when given a graph $G$ of maximum degree $\Delta=\omega(\log n)$, whose edges are presented in random order, computes a proper  $\left(\Delta+O\left(\Delta^{\gamma }\cdot \log^{1-\gamma} n\right)\right) = (1+o(1))\Delta$-edge-coloring of $G$ w.h.p.
\end{restatable}
\end{wrapper}

We complement this upper bound with a lower bound showing that, for some $\Delta = O(\log n)$, not only is it impossible to guarantee a $(1+o(1))\Delta$-edge-coloring under random-order arrivals, but it is even impossible to use any fewer than $2\Delta-1$ colors: see \Cref{sec:ro-lb}. 

We note that previous random-order online edge coloring algorithms  \cite{bahmani2012online,aggarwal2003switch} required the graph to be $\Delta$-regular.  This assumption is without loss of generality in an offline setting, but it is unclear whether the same holds in the random-order online model. 
Our algorithm from \Cref{thm:ro}, however, works on any graph (including non-regular ones): this is discussed in \Cref{ro:warm-up}.

\medskip
\noindent\textbf{(II) Dynamic edge coloring with recourse:} 
Another widely-studied relaxation of online algorithms is online algorithms with \emph{recourse}. Here, an algorithm must make immediate choices upon each arrival, but is also allowed to make a small number of changes to its solution after each arrival (referred to as \emph{recourse}). 
The concept of recourse helps us understand the robustness/sensitivity of (near-)optimal solutions. 
Accordingly,  an influential line of research in the online algorithms community is devoted to studying the tradeoffs between the solution quality and recourse for many well-known problems~\cite{gupta2017online,GuptaKS14,BernsteinHR18,feldkord2018fully,GuGK13-steiner-tree,MegowSVW12,BosekLSZ14,GuptaK14}.

Many results for bounded-recourse online algorithms hold in a more general, \emph{dynamic} setting.
In the context of   edge coloring, the dynamic (oblivious) version of the problem with recourse is captured by the following scenario: The input graph $G$ changes via a sequence of $\poly(n)$ {\em updates} chosen in advance by an adversary, where each update consists of an edge insertion or deletion in $G$. 
The maximum degree in $G$ remains at most $\Delta$ throughout. The algorithm  maintains a proper edge coloring, while changing the colors of some edges in $G$ after each update (the number of such  changes per update is the recourse of the algorithm). The challenge is to design an algorithm that simultaneously (a) maintains a proper coloring with few colors and (b) has small recourse.

In recent years, the edge coloring problem has been extensively studied  from a different, but highly related, perspective of {\em dynamic data structures}~\cite{barenboim2017fully,duan2019dynamic,bhattacharya2018dynamic,wajc2020rounding}. Here, the goal is to maintain a proper edge coloring with few colors in a dynamic graph, taking little time after each edge update (insertion/deletion), where this time is referred to as {\em update time}. 
Note that the update time  of any data structure for dynamic edge coloring upper bounds its recourse, since the data structure has to spend at least $\Omega(1)$ time per edge which changes its color after an update.

The state-of-the-art result for dynamic edge coloring with recourse follows from the work of \citet{duan2019dynamic}.  In any dynamic graph with $\Delta = \Omega(\log^2 n/\eps^2)$, their algorithm maintains a proper $(1+\eps)\Delta$-edge coloring  with $\poly(\log n, 1/\eps)$ recourse. 
Given that other dynamic problems are known to admit super-constant recourse lower bounds  (see, e.g.,~\cite{feldkord2018fully}),
it is natural to ask if one can get a recourse bound for $(1+\eps)\Delta$-edge coloring that is independent of $n$.  We answer this question in the affirmative.

\begin{wrapper}
\begin{restatable}{thm}{thmdynamic}\label{thm:dynamic}
	There is an algorithm that maintains a proper $(1+\eps)\Delta$-edge-coloring w.h.p., with $\poly(1/\eps)$ expected  recourse  in dynamic graphs of maximum degree $\Delta=\Omega\left(\log n/\poly(\eps)\right)$. 
	\end{restatable}
\end{wrapper}

\subsection{Our Techniques}\label{sec:techniques} At the heart of both our results is one common algorithmic approach, inspired by the R\"odl Nibble Method \cite{alon2004probabilistic}, as  applied to distributed edge coloring by \citet{dubhashi1998near}. This method and its variants have since found further uses in distributed settings \cite{chang2018complexity,elkin20142delta}. To the best of our knowledge, we are the first to export this method to online settings.

We analyze our basic algorithm, which is a variant of \cite{dubhashi1998near}, in an offline model. We then show how to implement this algorithm in online and dynamic settings, from which we obtain our results. 
We now outline this basic algorithm, and the ideas needed to implement it in the models we study. For simplicity, we focus on $\Delta$-regular graphs in this section.

\medskip
\noindent {\bf The High-Level Framework:} The Nibble Method in the framework of edge coloring was first used in \cite{dubhashi1998near} in the distributed model. Let us sketch how their algorithm would work in the offline setting. The algorithm consists of multiple rounds. In each round, each vertex $v$ selects a random $\eps$ fraction of its incident uncolored edges. Each sampled edge $e$ chooses a tentative color u.a.r.~among the colors in $[\Delta]$ not yet taken by incident edges (\emph{palette} of $e$).
We then assign the tentative color $c(e)$ to sampled edges $e$ for which no incident edge $e'$ picked the same tentative color $c(e)$, else we mark $e$ as \emph{failed}, and leave $e$ uncolored.
It turns out that each sampled edge fails at each round with probability $O(\epsilon)$. Crucially, picking $\epsilon$ appropriately results in a number of important parameters (degrees in the uncolored subgraph, palette sizes, etc') behaving in a predictable manner, and being sharply concentrated around their mean, w.h.p. In particular, this results in the uncolored subgraph's maximum degree decreasing w.h.p.~at a rate of roughly $1-\epsilon$ per application of this subroutine, or \emph{round}. Consequently, some $t_\eps = O(\log (1/\eps)/\epsilon)$ rounds leave an uncolored subgraph of maximum degree $\Delta'=\poly(\eps)\Delta$ w.h.p., which can then be greedily colored using a further $2\Delta'=\poly(\epsilon)\Delta$ colors. This approach therefore yields a proper $(1+\poly(\epsilon))\Delta$ edge coloring.

In part inspired by \cite{chang2018complexity}, we consider a slight modification of the above algorithm which is more convenient for our goals. In more detail, we make the following changes:
\\
(1) We do not attempt to re-color an edge $e$ which fails in a given round in future rounds, instead leaving $e$ to be colored greedily in the final stage. Intuitively, ignoring these edges still results in a low-degree uncolored graph after $t_\epsilon$ rounds, since few edges incident to each vertex fail. 
\\
(2) Whenever an edge $e$ picks a tentative color $c$, we remove $c$ from the palettes of its incident edges even if $e$ fails. 
Intuitively, this does not decrease the palette sizes much, again, since few edges incident to each vertex fail.
\\
(3) We sample each edge independently with probability $\epsilon$ in each round.

\medskip Our modifications bring two main advantages. 
First, the analysis can be substantially simplified: rather than using a specialized concentration inequality of \citet{grable1998large}, we mostly use 
Hoeffding bounds for negatively-associated variables (see \Cref{sec:prelims}).
This allows us to provide a relatively concise, but complete analysis for sub-constant values of $\eps$ and for non-regular graphs.
Second, and importantly for us, it is easier to adapt the modified algorithm to the online settings that we study.

\medskip
\noindent {\bf Random-Order Online Implementation:}
To obtain our results for random-order arrivals, we first observe that our edge-centric sampling of modification (3) allows us to use the randomness of edge arrivals to ``sample edges for us''. More formally, we implement the independent edge-sampling part of each round by considering an appropriate binomially-distributed prefix of the remaining edges (relying on our knowledge of the number of edges of the $\Delta$-regular graph, $m=\frac{n\Delta}{2}$). This results in each remaining edge of the graph being sampled independently with probability $\epsilon$.

For each round, we have each edge of the round sample a tentative color u.a.r.~from its palette.
In this online setting, however, we cannot always tell when an edge arrives whether it picked the same tentative color as its incident edges of the same round (since some of these arrive \emph{later}). 
We therefore assign the tentative color $c(e)$ to sampled edges $e$ for which no \emph{previous} incident edge $e'$ picked the same tentative color $c(e)$, else we mark $e$ as \emph{failed}.
Modification (2) in the basic algorithm implies that this change still results in a feasible (partial) coloring.
On the other hand, the uncolored subgraph ``after'' the rounds in this algorithm clearly has lower maximum degree than its counterpart in the basic algorithm, and so greedily coloring this subgraph requires fewer colors than the same stage of the basic algorithm.
Finally, modification (1) of our basic algorithm, whereby we do not attempt to re-color a failed edge in ``future rounds'' (which would require knowledge of future arrivals), 
implies that we can greedily color  every failed edge before the next edge arrives.
So, by the analysis of our basic algorithm, we obtain \Cref{thm:ro} for $\Delta$-regular graphs. In \Cref{sec:ro} we build on this approach to obtain our full result, for general graphs.

\medskip
\noindent {\bf Low-Recourse Implementation:} For our low-recourse dynamic implementation, we show how to maintain, after each update, a coloring drawn from the same distribution as that of the basic algorithm applied to the current graph. Modifications (1) and (3) of the basic algorithm imply that deciding the round in which we sample any edge can be done in advance (prior to any arrival), sampling this round number from the appropriately capped geometric distribution.
The more delicate point is dealing with the choice of tentative colors. For example, when an edge is added, its tentative color is removed from the palettes of its incident edges of later rounds. 
Thus, some incident edges no longer have a tentative color which is u.a.r.~drawn from their current palette. 
In \Cref{sec:recourse} we show that a natural approach of correcting these distributions---sampling a new tentative color if the previous one is no longer valid, and switching to a new color if one samples a newly-available color---results in bounded recourse. 
In particular, building on our analysis of the basic algorithm, we show that the expected number of edges of round $i+1$ whose tentative color changes due to the change of a tentative color of an edge in round $j\leq i$ is at most $O(\epsilon)$. Therefore, for any update, the number of tentative color changes is at most $(1+O(\eps))^{t_\eps} = \poly(1/\epsilon)$, from which we obtain \Cref{thm:dynamic}.
In \Cref{sec:recourse} we formalize this approach and its analysis

\subsection{Further Related Work}
Other than \cite{dubhashi1998near}, most closely related to our approach are other distributed edge coloring algorithms using the Nibble Method and its variants and extensions \cite{chang2018complexity,elkin20142delta}.
While these distributed algorithms improve on \cite{dubhashi1998near}, they require crucially that edges be tentatively colored in multiple rounds---a design pattern which seems hard to implement in online settings. 
Another approach is suggested by the work of \citet{cohen2019tight} for bipartite one-sided vertex arrivals; using an online matching algorithm of \citet{cohen2018randomized} which matches each edge with probability $\frac{1-o(1)}{\Delta}$, they color roughly one edge of each maximum-degree node per round, resulting in a $(1+o(1))\Delta$-edge-coloring. 
Unfortunately, for the models we study, no such matching algorithm is known---ruling out this approach.

Returning to previous algorithms in our models, we note that the algorithm of \citet{duan2019dynamic}, which uses an augmenting path based approach, has an inherent polylogarithmic recourse.
On the other hand, the approaches of \cite{aggarwal2003switch,bahmani2012online} for random-order arrivals seem challenging to extend to dynamic recourse-bounded algorithms. Moreover, in the random-order online model, it is unclear how to provably achieve a $(1+o(1))\Delta$-edge-coloring for simple graphs with $\Delta=\omega(\log n)$ by using the ideas in those papers. In this work we show that the Nibble Method, and in particular, a variant of the algorithm of \cite{dubhashi1998near}, allows us to obtain such desired results in \emph{both} these online models.

\section{The Basic Algorithm}
\label{sec:overview:algo}

In this section, we describe our basic algorithm for near-regular graphs in the {\em static setting}. 
and state the key theorem needed for its analysis. We defer
a more detailed analysis to \Cref{sec:agnostic}. Our online and dynamic algorithms in subsequent sections will be built on top of this basic algorithm.

The input to the algorithm is a graph $G = (V, E)$ with $|V| = n$ nodes, where the degree of each node lies in the interval $[(1-\eps^2)\Delta, (1+\eps^2) \Delta]$. The parameter $\eps$ satisfies the following condition:
\begin{equation}
\label{eq:eps}
1/10^4 \geq  \eps \geq 10 \cdot \left( \ln n/\Delta\right)^{1/6}.
\end{equation}
Note that such $\epsilon$ exist if $\Delta =\Omega(\log n)$ is large enough.  The algorithm runs in two {\em phases}, as follows. 

\medskip
\noindent\textbf{Phase One.} In  phase one, the algorithm properly colors a subset of edges of $G$ using $(1+\eps^2)\Delta$ colors, while leaving an uncolored subgraph of small maximum degree. This phase consists of $t_{\eps}-1$  {\em rounds} $\{1, \ldots, t_{\eps}-1\}$, for
\begin{equation}
\label{eq:t}
t_{\eps} := \left\lfloor \ln(1/\eps)/(2K\eps) \right\rfloor, \text{ and } K = 48.
\end{equation} 
Each round $i \in [t_{\eps}-1]$ operates on a subgraph $G_i := (V, E_i)$ of the input graph (with $E_1 = E$), identifies a subset of edges $S_i \subseteq E_i$, picks a {\em tentative}  color $c(e) \in [(1+\eps^2)\Delta] \cup \{ \text{null}\}$ for each edge $e \in S_i$, and  returns the remaining set of edges $E_{i+1} = E_i \setminus S_i$ for the next round. Thus, we have: $E = E_1  \supseteq E_2 \supseteq \cdots  \supseteq E_{t_{\eps}}$.  We now describe how a given round $i \in [t_{\eps}-1]$ works. We start by defining a couple of important notations. 
\begin{itemize}
\item  (a) For all  $v \in V$, let 
$P_i(v) := \{ \chi \in [(1+\eps^2)\Delta] : \chi \neq c(u,v) \text{ for all } (u, v) \in \bigcup_{j < i} S_j\}$ 
denote the {\em palette} of the node $v$ for round $i$.  A color $\chi \in [(1+\eps^2)\Delta]$ belongs to $P_i(v)$ iff no edge incident on $v$ has tentatively picked the color $\chi$ in previous rounds $j < i$.  (b) Similarly, for all  $(u, v) \in E_i$, let $P_i(u, v) := P_i(u) \cap P_i(v)$ denote the {\em palette} of the edge $(u,v)$ for round $i$. 
\end{itemize}
In round $i$, we first sample  each edge $e \in E_i$ {\em independently} with probability $\eps$. Let $S_i \subseteq E_i$ be the set of sampled edges. Next, each  edge $e \in S_i$ with $P_i(e) \neq \emptyset$  tentatively picks  a color $c(e)$ from its palette $P_i(e)$ uniformly and independently at random.  We say that an edge $e \in S_i$ {\em failed} in round $i$ iff  either  (a) $P_i(e) = \emptyset$ (in this case we set $c(e) := \text{null}$), or (b) among the  edges $N(e) \subseteq E$ that are adjacent to $e$, there is
some edge $e' \in S_i$ that tentatively picked the same color  (i.e., $c(e) = c(e')$). Let $F_i \subseteq S_i$ denote the set of failed edges in round $i$. The remaining sampled edges $e \in S_i \setminus F_i$ are called {\em successful} in round $i$. Each such edge $e \in S_i \setminus F_i$ is {\em assigned} the color $c(e)$ it tentatively picked in round $i$. Before terminating the current round, we set $E_{i+1} := E_i \setminus S_i$ and $G_{i+1} := (V, E_{i+1})$. We remark that the color tentatively sampled by a failed edge $e$ cannot be used by the edges incident to $e$ in subsequent rounds. This will prove useful both for our analysis and when implementing this algorithm in other models in subsequent sections.

\medskip
\noindent\textbf{Phase Two.} Finally,  in phase two, we greedily color all edges that were not successful in  phase one.  That is, letting $G_F := (V, \cup_i F_i)$ be the subgraph  consisting of all the edges that failed in phase one, and $G_{t_{\eps}} := (V, E_{t_{\eps}})$ be the subgraph consisting of all the edges that were never sampled in phase one, we color the edges of $G_{t_{\eps}} \cup G_F$ greedily, using a new palette of $2\Delta(G_{t_{\eps}} \cup G_F)-1$ colors. 
Here $\Delta(H)$ denotes the maximum degree in any graph $H$.  

 \begin{algorithm}[h]
	\caption{The Basic Algorithm}
	\label[algorithm]{alg:model-agnostic}
	\begin{algorithmic}[1]
		\State $E_1 \gets E$ and $G_1 \gets (V, E_1)$
		\For{$i=1,2,\dots,t_{\eps}-1$}
		\State $S_i\gets \emptyset$
		\For{each $e\in E_i$ independently}
		\State With probability $\eps$, add $e$ to $S_i$  \label{sample-alg}
		\State $\palette_i(e) \gets [(1+\eps^2)\Delta]\setminus \{c(e') \mid e'\in N(e)\cap \bigcup_{j<i} S_j\}$. 
		\State If $P_i(e)\neq\emptyset$, sample $c(e)\sim_R  \palette_i(e)$, else set $c(e)\gets \textrm{null}$ 
		\Comment{Tentative coloring of $e$}
		\EndFor
		\State Let $F_i \gets \{e\in S_i \mid c(e)\in \{\textrm{null}\}\cup \{c(e') \mid e'\in N(e) \cap S_i \}\}$ \Comment{The set of failed edges} \label{failed-edges}
		\State Color each edge $e\in S_i\setminus F_i$ using color $c(e)$ 
		\State $E_{i+1} \gets E_i\setminus S_i$ and $G_{i+1} \gets (V, E_{i+1})$.
		\EndFor
		\State  Let $G_F := \left(V, \bigcup_i F_i \right)$ denote the subgraph of $G$ consisting of the failed edges from phase one. 
		\State Color $G_{t_\epsilon} \cup G_F$ greedily using colors $(1+\eps^2)\Delta+1,(1+\eps^2)\Delta+2,\dots,(1+\eps^2)\Delta+2\Delta(G_{t_\epsilon} \cup G_F)-1$.
	    \label{line:greedy}
	\end{algorithmic}
\end{algorithm}

The algorithm's pseudocode is given in \Cref{alg:model-agnostic}. We now turn to discussing its analysis.

\begin{obs}
\label{main:obs:correct}
\Cref{alg:model-agnostic} outputs a proper $((1+\eps^2)\Delta+ 2\Delta(G_{t_{\eps}} \cup G_F)-1)$-edge-coloring of the input graph $G = (V, E)$. 
\end{obs}
\begin{proof}
First observe that the algorithm computes a valid partial coloring in Phase One. Indeed, any $e \in S_i \setminus F_i$ selects a color $c(e) \in P_i(e) \subseteq [(1+\eps^2)\Delta]$, and the definition of $P_i(e)$ and $F_i$ guarantees that no other edge $e'\in N(e)$ in any round of Phase One can be colored with $c(e)$. The claim follows by observing that in Phase One we use only colors from $ [(1+\eps^2)\Delta]$, while in Phase Two the greedy algorithm uses a disjoint set of at most $2\Delta(G_{t_{\eps}} \cup G_F)-1$ extra colors.
\end{proof}

The key property of the basic algorithm is captured in the following theorem.
 
\begin{restatable}{thm}{uncoloreddeg}
\label{main:th:main}
$\Delta(G_{t_{\eps}} \cup G_F) = O\left(\eps^{1/(3K)}   \Delta\right)$ w.h.p.
\end{restatable}

\begin{cor}
\label{main:cor:th:main}
The basic algorithm  $\left(\Delta+O\left(\eps^{1/(3K)} \Delta \right)\right)$ edge colors $G$, w.h.p.
\end{cor}

\begin{proof}
Follows from~\Cref{main:th:main} and~\Cref{main:obs:correct}.
\end{proof}

In some sense, the arguments behind the proof of \Cref{main:th:main}  were already apparent in the work of \citet{dubhashi1998near}.
Consequently, we defer a complete and self-contained proof of this theorem to \Cref{sec:agnostic}.
For now, we turn to exploring implications of this theorem and \Cref{alg:model-agnostic} to online edge coloring.
\section{Random-Order Online Algorithm}\label{sec:ro}

In this section we present algorithms which (essentially) implement \Cref{alg:model-agnostic} in the random-order online model. We start with a warm-up case, where the input  graph is near-regular, and we know the value of $m$ (the number of edges in the final graph).

\subsection{Warm-up: Near-Regular Graphs with Known $m$}\label{ro:warm-up}

One subroutine we rely on  is the ability to use the stream's randomness to simulate independent sampling of edges. For completeness, we provide a proof of the following simple fact in \Cref{sec:ro-appendix}.

\begin{restatable}{fact}{rosamplinglem}\label{ro-sampling}
	Consider a universe $U$ of $n$ elements, and let $p\in [0,1]$. 
	Let $U_k\subseteq U$ denote the first $k$ elements in a random-order stream of $U$, and let $X \sim Bin(n, p)$ be a binomial random variable with parameters $n$ and $p$. Then the random set $U_X$ contains every element in $U$ independently with probability $p$.
\end{restatable}

Using \Cref{ro-sampling}, we simulate (a variant of) \Cref{alg:model-agnostic} with parameter $\epsilon$ under random-order edge arrivals in a graph $G = (V, E)$ with  $m$ edges and $n$ nodes, where the degree of each node lies in the interval $(1\pm \eps^2)\Delta$, and $\Delta=\omega(\log n)$. The algorithm knows  $n, \Delta$ and $m$ (but not $G$).

\smallskip
\noindent
\textbf{Warm-up Algorithm:}
Set $\eps:= 10 \cdot (\ln n/\Delta)^{1/6}$ (see~(\ref{eq:eps})). For round $i=1,\dots,t_{\epsilon}-1$, sample an independent random variable $X_i\sim Bin(m-\sum_{j<i} X_j, \eps)$, 
and let $S_i$ be the set of edges in $G$ whose positions in the random-order stream lie in the interval $(\sum_{j<i} X_j,\sum_{j\leq i} X_j]$.
As with \Cref{alg:model-agnostic}, each edge $e\in S_i$, upon its arrival, samples a tentative color 
$$c(e)  \sim_R  \palette_i(e) := [(1+\eps^2)\Delta]\setminus \{c(e') \mid e'\in N(e)\cap S_j,\, j<i\},$$
where we set $c(e)\gets null$ if $P_i(e)=\emptyset$.
Unlike in \Cref{alg:model-agnostic}, in this online setting the algorithm cannot know whether the color $c(e)$ conflicts with  the tentative color of a neighboring edges $e' \in N(e) \cap S_i$ that arrives in the same round $i$, but {\em after} $e$ in the stream. Hence, we color each edge $e\in S_i$ with its tentative color $c(e)$, unless $c(e)=null$ or some {\em previously-arrived} neighboring edge $e' \in N(e) \cap S_i$ also  picked  color $c(e')=c(e)$.
In the latter case, we instead color $e$ greedily with the first available color $j>(1+\eps^2)\Delta$. We let $F'_i$ be the edges in $S_i$ which are colored greedily.

As we show, this online algorithm inherits the performance of the basic \Cref{alg:model-agnostic}.

\begin{thm}\label{thm:ro-regular}
For some absolute constant $\gamma \in (0, 1)$, the warm-up algorithm described above  yields a proper $\left(\Delta+O\left(\Delta^{\gamma }\cdot \log^{1-\gamma} n\right)\right) = (1+o(1))\Delta$-edge coloring of $G$ w.h.p.
\end{thm}
\begin{proof}
	This algorithm outputs a valid edge coloring, as it colors every edge (due to the greedy stage) and never assigns an edge a color used by an incident edge. It remains to bound its performance.

	For any $i\geq 0$, Let $E_i$ be the set of edges whose positions in the random-order stream lie in the interval $(\sum_{j<i} X_j, m]$.
	By \Cref{ro-sampling}, the set of edges  $S_i$ is  a random subset of $E_i$ which contains each edge in $E_i$ independently with probability $\epsilon$.
	A simple induction on $i$ shows that the sets $S_i$ and $E_i$ share the same distributions as their counterparts in \Cref{alg:model-agnostic}. Next, denote by $F_i\supseteq F'_i$ the set of edges $e \in S_i$ for which $c(e)\in \{null\}\cup \{c(e')\mid e'\in N(e) \cap S_i\}$. Since each edge $e \in S_i$ picks a color uniformly at random from the set of colors not picked by any of its neighboring edges in previous rounds (including the edges in $F_j$ for all $j<i$), 
	a simple induction on $i$ shows that the random variables $F_i$ and $c(e)$ in this algorithm are distributed exactly as their counterparts in  \Cref{alg:model-agnostic}.
	Consequently, the upper bounds on $\Delta(\bigcup_i F_i)\geq \Delta(\bigcup_i F'_i)$ and $\Delta(G_{t_\eps})$ of \Cref{alg:model-agnostic} hold for this online algorithm as well.
	Therefore, the greedy (online) algorithm colors the uncolored edges in $G_{t_\eps}\cup G_F$ using at most $2 \cdot \Delta(G_{t_\epsilon} \cup G_F) -1= O(\epsilon^{1/(3K)} \Delta)$ colors w.h.p., by~\Cref{main:th:main} and our choice of $\eps =  10 \cdot (\ln n/\Delta)^{1/6}$, as in \eqref{eq:eps}. 
	As we use $(1+\eps^2)\Delta$ distinct colors for all other edges, this online algorithm uses $\Delta + O\left(\epsilon^{1/(3K)} \Delta\right)$ colors overall w.h.p. Since $\Delta = \omega(\log n)$ and $K$ is an absolute constant (see~(\ref{eq:t})), the theorem follows from our choice of $\eps$.
\end{proof}

\textbf{Assuming near-regularity, and known $\mathbf{m}$.} The assumption of near-regularity used by the above algorithm is common in the literature. Indeed, all prior random-order online edge-coloring algorithms assume perfect regularity \cite{aggarwal2003switch,bahmani2012online}. As pointed out in those papers, this assumption is without loss of generality in the offline model, where we can add dummy edges to make the graph regular. In a random-order online setting, this is problematic, however, as these dummy edges should be interspersed among the real edges to create a regular graph \emph{presented in random order}. This last point seems impossible without prior knowledge of vertices' final degrees, and the number of edges, $m$, which we assume prior knowledge of. In the next section we show how to remove the assumption of near-regularity, as well as knowledge of $m$, while retaining the asymptotic performance of  \Cref{thm:ro-regular}.

\subsection{General Graphs}
\label{sec:non-regular}

We now present and analyze our random-order online edge coloring algorithm for general graphs $G = (V, E)$ with $n$ nodes, $m$ edges and maximum degree $\Delta = \omega(\log n)$. In particular, we do not assume that all nodes have degree close to $\Delta$. The algorithm knows  $n, \Delta$; but does {\em not} know  $m$ nor $G$. Let $e_1,\ldots,e_m$ be the random stream of edges, $G^{(k)}$ be the subgraph induced by $e_1,\ldots,e_k$, and $d^{(k)}(v)$ be the degree of node $v$ in $G^{(k)}$. Our key insight is to observe the first few edges in the input stream until some node reaches degree $\eps \Delta$. This is sufficient to infer (approximately) the value of $m$ and the degree of each node in $G$. At the same time, we can afford to color such initial set of edges greedily.

In more detail, our algorithm consists of $3$ main steps. In Step (I), we observe the first $T$ edges until some node $v$ reaches the degree $d^{(T)}(v)=\eps \Delta$ (or we reach the end of the stream). This first set of edges is colored greedily using the first available color. Let $\Delta_1$ be the largest color used in Step (I). The following technical lemma follows from a standard application of Chernoff bounds over sums of negatively-associated variables (proof in \Cref{sec:ro-appendix}).
\begin{restatable}{lem}{estimateslem}\label{estimate-using-eps-fraction}
	Let $\epsilon\leq \frac{1}{2}$, and let $\alpha>0$ be a constant, and assume $\Delta\geq \frac{24(\alpha+3) \ln n}{\epsilon^8}$. 
	Then, with probability at least $1-O(n^{-\alpha})$, the following properties hold:\footnote{We let $c=a\pm b$ denote $c\in [a-b,a+b]$.}
	\begin{enumerate}\itemsep0pt
	\item $T=\epsilon\cdot m (1\pm \epsilon^2)$. \label{m-estimate}
	\item $d^{(T)}(v) = \epsilon\cdot d(v) \pm 2\epsilon^3 \Delta$ for every node $v$. \label{degree-estimates}
	\item Let $m' :=T/(\eps(1+\eps^2))$. Conditioned on $m'\leq m$, every node $v$ has  $d(v) - d^{(m')}(v)\leq 2\epsilon^2 \Delta$. \label{last-edges}
	\end{enumerate} 
\end{restatable}
Henceforth, we assume that all the high-probability events in Lemma \ref{estimate-using-eps-fraction} actually occur (otherwise the algorithm fails). In Step (II),  we color the next $m'-T$ edges $R:=\{e_{T+1},\ldots,e_{m'}\}$ using colors larger than $\Delta_1$, as described below.

Let $G_R = (V, R)$ denote the subgraph of $G$ induced by the edges in $R$. Before processing the $(T+1)^{th}$ update $e_{T+1}$, we virtually expand $G_R$ by adding {\em dummy nodes} $W$ and {\em dummy edges} $D$ in the following manner. 
For each node $v \in V$, create $\Delta$ dummy nodes $v_1,v_2,\dots,v_\Delta$ which form a $\Delta$-clique via dummy edges, and add extra dummy edges from $v$ to $\max\{0,\Delta-(1/\eps-1)\cdot d^{(T)}(v)\}$ of these dummy nodes  $\{v_1, \ldots, v_{\Delta}\}$. Let $H$ be the resulting graph. Note that at this point we only know the dummy edges in $H$, as the  edges  in $G_R$ will arrive in future.

Let ${\cal A}$ denote the warm-up online algorithm from \Cref{ro:warm-up}. In Step (II), we run this   online algorithm ${\cal A}$   with parameter $2\eps$ on $H$, where the edges of $H$ are presented to ${\cal A}$ in \emph{random order}. More precisely, initializing $j=T+1$, $D' = D$ and $R' = R$, we perform the following operations for $|R|+|D|$ iterations.

\begin{itemize} 
\item With probability $|D'|/(|R'|+|D'|)$, we sample a random edge $e_d$ from $D'$, and feed the edge $e_d$ to the online algorithm ${\cal A}$. We then set  $D' = D' \setminus \{e_d\}$ before going to the next iteration.
\item With remaining probability, we feed the edge $e_j$ to ${\cal A}$ and color $e_j$ with the color $\chi(e_j)+\Delta_1$, where $\chi(e)$ is the color chosen by ${\cal A}$ for $e_j$. Then we let $R' = R' \setminus \{e_j\}$ and increase $j$ by one.
\end{itemize}

\noindent Let $\Delta_2$ be the largest color chosen in Step (II).

Finally, in Step (III), we color the remaining edges $e_{m'+1},\ldots,e_{m}$ greedily with the first available color $j>\Delta_2$. Let $\Delta_3$ be the largest color used at the end of Step (III).

We next analyze the above algorithm (assuming  the occurrence of the high probability events from \Cref{estimate-using-eps-fraction}). Obviously this algorithm computes a feasible coloring. 
By definition, Step (I) uses $\Delta_1 \leq 2\epsilon\Delta$ colors. 
Analogously, by \Cref{last-edges} of \Cref{estimate-using-eps-fraction}, the number of colors used in  Step (III) is at most $\Delta_3-\Delta_2=O(\eps^2 \Delta)$. 
It remains to upper bound the number of colors $\Delta_2-\Delta_1$ used in the second step.
To this end, we note that 
\Cref{estimate-using-eps-fraction} implies that $H$ is near-regular. More precisely, we have the following bound, whose proof is deferred to \Cref{sec:ro-appendix}.
\begin{restatable}{lem}{dummysubgraphlem}\label{dummy-subgraph}
	The graph $H$ satisfies $d_H(v) = \Delta(1\pm 4\eps^2)$ for all $v\in V(H)$, w.h.p.
\end{restatable}

It is easy to see that  Step (II) implements the warm-up algorithm on $H$, as the edges of $H$ are fed to this algorithm in a uniform random order.
Thus, by \Cref{thm:ro-regular}, w.h.p.~the number of colors used in  Step (II) is at most $\Delta_2-\Delta_1\leq \Delta+O\left(\Delta^{\gamma }\cdot \log^{1-\gamma} n\right)$ for a constant $\gamma\in (0,1)$. By choosing $\eps$ small enough so that $\eps \Delta\leq\Delta^{\gamma }\cdot \log^{1-\gamma} n$, we immediately obtain our main result.

\thmrandorder*

\section{Low-Recourse Dynamic Algorithm}\label{sec:recourse}
In this section, we give an implementation of \Cref{alg:model-agnostic} in a dynamic setting with low recourse. Before describing our algorithm, we explain why in the dynamic setting we can focus our attention on near-regular graphs, as required of inputs to \Cref{alg:model-agnostic}.

\smallskip
\noindent
\textbf{(Near-)Regularizing Gadget.} Consider a dynamic input graph $G = (V, E)$ on $n$ nodes, where the degree of each node remains at most $\Delta$ all the time.
In  this dynamic setting, we describe a procedure to  maintain a super-graph $G'=(V',E')$ of  $G=(V,E)$, with $V' \supseteq V$ and $E' \supseteq E$, such that: (a) each node in $G'$ always has degree either $\Delta$ or $\Delta-1$, and (b) each update (edge insertion/deletion) in $G$ results in a constant number of updates in $G'$. The node-set $V'$ of $G'$ consists of the nodes $v \in V$ of the input graph, plus $\Delta$ {\em dummy nodes} $v_1,v_2,\dots,v_{\Delta}$ for each $v\in V$. 

Initially, when $G$ is empty,\footnote{It is easy to extend our gadget to the setting where the input graph $G$ is not empty at preprocessing.} the edge-set of $G'$ is defined as follows. For each node $v \in V$ in the input graph, the set of nodes $\{v,v_1,\dots,v_{\Delta}\}$ induces a $(\Delta+1)$-clique of {\em dummy edges} in the supergraph $G'$.  In other words, there is a {\em dummy edge} $(v, v_i)$ for all $v \in V, i \in [\Delta]$; and a {\em dummy edge} $(v_i, v_j)$ for all $v \in V, i \in [\Delta], j \in [\Delta], i \neq j$. At this point in time, the edge-set $E'$ of $G'$ consists only of the dummy edges.

Subsequently,  whenever  an edge $e=(u,v)$ in added to (resp., removed from)  $G$ during an update, we perform the following steps. We first  add $e$ to (resp., remove $e$ from) $G'$. Next, let   $i$ and $j$ be the smallest indices in $[\Delta]$ for which the dummy edges $(u,u_i)$ and $(v,v_j)$ currently exist (resp., do not exist) in $G'$: we remove (resp., add) these two dummy edges $(u, u_i)$ and $(v, v_j)$. 

Note that every update in $G$ generated by an oblivious adversary results in three updates (likewise generated by an oblivious adversary) in $G'$.  Since  any $k$-edge-coloring of $G'$ trivially induces a $k$-edge-coloring of $G$, our goal will   be to maintain a $(1+\eps)\Delta$-edge-coloring of $G'$.

{\em Accordingly, w.l.o.g., henceforth we will  assume that  the  input graph $G$ given to us at preprocessing is near-regular (specifically, each of its nodes has degree either $\Delta$ or $\Delta-1$), and that the input graph $G$ remains near-regular  throughout the sequence of updates.}

\smallskip
\noindent {\bf Our dynamic algorithm:} At each time $t$ (i.e., after the $t$-th update), our dynamic algorithm strives to assign tentative colors $c^{(t)}(e)$ to all edges in the current input graph, denoted by $G^{(t)} = (V,  E^{(t)})$, as in \Cref{alg:model-agnostic}. It likewise defines failures in the same way as in \Cref{alg:model-agnostic}. Finally, it  maintains a  coloring of the failed and otherwise uncolored edges using a simple $O(\Delta)$-edge-coloring, constant-recourse {\em dynamic} algorithm, which we denote by \textsc{SimpleColor} (e.g., \cite{bhattacharya2018dynamic,wajc2020rounding}).

For the rest of this section,  the superscript $(t)$ on any given notation will indicate the state of the concerned object after the $t^{th}$ update in the input graph.\footnote{The reader should not confuse time $t$ in the dynamic setting with the last round $t_{\eps}$ in  phase one of \Cref{alg:model-agnostic}.}  Intuitively, the outcome of our dynamic algorithm just after the $t^{th}$ update will be the same as the outcome we would get if we ran \Cref{alg:model-agnostic} in the static setting with  $G^{(t)}$ as  input (see \Cref{main:recourse:analysis:100}).  We  will refer to an {\em unordered} pair of nodes $(u, v)$, with $u, v \in V$, as a {\em potential edge}. Thus, there are $|V| \choose 2$ many potential edges, and the set of potential edges do not change during the sequence of updates.  In contrast, we will refer to an edge $e \in E^{(t)}$ as a {\em current edge} in the input graph after the $t^{th}$ update.  

\smallskip
\noindent \textbf{Preprocessing:} We start by assigning a {\em round} $i(u,v) \in [t_{\eps}]$ to each  potential edge $(u, v)$, where these $i(u,v)$ values are i.i.d.~samples from the capped geometric distribution $CappedGeo(\epsilon,t_\epsilon)$ with success probability $\eps$ and at most $t_\eps$ attempts.\footnote{For a random variable $X\sim CappedGeo(\epsilon,t_\epsilon)$, we have  $\Pr[X=k] =\begin{cases}
	\epsilon \cdot (1-\eps)^{k-1} & k \in \{1, \ldots, (t_\eps - 1)\} \\
	(1-\eps)^{t_{\eps}-1} & k=t_\epsilon.
\end{cases}$.} 
For each $j \in [t_{\eps}]$, we  let $S_j := \{(u,v) \mid i(u,v)=j\}$ denote the set of all {\em potential} edges  that are assigned to round $j$. Furthermore, for each  $t \geq 0$,  let $S_j^{(t)} := S_j \cap E^{(t)}$ denote the set of {\em current} edges $(u, v)$ after the $t^{th}$ update which have $i(u, v) = j$. 


In future, throughout the sequence of updates, the random variables $\{i(u,v)\}$ will determine which sets of current edges  get sampled in which round (see \Cref{sample-alg} of \Cref{alg:model-agnostic}).  Specifically, consider any potential edge $(u, v)$, with $i(u,v) = j$, and any two nonnegative integers $t \neq t'$ such that $(u, v) \in E^{(t)} \cap E^{(t')}$. Then the edge $(u, v)$ will get sampled in the same round $j$ in both $G^{(t)}$ and $G^{(t')}$. Thus, a given edge gets assigned to the same round across all the updates.

After drawing the random variables $\{i(u,v)\}$ for all potential edges as described above, we  implement \Cref{alg:model-agnostic} on the input  $G^{(0)} = (V, E^{(0)})$ given to us at the preprocessing phase.

\smallskip
\noindent {\bf Handling an update:} For any $t \geq 1$, consider the $t^{th}$ update which changes the input graph $G$ from $G^{(t-1)} = (V, E^{(t-1)})$ to $G^{(t)} = (V, E^{(t)})$.  Our dynamic algorithm handles this update by computing an edge-coloring for $G^{(t)}$ in three steps, as described below. 

\smallskip
\noindent {\bf Step I:} We perform the following operations in increasing order of $i=1,2,\dots,t_\epsilon-1$:
\begin{itemize}
\item For every {\em potential} edge $e\in S_i$, we first define its palette
$$\palette_{i}^{(t)}(e) := [\Delta(1+\epsilon^2)]\setminus \left\{c^{(t)}(e')\,\, \Bigg\vert\,\, e'\in N^{(t)}(e)\cap \bigcup_{j<i}S_j^{(t)}\} \right\}.$$
Next, we  call \Cref{alg:recourse-change-color} to update the  tentative color $c^{(t)}(e)$ of  $e$. Note that as in \Cref{alg:model-agnostic}, if $\palette_{i}^{(t)}(e) = \emptyset$, then \Cref{alg:recourse-change-color} sets $c^{(t)}(e)\gets null$. 
\end{itemize}

\smallskip
\noindent {\bf Step II:} For every round $i \in [t_{\eps}-1]$, we now define the set of {\em failed} current edges 
$$F_i^{(t)} := \left\{ e \in S^{(t)}_i : c^{(t)}(e) \in \{ \text{null} \} \cup \{ c^{(t)}(e') : e' \in N^{(t)}(e) \cap S_i^{(t)} \} \right\}.$$
Let $F^{(t)} := \bigcup_{i=1}^{t_{\eps}-1} F_i^{(t)}$ denote the set of all  failed current edges across all rounds. Every current edge $e \in \left(\bigcup_{i=1}^{t_{\eps}-1} S_i^{(t)}\right) \setminus F^{(t)}$ gets colored with its tentative color $c^{(t)}(e)$. 

\smallskip
\noindent {\bf Step III:}
Finally, let $G^{(t)}_U:=G[F^{(t)}\cup S^{(t)}_{t_\eps}]$ be the subgraph of $G^{(t)}$ consisting of all the  edges that are not colored using their tentative colors in Step II above. We use \textsc{SimpleColor} to maintain an $O(\Delta(G^{(t)}_U))$ edge coloring of $G^{(t)}_U$.
In more detail, after each update, denoting by $A\bigoplus B := (A\setminus B)\cup (B\setminus A)$ the symmetric difference, we think of the graph $G^{(t)}_U:=G[F^{(t)}\cup S^{(t)}_{t_\eps}]$ as having undergone $|F^{(t)}\bigoplus F^{(t-1)}| + |S^{(t)}_{t_\epsilon}\bigoplus S^{(t-1)}_{t_\epsilon}|$ updates, which we feed to algorithm \textsc{SimpleColor}. 

\smallskip

This concludes the description of our dynamic algorithm.

\begin{algorithm}[H]
\caption{\textsc{TentativelyColor}($e$)}
\label[algorithm]{alg:recourse-change-color}
\begin{algorithmic}[1]
        \If{$c^{(t-1)}(e) \in \palette^{(t-1)}_i(e)\setminus \palette^{(t)}_i(e)$} 
	\State If $P^{(t)}_i\neq \emptyset$ sample $c^{(t)}(e) \sim_R \palette^{(t)}_i(e)$, else set $c^{(t)}(e)\gets null$  
	\Else 
	\label{line:recourse-give-new-colors-chance-start}
	\State If $P^{(t)}_i\neq \emptyset$ sample $c \sim_R \palette^{(t)}_i(e)$, else set $c\gets null$
	\If{$c=null$ \textrm{ or } $c\in \palette^{(t)}_i(e)\setminus \palette^{(t-1)}_i(e)$ }\label{line:recourse-give-new-colors-chance-test}
	\State $c^{(t)}(e)\gets c$ \label{line:recourse-give-new-colors-chance-end}
	\Else
	\State $c^{(t)}(e) \gets c^{(t-1)}(e)$  \Comment{We keep the previous color}
	\EndIf
	\EndIf
\end{algorithmic}
\end{algorithm}

\smallskip
\noindent {\bf Analysis:} Looking back at \Cref{alg:recourse-change-color}, a moment's thought reveals that if the tentative color $c^{(t-1)}(e)$  was chosen uniformly at random from the set $P_i^{(t-1)}(e)$, then the tentative color $c^{(t)}(e)$ is also chosen uniformly at random from the set $P_i^{(t)}(e)$. This, however, is far from being sufficient for our purpose. In particular, we need  a formal  (and much stronger) guarantee stated below.

\begin{restatable}{lem}{recourseimplement}
\label{main:recourse:analysis:100}
For each time $t\geq 0$, the joint distribution $\{c^{(t)}(e)\}_{e\in E^{(t)}}$ of colors sampled by the dynamic algorithm
is distributed identically to $\{c(e)\}_e$ of \Cref{alg:model-agnostic} when applied to graph $G^{(t)}$.
\end{restatable}

We defer the proof of \Cref{main:recourse:analysis:100} to \Cref{app:recourse:1}. This lemma, together  with our analysis of  \Cref{alg:model-agnostic}, immediately leads us to the following corollary.

\begin{cor}\label{recourse:num-colors}
For $\eps$ as in \eqref{eq:eps} and $K$ as in \eqref{eq:t}, the above dynamic algorithm $\Delta(1+O(\epsilon^{(1/3K)}))$-edge-colors $G$ at any time $t$, w.h.p.
\end{cor}

\smallskip
\noindent {\bf Bounding Recourse:} We now fix some $t \geq 1$, and  bound the expected recourse  our algorithm has to pay while handling  the $t^{th}$ update. Say that an edge $e \in E^{(t)} \cap E^{(t-1)}$ is  {\em dirty} iff $i(e) \leq t_{\eps}-1$ and it changes its tentative color $c(e)$ due to the $t^{th}$ update. Let $D_j$ denote the set of dirty edges $e$ assigned to round $i(e) = j$. We will use the symbol $D_{< i} = \cup_{j < i} D_j$ to denote all the dirty edges at rounds $j < i$. Let $D = D_{< t_{\eps}}$ denote the set of  dirty edges across all rounds.

\begin{lem}
\label{main:lm:recourse:1}
The recourse of the dynamic algorithm to handle the $t^{th}$ update is $O(1+ |D|)$.
\end{lem}

\begin{proof}
Let $e^*$ denote the edge being inserted/deleted during the $t^{th}$ update. The additive $+1$ term in the claimed recourse bound of $O(1+|D|)$ comes from the edge $e^*$. To simplify notations, we will assume $c^{(t-1)}(e^*) = c^{(t)}(e^*)$ for the rest of this proof. Any other edge $e \in E^{(t-1)} \cap E^{(t)}$  changes its final color during the $t^{th}$ update  only if: $i(e) < t_{\eps}$,  and  \{either (1) the edge $e$ changes its tentative color, or (2) the edge $e$ switches from being successful to failed (or vice versa) without changing its tentative color\}. In  case (1), we clearly have $e \in D$.  In  case (2), the edge $e$ must   have at least one  neighboring edge $e_d \in (D \cup \{e^*\}) \cap N(e)$ such that  $c^{(t-1)}(e) = c^{(t)}(e) \in \{ c^{(t-1)}(e_d), c^{(t)}(e_d)\}$. We {\em charge} the $1$ unit of recourse the algorithm has to pay for $e$ to any one such edge $e_d$. A moment's thought reveals that each edge $(u, v) \in D \cup \{e^*\}$ can receive at most $4$ units of charge in this scheme, one for each ordered pair  $\{u, v\} \times \{ c^{(t-1)}(u,v), c^{(t)}(u,v)\}$. Hence, the algorithm's recourse  is at most $O(1+|D|)$. 
\end{proof}

Below, we  state the key lemma that contains the technical meat of our argument. Since the argument requires some nontrivial properties of \Cref{alg:model-agnostic} which are proved only in \Cref{sec:agnostic}, we defer the proof of \Cref{main:lm:recourse:2} to \Cref{app:recourse:2}. 

\begin{lem}
\label{main:lm:recourse:2}
For every round $i \in [t_{\eps}-1]$, we have: $\E[|D_i|] \leq 6\eps + 6 \eps \cdot \E[|D_{< i}|]$.
\end{lem}

\begin{cor}
\label{main:cor:recourse:recourse}
Our dynamic algorithm has an expected recourse of $O(1/\eps^{(3/K)})$ per update. 
\end{cor}

\begin{proof}
Since $|D_{<1}| = 0$ and $D = D_{< t_{\eps}}$, from \Cref{main:lm:recourse:2} we derive that: 
\begin{equation}
\E[|D|]   \leq   \sum_{i=0}^{t_{\eps}-1} 6\eps \cdot (1+6\eps)^{i} \leq (1+6\eps)^{t_{\eps}} \leq \exp(6\eps \cdot t_{\eps}) \leq \exp((3/K) \cdot \ln (1/\eps))  = 1/\eps^{(3/K)}.
\end{equation}
The last inequality holds because of~(\ref{eq:t}). The corollary now follows from \Cref{main:lm:recourse:2}.
\end{proof}

From~\eqref{eq:eps},~\eqref{eq:t},~\Cref{recourse:num-colors} and~\Cref{main:cor:recourse:recourse}, we obtain the following theorem.

\begin{thm}
\label{main:th:recourse:100}
There exists two absolute constants $\gamma, \gamma' \in (0,1)$ such that for all $\eps \in \left[ \left(\log n/\Delta\right)^{\gamma'},  \gamma\right]$ we can maintain, w.h.p., a proper $(1+\eps)\Delta$-edge-coloring in a dynamic graph $G$ of maximum degree $\Delta = \omega(\log n)$, with $O(\poly(1/\eps))$ expected recourse per update.
\end{thm}
\section{Conclusions and Open Questions}
We presented one common approach for tackling edge coloring in two widely-studied relaxations of the online model of computation, making progress on (and in one case, resolving) the conjecture of \citet{bar1992greedy} for these models.
We conclude with a few interesting research directions.

\smallskip
\noindent\textbf{Adversarial Online Arrivals.} The most natural question is whether the Bar-Noy et al.~conjecture holds in the strictest, adversarial edge-arrival model. This question still seems out of reach. One algorithmic approach which suggests itself is to extend the ideas in \cite{cohen2019tight}. This would require some form of online dependent rounding for fractional matching under edge arrivals, generalizing \cite{cohen2018randomized}. Alternatively, it is not implausible that the Bar-Noy et al.~conjecture is false under adversarial edge arrivals, despite being true for vertex arrivals \cite{cohen2019tight}. Such a refutation of this conjecture would mirror a similar separation between these arrival models recently proven for online matching \cite{gamlath2019online}.

\smallskip
\noindent\textbf{Knowledge of $\Delta$.} All our algorithms assume knowledge of the maximum degree $\Delta$. This assumption is common to all prior best algorithms in the models we study \cite{bahmani2012online,aggarwal2003switch,duan2019dynamic}.\footnote{\citet{duan2019dynamic} run a logarithmic number of algorithms for unknown $\Delta$, and switch between their colorings whenever $\Delta$ changes. This results in {\em fast update time}, but high recourse.}
In fact, \citet{cohen2019tight} showed that not knowing $\Delta$ in their online model results in a strictly harder problem, for which no better than $\frac{e}{e-1}\Delta$-edge-coloring algorithm exists, for any (unknown) $\Delta$.
Is the same separation between known and unknown $\Delta$ true for the models studied in this paper?

\medskip
\noindent {\bf Acknowledgements.}
We thank Janardhan Kulkarni for many helpful discussions. 
The work of Sayan Bahttacharya was supported by Engineering and Physical Sciences Research Council, UK (EPSRC)  Grant  EP/S03353X/1. The work of Fabrizio Grandoni was supported in part by the SNF Excellence Grant 200020B\_182865/1.
The work of David Wajc was supported in part by NSF grants CCF-1527110, CCF-1618280, CCF-1814603, CCF-1910588, NSF CAREER award CCF-1750808 and a Sloan Research Fellowship.

\section*{Appendix}
\appendix
\section{A Lower Bound for Random Order Online Edge Coloring}\label{sec:ro-lb}
\citet{bar1992greedy} gave a simple lower bound for edge coloring under adversarial arrivals. Specifically, they showed a family of graphs $\mathcal{F}$ with maximum degree $\Delta=O(\sqrt{\log n})$ for which any 
randomized online algorithm $\mathcal{A}$ colors some graphs in $\mathcal{F}$ with $2\Delta-1$ colors with constant probability. Extending these ideas slightly, we show that the same holds even if the arrival order is randomized. 

\begin{lem}
	There exists a distribution over $n$-node graphs $\mathcal{G}$ of maximum degree $\Delta=\Omega(\sqrt{\log n})$, for which any 
	online edge coloring algorithm $\mathcal{A}$ must, with constant probability, use $2\Delta-1$ colors on a graph $G\sim \mathcal{G}$ presented in random order.
\end{lem}
\begin{proof}
	Consider a star on $\Delta-1$ leaves. If Algorithm $\mathcal{A}$ uses $2\Delta-2$ or fewer colors, then it may color any such star's edges with ${2\Delta-2 \choose \Delta-1}$ possible subsets of colors. 
	If $\Delta$ such stars' edges are colored using the same subset $S\subseteq [2\Delta-2]$ of $\Delta-1$ colors and some node $v$ neighbors the roots of these $\Delta$ stars, then the algorithm fails, as it is forced to use $\Delta$ colors outside of $S$ for the edges of $v$, for a total of $2\Delta-1$ distinct colors. We show a random graph $\mathcal{G}$ for which this bad event happens with constant probability, even when the edges are presented in random order.
	
	Our graph consists of independent copies of the following random graph, $\mathcal{H}$. The graph $\mathcal{H}$ contains $\beta := 2\Delta\cdot {2\Delta-2 \choose \Delta - 1} \cdot {2\Delta-1 \choose \Delta} \leq 4^{O(\Delta)}$ stars with $\Delta-1$ leaves, and one node $v$ which neighbors the centers of $\Delta$ randomly-chosen such stars. 
	For any star, the probability that all $\Delta-1$ edges of the star arrive before any of the $\Delta$ edges of $v$ arrive is $\frac{(\Delta-1)!\Delta!}{(2\Delta-1)!} = 1/{2\Delta-1 \choose \Delta}$. Therefore, by linearity of expectation, if we denote by $X$ the fraction of stars in $\mathcal{H}$ whose edges arrive before all edges of $v$, we have that $\mu:=\E[X] = 1/{2\Delta-1 \choose \Delta}$. By Markov's inequality applied to the non-negative variable $Y:=1-X$, whose expectation is $\E[Y]=1-\mu$, we have that
	\begin{equation}\label{ro-coloring-reverse-markov}
	\Pr\left[X\leq \frac{\mu}{2}\right] = \Pr\left[Y \geq 1-\frac{\mu}{2}\right] \leq \frac{1-\mu}{1-\frac{\mu}{2}} = 1 - \frac{\mu}{2-\mu} \leq 1-\frac{\mu}{2}.
	\end{equation}

	Now, if $X\geq \frac{\mu}{2} = 1/(2\cdot {2\Delta-1 \choose \Delta})$, then at least $\beta\cdot \frac{\mu}{2} = \Delta\cdot {2\Delta-2 \choose \Delta - 1}$ of the stars of $H\sim \mathcal{H}$ have all their edges arrive before any edge of $v$ arrives. By pigeonhole principle, some $\Delta$ of these stars are colored with a common set of $\Delta-1$ colors, $S\subset [2\Delta-1]$. If $v$ neighbors the roots of $\Delta$ such stars whose edges are colored with the colors in $S$, then Algorithm $\mathcal{A}$ fails, as it must color the graph using $2\Delta-1$ colors, as argued above. Therefore, conditioned on $X\geq \frac{\mu}{2}$, Algorithm $\mathcal{A}$ fails when coloring $H\sim \mathcal{H}$ with probability at least
	\begin{equation}\label{ro-prob-fail-many-colored-stars}
		\Pr\left[\mathcal{A} \textrm{ fails on } H\sim \mathcal{H} \midd X\geq \frac{\mu}{2}\right] \geq 1/{\beta \choose \Delta} \geq 4^{-O(\Delta^2)}.
	\end{equation}

	Consequently, combining \Cref{ro-coloring-reverse-markov} and \Cref{ro-prob-fail-many-colored-stars}, and using $\mu=1/{2\Delta-1 \choose \Delta} = 4^{-O(\Delta)}$, we find that the unconditional probability of Algorithm $\mathcal{A}$ not failing due to $H$ is at most 
	\begin{align*}
		\Pr\left[\mathcal{A} \textrm{ does not fail on } H\sim \mathcal{H}\right] & \leq 1 - \Pr\left[\mathcal{A} \textrm{ fails on } H\sim \mathcal{H} \midd X\geq \frac{\mu}{2}\right] \cdot \Pr\left[X\geq \frac{\mu}{2}\right] \\
		& \leq 1 - 4^{-O(\Delta^2)}\cdot \frac{\mu}{2} \\
		& = 1 - 4^{-O(\Delta^2)}.
	\end{align*}
	
	As stated above, the random graph $\mathcal{G}$ we consider consists of some $\gamma$ independent copies of $\mathcal{H}$.
	For independent copies of $\mathcal{H}$, 
	the above upper bound on the probability of $\mathcal{A}$ not failing on a copy $H\sim \mathcal{H}$ holds independently of other copies' realization and coloring by $\mathcal{A}$.
	Therefore, letting $\mathcal{G}$ consist of some sufficiently large $\gamma:=4^{\Theta(\Delta^2)}$ independent copies of $\mathcal{H}$, we have that 
	\begin{align*}
	\Pr\left[\mathcal{A} \textrm{ does not fail on } G\right] \leq \left(1 - 4^{-O(\Delta^2)}\right)^\gamma \leq \frac{1}{e}.
	\end{align*}	
	Therefore, Algorithm $\mathcal{A}$ fails on $G$ with constant probability. The lemma follows by noting that $G$ consists of some $n = \gamma\cdot (\beta+1) = 4^{\Theta(\Delta^2)}$ nodes, and therefore $\Delta = \Omega(\sqrt{\log n})$.
\end{proof}
\section{Deferred Proofs from \Cref{sec:ro}}\label{sec:ro-appendix}

In this section we give the full proofs deferred from \Cref{sec:ro}, restated below for ease of reference.

We start by proving \Cref{ro-sampling}, which intuitively implies that we can use the stream's random order to sample edges independently.
\rosamplinglem*
\begin{proof}
	For any given subset $S\subseteq U$ of size $|S|=k$, we have $U_k = S$ precisely when $X=k$ and the elements of $S$ are the first $k$ elements in the stream. As $X$ is independent of the stream's randomness, this gives us:
	$$\Pr[U_k = S] = \Pr[X = k]\cdot \Pr[S \textrm{ is a prefix of the stream}] = \Pr[X=k]/{n \choose k} = p^k (1-p)^{n-k}.$$	
	This is precisely the probability of  getting a specific set $S$, when each of the $n$ elements in $U$ is sampled independently with probability $p$.
\end{proof}

Next, we prove \Cref{estimate-using-eps-fraction}, which asserts that the edges until time $T$ where some vertex reaches degree $\epsilon\cdot \Delta$ gives sharp estimates of natural graph parameters, such as the number of edges, and every vertex's degree, w.h.p. 

\estimateslem*
\begin{proof}
	The proof relies on several applications of Chernoff bound and union bound, as follows. Fix some vertex $v$, and let $X_i$ be an indicator variable for the $i$-th edge in the stream containing $v$.
	Clearly, we have that $\E[X_i]=\frac{d(v)}{m}$, and so by linearity of expectation 
	\begin{equation}\label{expected-degree}
	\E[d^{(k)}(v)] = k\cdot \frac{d(v)}{m}.
	\end{equation} 
	On the other hand, the joint distribution $(X_1,\dots,X_m)$ is a permutation distribution, and so it is NA. We may therefore apply Chernoff bounds to sums of such variables, such as $d^{(k)}(v)=\sum_{i=1}^k X_i$. We will make use of this to prove that the three properties hold w.h.p.
	
	We begin by proving Property \ref{m-estimate}. Fix any vertex $v$ of degree $d(v)=\Delta$. For any $k \geq \epsilon\cdot m(1+\epsilon^2)$, we have by \Cref{expected-degree} that $\E[d^{(k)}(v)]\geq \eps \cdot \Delta (1+\eps^2)$. Consequently,  by Chernoff bound, we have
	\begin{align*}
	\Pr\left[d^{(k)}(v) \leq \epsilon\cdot \Delta\right] & 
	= \Pr\left[d^{(k)}(v) \leq \frac{1}{1+\epsilon^2}\cdot \epsilon\cdot \Delta(1+\eps^2)\right] \\
	& \leq \Pr\left[d^{(k)}(v) \leq (1-\nicefrac{\epsilon^2}{2})\cdot \epsilon\cdot \Delta(1+\eps^2)\right] & \frac{1}{1+x}\leq 1-\frac{x}{2}\,\, \forall x\leq 1  \\
	& \leq \exp\left(\frac{-\nicefrac{\epsilon^4}{4} \cdot \epsilon\cdot \Delta(1+\epsilon^2)}{3}\right) & \textrm{\Cref{chernoff-NA}}\\
	& \leq n^{-\alpha}. & \Delta\geq \frac{12\alpha \log n}{\epsilon^5}
\end{align*}
	Therefore, $\Pr[T\geq \epsilon\cdot m(1+\epsilon^2)] \leq n^{-\alpha}$.
	
	Now, fix a vertex $v$. For any $k \leq \epsilon\cdot m(1-\epsilon^2)$, we have by \Cref{expected-degree} that $\E[d^{(k)}(v)]\leq \eps \cdot d(v) (1-\eps^2) \leq \eps \cdot \Delta (1-\eps^2)$. Consequently,  by Chernoff bound, we have
		\begin{align*}
	\Pr\left[d^{(k)}(v) \geq \epsilon\cdot \Delta\right] & 
	= \Pr\left[d^{(k)}(v) \geq \frac{1}{1-\epsilon^2}\cdot \epsilon\cdot \Delta(1-\eps^2)\right] \\
	& \leq \Pr\left[d^{(k)}(v) \geq (1+\epsilon^2)\cdot \epsilon\cdot \Delta(1-\eps^2)\right] & 1+x\leq \frac{1}{1-x}\,\,\forall x<1 \\
	& \leq \exp\left(\frac{-\epsilon^4 \cdot \epsilon\cdot \Delta(1+\epsilon^2)}{3}\right) & \textrm{\Cref{chernoff-NA}}\\
	& \leq n^{-(\alpha+1)}. & \Delta\geq \frac{6(\alpha+1) \log n}{\epsilon^5}
	\end{align*}
	Taking union bound over all vertices, we find that $\Pr[T\leq \epsilon\cdot m(1-\epsilon^2)]\leq n^{-\alpha}$, which together with the above, implies that Property \ref{m-estimate} holds with probability at least $1-2n^{-\alpha}$.
	
We next prove that Properties \ref{degree-estimates} and \ref{last-edges} hold w.h.p. Consider the following property for a generic $k\leq m$:	
	\begin{equation}\label{degree-estimates-equation}
	d^{(k)}(v) = \epsilon\cdot d(v) \pm 2\epsilon^3 \Delta.
	\end{equation}
Let $B^{(k)}$ denote the bad event that some vertex $v$ \emph{fails} to satisfy \Cref{degree-estimates-equation} for this $k$. We will show that the probability of any event $B^{(k)}$ to happen for $k\in [\epsilon m (1-\eps^2),m]$ is at most $2n^{-\alpha}$.
Observe that Property \ref{degree-estimates} and Property \ref{last-edges} hold if events $B^{(T)}$ and $B^{(m')}$, resp., do \emph{not} happen. Assuming Property \ref{m-estimate}, both $T$ and $m'$ fall in the range $[\eps m (1-\eps^2),m]$. It then follows that the $3$ properties simultaneously hold with probability at least $1-4n^{-\alpha}$.

	\Cref{degree-estimates-equation} trivially holds for vertices $v$ with $d(v)\leq \epsilon^3\Delta$, for which $0\leq d^{(k)}(v) \leq d(v) \leq \epsilon^3\Delta$. On the other hand, for vertices $v$ with $d(v)\geq \epsilon^3 \Delta$ and $k$ as above, we have by \Cref{expected-degree} that $\E[d^{(k)}(v)] = \eps \cdot \Delta (1\pm \eps^2)$. Consequently,  by Chernoff bound, we have
	\begin{align*}
	\Pr\left[d^{(k)}(v) \geq \epsilon\cdot d(v) + 2\epsilon^3\Delta \right] & 
	\leq \Pr\left[d^{(k)}(v) \geq (1+2\epsilon^2)\cdot \epsilon\cdot d(v)\right] & d(v)\leq \Delta \\
	& \leq \Pr\left[d^{(k)}(v) \geq (1+\nicefrac{\epsilon^2}{2})\cdot \epsilon\cdot d(v)(1+\eps^2)\right] & 1+2\epsilon^2\geq (1+\epsilon^2)(1+\nicefrac{\epsilon^2}{2})\\
	& \leq \exp\left(\frac{-\nicefrac{\epsilon^4}{4} \cdot \epsilon\cdot d(v)(1+\epsilon^2)}{3}\right) & \textrm{\Cref{chernoff-NA}}\\
	& \leq \exp\left(\frac{-\nicefrac{\epsilon^4}{4} \cdot \epsilon\cdot \epsilon^3\Delta (1+\epsilon^2)}{3}\right) & d(v)\geq \epsilon^3\Delta \\	
	& \leq n^{-(\alpha+3)}. & \Delta\geq \frac{12(\alpha+3) \ln n}{\epsilon^8}
	\end{align*}	
	Similarly, we have that
	\begin{align*}
	\Pr\left[d^{(k)}(v) \leq \epsilon\cdot d(v) - 2\epsilon^3\Delta \right] & 
	\leq \Pr\left[d^{(k)}(v) \leq (1-2\epsilon^2)\cdot \epsilon\cdot d(v)\right] & d(v)\leq \Delta \\
	& \leq \Pr\left[d^{(k)}(v) \leq (1-\nicefrac{\epsilon^2}{2})\cdot \epsilon\cdot d(v)(1-\eps^2)\right] & 1-2\epsilon^2\leq (1-\epsilon^2)(1-\nicefrac{\epsilon^2}{2})\\
	& \leq \exp\left(\frac{-\nicefrac{\epsilon^4}{4} \cdot \epsilon\cdot d(v)(1-\epsilon^2)}{3}\right) & \textrm{\Cref{chernoff-NA}}\\
	& \leq \exp\left(\frac{-\nicefrac{\epsilon^4}{4} \cdot \epsilon\cdot \epsilon^3\Delta (1-\epsilon^2)}{3}\right) & d(v)\geq \epsilon^3\Delta \\	
	& \leq n^{-(\alpha+3)}. & \Delta\geq \frac{24(\alpha+3) \ln n}{\epsilon^8}
	\end{align*}	
	
	By union bound, we have that 
	$$\Pr[B^{(k)}] \leq \sum_{v} \Pr[d^{(k)}(v) \neq \eps\cdot d(v)\pm 2\eps^3\Delta] \leq 2n^{-(\alpha+2)}.
	$$ 
The claim follows by union bound over the values $k\in [\epsilon m (1-\eps^2),m]$.
\end{proof}

Finally, we show that adding $\Delta$ dummy vertices per vertex $v$ in $G$ forming a $\Delta$-clique, and $(\Delta-(1/\eps-1)\cdot d^{(T)}(v))^+$ dummy edges\footnote{We use the symbol $x^+$ to denote $\max(0,x)$.} from $v$ to some of its dummy nodes, to the edges in the time range $(T,m'-T]$, yields a near-regular graph $H$.

\dummysubgraphlem*
\begin{proof}
	By \Cref{estimate-using-eps-fraction} (\Cref{degree-estimates}), the number of dummy edges of  $v$ is $\Delta - (1/\eps-1)\cdot d^{(T)}(v)$, as this number is non-negative for every vertex $v\in V$ w.h.p., since
	$$\Delta - (1/\eps-1)\cdot d^{(T)}(v) = \Delta - (1-\eps)\cdot (d(v) \pm 2\epsilon^2 \Delta) \geq \Delta(1-(1-\eps)-2\epsilon^2) = \Delta(\epsilon-2\epsilon^2) \geq 0.$$
	Each dummy node $v_i$ belongs to a single clique of size $\Delta$, and possibly has another dummy edge to a single real node $v$, and therefore it has degree $d_H(v)\in \{\Delta-1,\Delta\}$.
	As for real vertices $v$, again appealing to  \Cref{estimate-using-eps-fraction} (\Cref{degree-estimates}), combined with $v$ having at most $\epsilon^2 \Delta$ edges in the range $(m',m]$ by  \Cref{estimate-using-eps-fraction} (\Cref{last-edges}), we find that the degree of $v$ in $H$ satisfies
	\begin{align*}
	d_H(v) & = d(v) - d^{(T)}(v) + \Delta - (1/\eps-1)\cdot d^{(T)}(v) \pm 2\epsilon^2\Delta \\
	& = d(v) + \Delta - (d(v) \pm 2\epsilon^2\Delta) \pm 2\epsilon^2 \Delta \\
	& = \Delta(1\pm 4\epsilon^2).\qedhere
	\end{align*}
\end{proof}
\section{Analysis of the Basic Algorithm (Proof of \Cref{main:th:main})}\label{sec:agnostic}
In this section, we analyze our  {\em basic algorithm} (\Cref{alg:model-agnostic}) from \Cref{sec:overview:algo}. Before proceeding any further, the reader will find it useful to review \Cref{sec:overview:algo}.  The analysis of \Cref{alg:model-agnostic} boils down to proving that the uncolored subgraph after phase one, consisting of all the edges that either failed or were not sampled in the first place, has bounded degree. This is (re-)stated in \Cref{main:th:main} below.

\uncoloreddeg*

The rest of this section  (\Cref{sec:agnostic})  is   dedicated to the proof of \Cref{main:th:main}, and it is organized as follows. In \Cref{main:sec:analysis}, we give a brief, high-level and {\em informal} overview of the proof of \Cref{main:th:main}. We start the formal proof of \Cref{main:th:main} in \Cref{sec:notations:key}, which defines some key random variables and events that will be extensively used in our analysis. In \Cref{app:model:agnostic:main:proof}, we show how \Cref{main:th:main} follows from a sequence of lemmas. The remainder of \Cref{sec:agnostic} is devoted to the proofs of these individual lemmas from \Cref{app:model:agnostic:main:proof}.

\medskip
Throughout \Cref{sec:agnostic}, we use the notation $a\pm b$ to denote the interval $[a-b, a+b]$. Thus, whenever we write $x = a\pm b$ in this section, it means that $x \in [a-b, a+b]$. Similarly, whenever we write $a\pm b = a' \pm b'$, it means that $[a-b, a+b] \subseteq [a'-b', a'+b']$.

\subsection{An informal overview of the proof of~\Cref{main:th:main}}
\label{main:sec:analysis}

For all $i \in [t_{\eps}-1], v \in V$, let $N_i(v) := \{ (u, v) \in E_i\}$ denote the edges of $v$ that are present in $G_i = (V, E_i)$. The next lemma helps us bound the maximum degree in the subgraph $G_{t_{\eps}}$. 

\begin{lem}
	\label{main:lm:degree}
	$|N_i(v)| \approx (1-\eps)^{(i-1)} \cdot (1\pm \eps^2) \Delta$ for all $i \in [t_{\eps}], v \in V$, w.h.p.
\end{lem}

\begin{proof}(Sketch)
	The statement clearly holds for $i = 1$. The proof follows from induction on $i$. Condition on all the random choices made by the algorithm during rounds $\{1, \ldots, i-1\}$. Fix any node $v \in V$, and suppose that $|N_i(v)| \approx (1-\eps)^{(i-1)} \cdot (1\pm \eps^2)\Delta$ for some $i \in [t_{\eps}-2]$.  Each edge $e \in N_i(v)$ belongs to $N_{i+1}(v)$ independently with probability $1-\eps$. From linearity of expectation, we derive that $\E[|N_{i+1}(v)|] = (1-\eps) \cdot |N_i(v)|$. Now,  a standard Hoeffding bound gives us: $|N_{i+1}(v)| = (1-\eps) \cdot |N_i(v)| \pm \Theta(\sqrt{\Delta \ln n})$ w.h.p. From~(\ref{eq:eps}),~(\ref{eq:t}) and the inductive hypothesis, we get: $\sqrt{\Delta \ln n} \ll \eps^2 \cdot \Delta \ll \eps \cdot |N_{i}(v)|$. This implies that $|N_{i+1}(v)| \approx (1-\eps)^{i} \cdot |N_i(v)| \approx (1-\eps)^i \cdot (1\pm \eps^2) \Delta$ w.h.p. 
\end{proof}

The proof of~\Cref{main:lm:degree} gives us a glimpse as to why we need  a lower bound on $\eps$: Whenever we take a concentration bound during the analysis, we  will end up with an additive error term of the form $\sqrt{\Delta \ln n}$. We will like this additive error term to get subsumed within  $\eps^{\beta} \Delta$ for some large constant $\beta > 1$. Specifically, we will like to have $\eps^{\beta} \Delta \gg \sqrt{\Delta \ln n}$, which implies that $\eps \gg \left(\frac{\ln n}{\Delta}\right)^{1/(2\beta)}$. Henceforth, to convey the main ideas, we will often gloss over this issue in this section. 

\begin{cor}
	\label{main:cor:degree}
	$\Delta(G_{t_{\eps}}) = O\left(\eps^{1/(3K)}  \Delta\right)$, w.h.p.
\end{cor}

\begin{proof}(Sketch)
	Fix any node $v \in V$. Since $t_{\eps} \approx (\ln (1/\eps))/(2K\eps)$,~\Cref{main:lm:degree} gives us: $|N_{t_{\eps}}(v)| \approx (1-\eps)^{(t_{\eps}-1)} \cdot \Delta \approx \exp(\eps \cdot (t_{\eps}-1)) \cdot \Delta \approx \eps^{1/(2K)} \Delta$ w.h.p. In other words, every node in $G_{t_{\eps}}$ has degree at most $\eps \Delta$ w.h.p. This implies the corollary.
\end{proof}

It now remains to upper bound the maximum degree of any node in the subgraph $G_F$. Before proceeding any further, we need to introduce the following notation. For all $i \in [t_{\eps}-1]$,  $v \in V$ and  $c \in [(1+\eps^2)\Delta]$, let $N_{i, c}(v) = \{ (u, v) \in N_i(v) : c \in P_i(u)\}$ denote the set of edges in $N_i(v)$ whose other endpoints have  the color $c$ in their palettes for round $i$. We refer to the quantity $|N_{i, c}(v)|$ as the {\em $c$-degree} of the node $v$ for round $i$. The main challenge will be to bound the $c$-degrees of the nodes and the palette sizes for the edges in each round, as captured in the lemma below.

\begin{lem}
	\label{main:lm:split}
	The following guarantees hold for all rounds $i \in [t_{\eps}-1]$, w.h.p. 
	\begin{itemize}
		\item (a) $|P_i(e)| \approx (1-\eps)^{2(i-1)} \cdot (1\pm \eps^2)\Delta$ for all edges $e \in E_i$. 
		\item  (b) $|N_{i, c}(v)| \approx (1-\eps)^{2(i-1)} \cdot (1\pm \eps^2)\Delta$ for all nodes $v \in V$ and colors $c \in [\Delta]$. 
	\end{itemize}
\end{lem}

\Cref{main:lm:split} is proved via an induction on $i$. We skip the rather technical proof of this lemma in this overview section. Instead, here we only explain how this lemma is used to give an upper bound on the maximum degree in $G_F$. The following lemma will be very useful towards this end.

\begin{lem}
	\label{main:lm:failed}
	$|N_{i}(v) \cap F_i| \approx 2\eps  \cdot |N_i(v) \cap S_i| \pm \Theta(\sqrt{\Delta \ln n})$, for all $i \in [t_{\eps} -1], v \in V$, w.h.p.
\end{lem}

\begin{proof}(Sketch)
	Fix any round $i \in [t_{\eps}-1]$ and any node $v \in V$. Since each edge $e \in E_i$ is sampled in $S_i$ independently with probability $\eps$, standard concentration bounds imply that:
	\begin{equation}
	\label{main:eq:lm:failed:1}
	|N_{i, c}(x) \cap S_i| \approx \eps \cdot |N_{i, c}(x)| \text{ for all colors } c \in [(1+\eps^2)\Delta], \text{ w.h.p.} 
	\end{equation}
	Condition on all the random choices made by the algorithm during rounds $\{1, \ldots, i-1\}$, as well as the random choices  which determine the set $S_i$. Suppose that these random choices we are conditioning upon satisfy~(\ref{main:eq:lm:failed:1}) and~\Cref{main:lm:split}  for round $i$ (which anyway occur w.h.p.).
	
	Fix any edge $(u, v) \in N_i(v) \cap S_i$. The probability that $(u, v)$ belongs to the set $F_i$, conditioned on it having tentatively picked any color $c \in P_i(u, v)$ in round $i$, is given by:
	\begin{eqnarray*}
		\Pr\left[ (u,v) \in F_i \mid c(u,v) = c\right] & = & 1 - \prod_{e' \in (N_{i, c}(u) \cap S_i)} \left(1 - \frac{1}{|P_i(e')|}\right) \cdot \prod_{e' \in (N_{i, c}(v) \cap S_i)} \left(1 - \frac{1}{|P_i(e')|}\right) \\
		& \approx & 1 - \left(1 - \frac{1}{(1-\eps)^{2(i-1)} \cdot (1\pm \eps^2) \Delta}\right)^{\eps \cdot |N_{i, c}(u)| + \eps \cdot |N_{i, c}(v)|} \\
		& \approx & 1 - \left(1 - \frac{1}{(1-\eps)^{2(i-1)} \cdot (1\pm \eps^2)  \Delta}\right)^{2\eps \cdot (1-\eps)^{2(i-1)} \cdot (1\pm \eps^2) \Delta} \\
		& \approx & 1 - \exp(-2\eps) \approx 2\eps.
	\end{eqnarray*}
	Since the above derivation holds for all colors $c  \in P_i(u,v)$, we infer that $\Pr[(u, v) \in F_i] \approx 2\eps$. Now, by linearity of expectation, we get: $\E[|N_{i}(v) \cap F_i|] = \sum_{e \in N_{i}(v) \cap S_i} \Pr[(u, v) \in F_i] \approx 2 \eps \cdot |N_{i}(v) \cap S_i|$. With some extra effort (see the full version), we can show that the value of $|N_{i}(v) \cap F_i|$ is tightly concentrated around $\pm \Theta(\sqrt{\Delta \ln n})$ of its expectation.  
\end{proof}

\begin{cor}
	\label{main:cor:failed}
	$\Delta(G_F) = O(\eps \Delta)$, w.h.p.
\end{cor}

\begin{proof}
	Consider any node $v \in V$. The degree of $v$ in $G_F$ is given by $\text{deg}_v(G_F) := \sum_{i=1}^{t_{\eps} -1} |N_i(v) \cap F_i|$. Hence, w.h.p.,~\Cref{main:lm:failed} gives us the following bound on $\text{deg}_v(G_F)$. 
	\begin{eqnarray*}
		\text{deg}_v(G_F) & \approx & \sum_{i=1}^{t_{\eps}-1}\left( 2\eps  \cdot |N_i(v) \cap S_i| \pm \Theta(\sqrt{\Delta \ln n})\right) \approx 2\eps \cdot \sum_{i=1}^{t_{\eps}} |N_i(v) \cap S_i| \pm \Theta(t_{\eps} \cdot \sqrt{\Delta \ln n}) \\
		& \leq & 2 \eps \cdot \text{deg}_v(G) \pm \Theta(t_{\eps} \cdot \sqrt{\Delta \ln n})  \leq  2\eps \cdot (1+\eps^2)\Delta + \Theta(t_{\eps} \cdot \sqrt{\Delta \ln n}) = \Theta(\eps \Delta).
	\end{eqnarray*}
	In the above derivation, the last step follows from~(\ref{eq:eps}) and~(\ref{eq:t}). Since every node in $G_F$ has degree at most $O(\eps \Delta)$ w.h.p., we conclude that $\Delta(G_F) = O(\eps \Delta)$ w.h.p.
\end{proof}

\Cref{main:th:main} follows from~\Cref{main:cor:degree} and~\Cref{main:cor:failed}.

\subsection{Key random variables and events}
\label{sec:notations:key}

This section defines some random variables and events that will  be extensively used in our analysis.

\medskip
\noindent {\bf Random variables:} We will need to deal with the following random variables for each  $i \in [t_{\eps}]$. 
\begin{itemize}
\item $\palette_i(v) \subseteq [(1+\eps^{2})\Delta]$:  A color $c \in [(1+\eps^{2})\Delta]$ belongs to the set $\palette_i(v)$ iff there is no edge $(u,v) \in S_j$ which picked the color $c$ at some earlier round $j < i$. We refer to the set $\palette_i(v)$ as the {\em palette} of  the node $v \in V$ for round $i$. 
\item $\palette_i(e) \subseteq [(1+\eps^2)\Delta]$:  A color $c \in [(1+\eps^2)\Delta]$ belongs to the set $\palette_i(e)$ iff $c \in \palette_i(u) \cap \palette_i(v)$, where $e = (u, v) \in E$. We refer to the set $\palette_i(e)$ as the {\em palette} of the edge $e \in E$ for round $i$. 
\item $N_i(v)$: The set of neighboring edges of  $v \in V$ in $G_i$, that is, $N_i(v) = \{ (u, v) \in E : (u,v) \in E_i\}$.
\item $N_{i, c}(v)$: This is the set of neighboring edges of the node $v \in V$   in $G_i$ whose other endpoints  have the color $c$ in their palettes for round $i$, that is, $N_{i, c}(v) = \{ (u, v) \in N_i(v) :  c \in \palette_i(u)\}$. 
\end{itemize}
\medskip
\noindent {\bf Error parameters:} While analyzing the basic algorithm, we need to keep track of the amount by which the random variables $|P_i(e)|$ and $|N_{i,c}(v)|$ can deviate from their expected values. The magnitude of these deviations will be captured by the error-parameters $\{\gamma_i\}_{i \in [t_{\eps}]}$, where:
\begin{eqnarray}
\label{eq:gamma:1}
\gamma_1 & = & K\eps^2. \\
\gamma_{i+1} & = & (1+K\eps)\gamma_i + K \eps^2 \text{ for all } i \in [t_{\eps}-1].  \label{eq:gamma:2}
\end{eqnarray}

\begin{cor}
\label{cor:gamma}
We have $\gamma_i \leq \eps^{1/2}$ for all $i \in [t_{\eps}]$.
\end{cor}

\begin{proof}
From~(\ref{eq:gamma:1}) and~(\ref{eq:gamma:2}), we derive that:
\begin{align*}
\gamma_{t_{\eps}} = & K\eps^2 \cdot \sum_{i=0}^{t_{\eps-1}} (1+K\eps)^{i} \\
\leq & (K \eps^2) \cdot \frac{(1+K\eps)^{t_{\eps}}}{(K\eps)} \\
\leq & \eps \cdot \exp(K \eps t_{\eps}) \\
\leq & \eps \cdot \exp((1/2) \cdot \ln(1/\eps)) && (\text{by~(\ref{eq:t})}) \\
= & \eps^{1/2}.
\end{align*}
The corollary now follows from the observation that $\gamma_i \leq \gamma_{t_{\eps}}$ for all $i \in [t_{\eps}]$.
\end{proof}

\medskip
\noindent {\bf Random  events:} We will track  the  three events $\e_i, \C_i$ and $\B_i$ as defined below, for each  $i \in [t_{\eps}-1]$. 
\begin{itemize}
\item Event $\e_i$ occurs iff  $|P_{i}(e)| = (1-\epsilon)^{2(i-1)} \cdot (1\pm \gamma_i)\cdot \Delta$ for all $e \in E_i$. 
\item Event $\C_i$  occurs iff $|N_{i,c}(v)| = (1-\epsilon)^{2(i-1)} \cdot (1\pm \gamma_i) \cdot \Delta$ for all  $v \in V$ and   $c \in [(1+\eps^2)\Delta]$.
\item Event $\B_i$ occurs iff the following conditions hold for all nodes $v \in V$ and  colors $c \in [(1+\eps^2)\Delta]$:
\begin{eqnarray*}
|S_i \cap N_i(v)|  =  (\epsilon \pm \eps^2) \cdot |N_i(v)|,   \ \text{ and }  \ |S_i \cap N_{i,c}(v)|  =  (\epsilon \pm \eps^2) \cdot |N_{i,c}(v)|.
 \end{eqnarray*}
\end{itemize}

\paragraph{Remark:} Since the degree of every node in $G$ is  $(1\pm \eps^2)\Delta$, we have $\Pr[\e_1] = \Pr[\C_1] = 1$.  

\medskip
\noindent {\bf Random bits used by our algorithm:} While proving \Cref{main:th:main}, we will often need to condition upon certain critical events. It will be easier to follow the proof if we view these conditionings via the prism of a classification of random bits used by the algorithm, as described below.

During any given round $i \in [t_{\eps}-1]$, there are two distinct tasks for which the algorithm uses randomness: (a) To determine the set of sampled edges $S_i$, and (b) to pick a color $c(e)$ for each sampled edge $e \in S_i$. We let $r_i^{\text{(edges)}}$ and $r_i^{\text{(colors)}}$ respectively denote the random bits used by the algorithm for task (a) and task (b). The random bits $r_i^{\text{(edges)}}$ and $r_i^{\text{(colors)}}$ are mutually independent of each other, and they are also independent of all the random bits used in the previous rounds $j < i$. We let $r_i = r_i^{\text{(edges)}} \cup r_i^{\text{(colors)}}$ denote all the random bits used by the algorithm in round $i$. Furthermore, we let $r_{< i} = \bigcup_{j=1}^{i-1} r_j$ denote the set of all random bits used by the algorithms in rounds $\{1, \ldots, i-1\}$. Note that the random bits $r_{<i}$ completely determine the occurrences of the following events: $\{\e_j, \C_j, \B_j\}_{j < i}$ and $\{\e_i, \C_i\}$.  On the other hand, the occurrence of the event $\B_i$ is completely determined by the random bits $r_{< i} \cup r_i^{\text{(edges)}}$.

\subsection{Proof of \Cref{main:th:main}}
\label{app:model:agnostic:main:proof}

The main challenge is to show that all the key events defined in~\Cref{sec:notations:key} occur w.h.p. This is summarized in~\Cref{cor:key-lem:key}, which in turn follows from~\Cref{lm:sampled:degree},~\Cref{key-lem} and~\Cref{key-lem-c}. The proofs of these three crucial lemmas respectively appear in~\Cref{sec:sampled:degree},~\Cref{sec:key-lem} and~\Cref{sec:key-lem-c}.

\begin{lem}\label{lm:sampled:degree}
Consider any round $i \in [t_{\eps}-1]$, and fix any instantiation of the  bits $r_{<i}$ which ensure the occurrence of the event $\e_i \cap \C_i$. Then we have: $\Pr[\B_i \mid r_{<i}] \geq 1 - 1/n^{1500}$.
\end{lem}

\begin{lem}
\label{keylem}
\label{key-lem}
Consider any round $i \in [t_{\eps}-2]$, and fix any instantiation of the  bits $r_{<i} \cup r_i^{(\text{edges})}$ which ensure the occurrence of the event $\e_i \cap \C_i \cap \B_i$. 	Then we have: 
$\Pr\left[\mathcal{E}_{i+1}  \mid r_{<i} \cup r_i^{(\text{edges})}\right]   \geq 1-1/n^{1500}.$
\end{lem}

\begin{lem}
\label{key-lem-c}
Consider any round $i \in [t_{\eps}-2]$, and fix any instantiation of the  bits $r_{<i} \cup r_i^{(\text{edges})}$ which ensure the occurrence of the event $\e_i \cap \C_i \cap \B_i$. 	Then we have:  $\Pr\left[\mathcal{C}_{i+1} \mid r_{<i} \cup r_i^{(\text{edges})}\right]   \geq 1-1/n^{500}$.
\end{lem}

\begin{cor}
\label{cor:key-lem:key}
We have $\Pr\left[\bigcap_{i=1}^{t_{\eps}-1} \left(\mathcal{E}_{i} \cap \mathcal{C}_{i} \cap \B_i \right)\right]   \geq 1-1/n^{400}$.
\end{cor}

\begin{proof}
First, recall that $\Pr[\e_1] = \Pr[\C_1] = 1$, and hence~\Cref{lm:sampled:degree} gives us:
\begin{equation}
\label{eq:cor:key-lem:key:100}
\Pr[\e_1 \cap \C_1 \cap \B_1] \geq 1 - 1/n^{1500}. 
\end{equation}
Next, applying a union bound over  \Cref{key-lem} and~\Cref{key-lem-c}, we get:
\begin{equation}
\label{eq:cor:key-lem:key:1}
\Pr[\mathcal{E}_{i+1} \cap \C_{i+1} \mid \mathcal{E}_i \cap \mathcal{C}_i \cap \B_i]   \geq 1-2/n^{500} \text{ for all rounds } i \in [t_{\eps} - 2].
\end{equation}
Accordingly, from~(\ref{eq:cor:key-lem:key:1}) and~\Cref{lm:sampled:degree} we infer that:
\begin{eqnarray}
\Pr[\e_{i+1} \cap \C_{i+1} \cap \B_{i+1} \mid \e_i \cap \C_i \cap \B_i] & = & \Pr[\e_{i+1} \cap \C_{i+1} \mid \e_i \cap \C_i \cap \B_i] \cdot \Pr[\B_{i+1} \mid \e_{i+1} \cap \C_{i+1}] \nonumber \\
& \geq & \left(1 - 2/n^{500}\right) \cdot \left(1 - 1/n^{1500}\right) \text{ for all  } i \in [t_{\eps} - 2]. \label{eq:cor:key-lem:key:105}
\end{eqnarray}
Now, from~(\ref{eq:cor:key-lem:key:100}) and~(\ref{eq:cor:key-lem:key:105}) we derive that:
\begin{eqnarray*}
\Pr\left[\bigcap_{i=1}^{t_{\eps}-1} \left(\mathcal{E}_{i} \cap \mathcal{C}_{i} \cap \B_i \right)\right] & = & \Pr[\e_1 \cap \C_1 \cap \B_1] \cdot \prod_{i=1}^{t_{\eps}-2} \Pr[\e_{i+1} \cap \C_{i+1} \cap \B_{i+1} \mid \e_i \cap \C_i \cap \B_i] \\
& \geq & \left(1 - 1/n^{1500}\right) \cdot \left(1 - 2/n^{500}\right)^{t_{\eps}} \cdot \left(1 - 1/n^{1500}\right)^{t_{\eps}} \\
& \geq & \left(1 - 1/n^{1500}\right) \cdot \left(1 - 2t_{\eps}/n^{500}\right) \cdot \left(1 - t_{\eps}/n^{1500}\right) \\
& \geq & 1 - 1/n^{1500} - 2t_{\eps}/n^{500} - t_{\eps}/n^{1500} \geq 1 - 1/n^{400}.
\end{eqnarray*}
In the derivation above, the last inequality holds since $t_{\eps} \leq n$ and $n \geq 2$. 
\end{proof}

In order to prove \Cref{main:th:main}, we need to upper bound  the maximum degree of any node in the subgraph $G_{t_{\eps}} \cup G_F$. Accordingly,~\Cref{cor:sampled:degree} upper bounds the maximum degree in the subgraph $G_{t_{\eps}}$, whereas~\Cref{cor:degree-failures} (which follows from~\Cref{degree-failures}), upper bounds the maximum degree in the subgraph $G_F$.~\Cref{sec:degree-failures} contains the proof of~\Cref{degree-failures}.

\begin{cor}
\label{cor:sampled:degree}
$\Delta(G_{t_{\eps}}) = O\left(\epsilon^{1/(3K)} \cdot \Delta\right)$ with probability at least $1 - 1/n^{400}$.
\end{cor} 

\begin{proof}
Define the event $\B := \bigcap_{i \in [t_{\eps}-1]} \B_i$. From~\Cref{cor:key-lem:key}, we infer that $\Pr[\B] \geq 1 - 1/n^{400}$. Henceforth, we condition on the event $\B$.

Consider any node $v \in V$. Conditioned on the event $\B$, we have $|S_i \cap N_i(v)| \geq (\eps-\eps^2)  \cdot |N_i(v)|$ for each round $i \in [t_{\eps}- 1]$. Since $|N_{i+1}(v)| = |N_i(v)| - |N_i(v) \cap S_i|$, we infer that:
\begin{equation}
\label{eq:cor:sampled:degree:1}
|N_{i+1}(v)| \leq (1-\eps+\eps^2) \cdot |N_i(v)|  \text{ for all } i \in [t_{\eps}-1], \text{ conditioned on the event } \B.
\end{equation}
Since $|N_1(v)| \leq (1+ \eps^{2})\Delta$, it is now easy to derive from~(\ref{eq:cor:sampled:degree:1}) that:
\begin{eqnarray}
|N_{t_{\eps}}(v)| & \leq & (1-\eps +\eps^2)^{t_{\eps}} \cdot (1+\eps^{2}) \cdot \Delta  \nonumber \\
& \leq & \exp(-\eps(1-\eps) t_{\eps}) \cdot  (1+\eps^2) \cdot \Delta   \nonumber \\
& \leq & \exp\left( - \eps(1-\eps) \cdot \left(\frac{\ln (1/\eps)}{2K\eps} - 1\right)\right) \cdot (1+\eps^2) \cdot \Delta  \qquad \qquad \ \ \ \  \  (\text{by~(\ref{eq:eps}) and~(\ref{eq:t})})   \nonumber \\
& \leq & \exp\left( - \frac{(1-\eps)\ln(1/\eps)}{2K} + 2\eps\right) \cdot (1+\eps^2) \cdot \Delta \qquad \qquad \qquad  \ \ \ \ \ (\text{by~(\ref{eq:eps})}) \nonumber \\
& \leq & \eps^{(1-\eps)/(2K)} \cdot (1+4\eps) \cdot (1+\eps^2) \cdot  \Delta \qquad \qquad \qquad  \qquad \qquad \ \ \  \  (\text{by~(\ref{eq:eps})})  \nonumber \\
& \leq &  O(\eps^{1/(3K)} \cdot  \Delta) \text{ for all  } v \in V, \text{ conditioned on the event } \B.  \ \ \ (\text{by~(\ref{eq:eps}) and~(\ref{eq:t})}) \nonumber
\end{eqnarray}
 The degree of a node $v \in V$ in $G_{t_{\eps}}$ exactly equals $|N_{t_{\eps}}(v)|$. Hence,  we conclude that $\Delta(G_{t_{\eps}}) = O\left(\eps^{1/(3K)}  \Delta\right)$, conditioned on the event $\B$. The corollary follows since   $\Pr[\B] \geq 1 - 1/n^{400}$.
\end{proof}

\begin{lem}\label{degree-failures}
Consider any round $i \in [t_{\eps}-1]$, and fix any instantiation of the bits $r_{< i} \cup r_i^{(\text{edges})}$ which ensure the occurrence of  the event $\e_i \cap \C_i \cap \B_i$. Then 
$\Pr\left[ \Delta(G_{F_i}) \leq 9 \epsilon^2 \Delta \mid  r_{< i}  \cup r_i^{(\text{edges})} \right] \geq 1 - 1/n^{300}$, where $G_{F_i} = (V, F_i)$ denotes the subgraph of $G$ consisting of the failed edges in round $i$.
\end{lem}

\begin{cor}
\label{cor:degree-failures}
$\Delta(G_F) = O(\epsilon \ln(1/\epsilon) \cdot \Delta)$ with probability at least $1 - 1/n^{200}$.
 \end{cor}

\begin{proof}
From~\Cref{cor:key-lem:key} and~\Cref{degree-failures}, we infer that:
\begin{eqnarray}
\Pr\left[\Delta(G_{F_i}) =O(\eps^{2}\cdot \Delta)\right] & = & \Pr\left[\Delta(G_{F_i}) =O(\eps^{2}\cdot \Delta) \mid \e_i \cap \C_i \cap \B_i \right] \cdot \Pr[\e_i \cap \C_i \cap \B_i] \nonumber \\
& \geq & \left(1 - 1/n^{300} \right) \cdot \left(1 - 1/n^{400}\right) \nonumber \\
& \geq & 1 - 1/n^{300} - 1/n^{400} \nonumber \\
& \geq & 1 - 1/n^{250} \text{ for each round } i \in [t_{\eps} - 1]. \label{eq:cor:degree-failures:1}
\end{eqnarray}
As $\Delta(G_F) \leq \sum_{i=1}^{t_{\eps}-1} \Delta(G_{F_i})$ and $t_{\eps} \leq n$, the corollary follows from~(\ref{eq:t}) and a union bound over~(\ref{eq:cor:degree-failures:1}).\end{proof}

\Cref{main:th:main} now follows by  \Cref{cor:sampled:degree} and \Cref{cor:degree-failures}.

\subsection{A couple of important technical claims}

Here, we prove two technical claims that will be used multiple times in the subsequent sections.

\begin{cla}
\label{cl:technical:1}
Fix any round $i \in [t_{\eps}-1]$ and condition on the event $\e_i \cap \C_i$. Then we have:
\begin{eqnarray}
(\eps^2/2) \cdot |N_{i, c}(v)| & \geq &  50 \sqrt{\Delta \ln n}    \text{ for all nodes } v \in V \text{ and colors } c \in [(1+\eps^2)\Delta]. \label{eq:technical:1:1} \\
(\eps^2/2) \cdot |N_i(v)| & \geq & 50 \sqrt{\Delta \ln n} \text{ for all nodes } v \in V. \label{eq:technical:1:2}\\
(\eps^2/2) \cdot |P_i(e)| & \geq & 50 \sqrt{\Delta \ln n} \text{ for all edges } e \in E_i. \label{eq:technical:1:3}
\end{eqnarray} 
\end{cla}

\begin{proof}
Fix any node $v \in V$ and any color $c \in [(1+\eps^2)\Delta]$. Conditioned on the event $\e_i \cap \C_i$, we get:
\begin{align*}
(\eps^2/2) \cdot |N_{i, c}(v)|  \geq & (\eps^2/2) \cdot (1-\eps)^{2(i-1)} \cdot (1-\gamma_i) \cdot \Delta \\
\geq & (\eps^2/2) \cdot (1-\eps)^{2(t_{\eps}-1)} \cdot (1-\eps^{1/2}) \cdot \Delta && (\text{by~\Cref{cor:gamma}}) \\
\geq & (\eps^2/2) \cdot \exp(-4 \eps (t_{\eps}-1)) \cdot (\Delta/2) && (\text{by~(\ref{eq:eps})}) \\
\geq & (\eps^2/2) \cdot \exp\left(- (2/K) \cdot \ln(1/\eps) \right) \cdot (\Delta/2) && (\text{by~(\ref{eq:t})}) \\
= & (\eps^2/2) \cdot \eps^{2/K} \cdot (\Delta/2) \\
\geq & \eps^{3} \cdot (\Delta/4) && (\text{by~(\ref{eq:eps}) and~(\ref{eq:t})}) \\
\geq & 50 \sqrt{\Delta \ln n} && (\text{by~(\ref{eq:eps})})
\end{align*}
Applying the same line of reasoning, one can  derive that $(\eps^2/2) \cdot |P_i(e)| \geq 50 \sqrt{\Delta \ln n}$ for all  $e \in E_i$. Finally,~(\ref{eq:technical:1:2}) follows from~(\ref{eq:technical:1:1}) and the observation that $N_{i}(v) \supseteq N_{i, c}(v)$.
\end{proof}

\begin{cla}
\label{cl:technical:2}
We have: 
$\left(1 - \frac{1}{(1-\epsilon)^{2(i-1)}\cdot (1\pm\gamma_i)  \Delta}\right)^{(\eps \pm \eps^{2}) \cdot (1-\eps)^{2(i-1)} \cdot (1\pm\gamma_i)  \Delta}  = 1 - \eps \pm ( 4 \eps \gamma_i + \eps^2).$
\end{cla}

\begin{proof}
Let $M = \left(1 - \frac{1}{(1-\epsilon)^{2(i-1)}\cdot (1\pm\gamma_i)  \Delta}\right)^{(\eps \pm \eps^{2}) \cdot (1-\eps)^{2(i-1)} \cdot (1\pm\gamma_i)  \Delta}$. We first upper bound  $M$  as follows. 
\begin{align}
M &  \leq  \left(1 - \frac{1}{(1-\epsilon)^{2(i-1)}\cdot (1+\gamma_i)  \Delta}\right)^{(\eps - \eps^{2}) \cdot (1-\eps)^{2(i-1)} \cdot (1-\gamma_i)  \Delta} \nonumber \\
& \leq  \exp\left(- (\eps - \eps^{2}) \cdot \left(1-\gamma_i\right) \cdot \left(1+\gamma_i\right)^{-1}    \right) \nonumber \\
& \leq  \exp\left(- (\eps - \eps^{2}) \cdot \left(1-2\gamma_i\right)  \right) \qquad \qquad \qquad \qquad \qquad \qquad \,  (\text{by~(\ref{eq:eps}) and~\Cref{cor:gamma}}) \nonumber \\
& \leq \exp\left( - (\eps - 3 \eps \gamma_i)  \right) \qquad \qquad \qquad \qquad \qquad \qquad \qquad \qquad (\text{by~(\ref{eq:eps}) and~\Cref{cor:gamma}}) \nonumber \\
& \leq 1 - (\eps - 3 \eps \gamma_i) + (1/2) \cdot (\eps - 3 \eps \gamma_i)^2 \nonumber \\
& \leq 1 - \eps + (4 \eps \gamma_i + \eps^2). \qquad \qquad \qquad \qquad \qquad \qquad \qquad \ \ \ \ (\text{by~(\ref{eq:eps}) and~\Cref{cor:gamma}}) \nonumber
\end{align}
Next, we lower bound $M$ as follows. 
\begin{align}
M & \geq  \left(1 - \frac{1}{(1-\epsilon)^{2(i-1)}\cdot (1-\gamma_i)  \Delta}\right)^{(\eps + \eps^{2}) \cdot (1-\eps)^{2(i-1)} \cdot (1+\gamma_i)  \Delta} \nonumber \\
& \geq  \exp\left(-  \frac{(\eps+\eps^{2})(1+\gamma_i)}{(1-\gamma_i)}   \right) \cdot \left(1 -  \frac{(\eps+\eps^{2}) (1+\gamma_i)}{(1-\eps)^{2(i-1)} \cdot (1-\gamma_i)^2  \Delta}\right) \nonumber \\
& \geq  \exp\left(-  \frac{(\eps+\eps^{2})(1+\gamma_i)}{(1-\gamma_i)}   \right) \cdot \left(1-\frac{8}{(1-\eps)^{2 t_{\eps}} \Delta}\right)  \qquad \qquad (\text{by~(\ref{eq:eps}) and~\Cref{cor:gamma}})  \nonumber \\
& \geq \exp\left(- (\eps + 2\eps \gamma_i )   \right) \cdot \left(1-\frac{8}{(1-\eps)^{2 t_{\eps}} \Delta}\right) \qquad \qquad \qquad \ \ \  \ (\text{by~(\ref{eq:eps}) and~\Cref{cor:gamma}}) \nonumber \\
& \geq \exp\left(- (\eps + 2\eps \gamma_i )   \right) \cdot \left(1-\frac{8 \cdot \exp(4 \eps t_{\eps})}{\Delta}\right) \qquad \qquad \qquad  \ \ \ (\text{by~(\ref{eq:eps})}) \nonumber \\ 
&  \geq \exp\left(- (\eps + 2\eps \gamma_i )   \right) \cdot \left(1-\frac{8}{\Delta} \cdot \exp\left(4 \eps + \frac{2 \cdot\ln(1/\eps)}{K}\right)\right) \ \ \,(\text{by~(\ref{eq:t})}) \nonumber \\ 
&  \geq \exp\left(- (\eps + 2\eps \gamma_i )   \right) \cdot \left(1-\frac{8}{\Delta} \cdot \left(1+8 \eps\right) \cdot \frac{1}{\eps^{2/K}}\right) \qquad \qquad  \ (\text{by~(\ref{eq:eps})}) \nonumber \\
& \geq \exp\left(- (\eps + 2\eps \gamma_i )   \right) \cdot (1-\eps^{2}) \qquad \qquad \qquad \qquad \qquad \qquad \  (\text{by~(\ref{eq:eps}) and~(\ref{eq:t}}))\nonumber \\
& \geq (1 - \eps - 2\eps \gamma_i) \cdot (1-\eps^{2}) \nonumber \\
& \geq 1 - \eps - (4 \eps\gamma_i  + \eps^2). \qquad \qquad \qquad \qquad \qquad \qquad \qquad \ \ \  \ \ (\text{by~(\ref{eq:eps}) and~\Cref{cor:gamma}}) \nonumber
\end{align}
The second inequality holds since $\left(1- \frac{\lambda}{x}\right)^x \geq \exp(-\lambda) \cdot \left(1 - \frac{\lambda^2}{x}\right)$ for all $0 < \lambda < x$.
\end{proof}

\subsection{Proof of \Cref{lm:sampled:degree}}\label{sec:sampled:degree}

For any node $v \in V$ and any color $c \in [(1+\eps^2)\Delta]$, let $\B_i(v, c)$ denote the event which occurs iff: $$|S_i \cap N_i(v)| = (\epsilon \pm \eps^2)  \cdot |N_i(v)|  \text{ and } |S_i \cap N_{i,c}(v)| = (\epsilon \pm \eps^2) \cdot |N_{i,c}(v)|.$$
We first focus on bounding  $\Pr[\B_i(v, c) \mid r_{< i}]$ for a given node $v \in V$ and color $c \in [(1+\eps^2)\Delta]$. We start by observing that $N_{i, c}(v) \subseteq N_i(v)$. For each edge $e \in N_i(v)$, consider an indicator random variable $X_e \in \{0,1\}$  that is set to one iff the edge $e$ is sampled in round $i$. Thus, we have $|S_i \cap N_i(v)| = \sum_{e \in N_i(v)} X_e$ and $|S_i \cap N_{i, c}(v)| = \sum_{e \in N_{i,c}(v)} X_e$. We also have $\E\left[ X_e \mid r_{< i} \right] = \Pr\left[X_e = 1 \mid r_{< i} \right] = \epsilon$ for all edges $e \in N_i(v)$. Hence, applying linearity of expectation, we get:
\begin{eqnarray*}
\E\left[|S_i \cap N_i(v)|  \mid  r_{< i} \right] & = & \sum_{e \in N_i(v)} \E\left[X_e  \mid  r_{< i} \right]  = \epsilon \cdot |N_i(v)|. \\
\E\left[|S_i \cap N_{i,c}(v)|  \mid  r_{< i} \right] & = & \sum_{e \in N_{i,c}(v)} \E\left[X_e  \mid  r_{< i} \right]  = \epsilon \cdot |N_{i,c}(v)|. 
\end{eqnarray*}
The  random variables $\{X_e\}$ are mutually independent (even after conditioning on $r_{< i}$). Since  $|N_{i, c}(v)| \leq |N_i(v)| \leq (1+\eps^2)\Delta$ and each $X_e$ is a $0/1$ random variable,~\Cref{hoeffding} gives us:
\begin{eqnarray}
\Pr\left[ |S_i \cap N_i(v)| = \eps \cdot |N_i(v)| \pm 50 \sqrt{\Delta \ln n} \, \Big| r_{< i} \, \right] & \geq & 1 - 1/n^{2000}. \label{eq:lm:sampled:degree:1} \\
\Pr\left[ |S_i \cap N_{i,c}(v)| = \eps  \cdot |N_{i,c}(v)| \pm 50 \sqrt{\Delta \ln n} \, \Big| r_{< i} \, \right] & \geq & 1 - 1/n^{2000}. \label{eq:lm:sampled:degree:2}
\end{eqnarray}
From~(\ref{eq:eps}) and~\Cref{cl:technical:1}, we infer that:
\begin{eqnarray}
\label{eq:lm:sampled:degree:3}
\eps   \cdot |N_i(v)| \pm 50 \sqrt{\Delta \ln n} & = & (\eps \pm \eps^2) \cdot |N_i(v)|. \\
\eps  \cdot |N_{i,c}(v)| \pm 50 \sqrt{\Delta \ln n} & = & (\eps \pm \eps^2) \cdot |N_{i,c}(v)|.
\label{eq:lm:sampled:degree:4}
\end{eqnarray}
From~(\ref{eq:lm:sampled:degree:1}) and~(\ref{eq:lm:sampled:degree:3}), we get:
\begin{equation}
\Pr\left[ |S_i \cap N_i(v)| =(\eps \pm \eps^2) \cdot |N_i(v)|  \, \Big| r_{< i} \, \right]  \geq  1 - 1/n^{2000}.
\label{eq:lm:sampled:degree:5}
\end{equation}
Similarly, from ~(\ref{eq:lm:sampled:degree:2}) and~(\ref{eq:lm:sampled:degree:4}), we get:
\begin{equation}
\Pr\left[ |S_i \cap N_{i,c}(v)| =(\eps \pm \eps^2) \cdot |N_{i,c}(v)|  \, \Big| r_{< i} \, \right]  \geq  1 - 1/n^{2000}.
\label{eq:lm:sampled:degree:6}
\end{equation}
Applying a union bound over~(\ref{eq:lm:sampled:degree:5}) and~(\ref{eq:lm:sampled:degree:6}), we get: $\Pr[\B_i(v, c) \mid r_{< i}] \geq 1 - 1/n^{1900}$. Since $\B_i = \bigcap_{v, c} \B_i(v, c)$,  the lemma follows from one last union bound  over all  $v \in V$ and $c \in [(1+\eps^2)\Delta]$.

\subsection{Proof of \Cref{key-lem}}\label{sec:key-lem}

Throughout this section, we fix the bits $r_{< i} \cup r_i^{(\text{edges})}$ which ensure the occurrence of the event $\e_i \cap \C_i \cap \B_i$. To ease notations, henceforth we  refrain from repeatedly stating that we are conditioning on the bits $r_{< i} \cup r_i^{(\text{edges})}$. However, the reader should keep in mind that we are relying upon this conditioning for the rest of~\Cref{sec:key-lem}.

We  first bound the expected value of $|P_{i+1}(e)|$ for any given edge $e \in E_{i+1}$.\footnote{Note that the subset of edges $E_{i+1}$ is completely determined by the bits $r_{<i} \cup r_i^{\text{(edges)}}$.} Next, we  show that w.h.p.~the value of $|P_{i+1}(e)|$ does not deviate too far away from its expectation. Finally, we take a union bound over all the edges $e \in E_{i+1}$ to argue that the event $\e_{i+1}$ occurs w.h.p.

\medskip
\noindent {\bf Calculating the expected value of $|P_{i+1}(e)|$ for a given edge $e = (u,v) \in E_{i+1}$:}  

\smallskip
\noindent
Since $e \in E_{i+1}$, we have $e \notin S_i$. In particular, this implies that the sets $(N_{i, c}(u) \cap S_i)$ and $(N_{i,c}(v) \cap S_i)$ are mutually disjoint. This observation will be useful in subsequent derivations. 

Consider any color $c\in \palette_i(e)$. For any endpoint $x \in \{u, v\}$ of the edge $e$, let  $\Gamma_{x,c}$ be the event that at least one edge $e' \in (N_{i, c}(x) \cap S_i)$ picks the color $c$ in round $i$. Note that $\Pr\left[ c \notin \palette_{i+1}(e)\right] = \Pr\left[ \Gamma_{u,c} \cup \Gamma_{v, c}\right]$. 
Since the sets $(N_{i, c}(u) \cap S_i)$ and $(N_{i,c}(v) \cap S_i)$ are mutually disjoint, the events $\Gamma_{u, c}$ and $\Gamma_{v, c}$ are mutually independent. Hence, from the inclusion-exclusion principle we infer that:
\begin{eqnarray}
\label{eq:key-lem-1-1}
	\Pr[c\notin  \palette_{i+1}(e)] 
	= \Pr\left[ \Gamma_{u,c}\right] + \Pr\left[ \Gamma_{v,c}\right]  - 
	\Pr\left[ \Gamma_{u,c}\right] \cdot \Pr\left[ \Gamma_{v,c}\right].
	\end{eqnarray}
	
We now  focus on estimating the value of $\Pr\left[ \Gamma_{x,c}\right]$ for a given node $x \in \{u, v\}$.  Recall that the bits $r_{< i} \cup r_i^{(\text{edges})}$  we condition upon ensure the occurrence of the event $\e_i \cap \C_i \cap \B_i$. Hence,  we have:
\begin{eqnarray}
\label{eq:key-lem-1-proof-100}
 |P_i(e')| & = &  (1-\epsilon)^{2(i-1)}\cdot (1\pm \gamma_i) \cdot \Delta \text{ for all edges } e' \in E_{i}. \\
 |N_{i,c}(x)| & = & (1-\eps)^{2(i-1)} \cdot (1\pm\gamma_i) \cdot \Delta. \\
 \label{eq:key-lem-1-proof-101}
|N_{i,c}(x) \cap S_i| & = & (\eps \pm \eps^{2}) \cdot (1-\eps)^{2(i-1)} \cdot (1\pm\gamma_i) \cdot \Delta. 
\end{eqnarray}
Since the event $\Gamma_{x,c}$ occurs iff some edge $e' \in N_{i, c}(x) \cap S_i$  picks color  $c$  in round  $i$, we infer that:
\begin{eqnarray}
\Pr\left[ \Gamma_{x,c} \right]  &  = &   1- \prod_{e' \in N_{i,c}(x)\cap S_i} \left(1-\frac{1}{|\palette_i(e')|}\right) \nonumber  \\
&  = &  1 - \left(1 - \frac{1}{(1-\epsilon)^{2(i-1)}\cdot (1\pm\gamma_i) \cdot \Delta}\right)^{(\eps \pm \eps^{2}) \cdot (1-\eps)^{2(i-1)} \cdot (1\pm\gamma_i) \cdot \Delta} \nonumber  \\
& = & \eps \pm (4\eps \gamma_i+\eps^2). \label{eq:key-lem-1-proof-700}
\end{eqnarray}
In the derivation above, the  second step follows from~(\ref{eq:key-lem-1-proof-100}) and~(\ref{eq:key-lem-1-proof-101}), whereas the last step follows from~\Cref{cl:technical:2}. From~(\ref{eq:key-lem-1-1}) and~(\ref{eq:key-lem-1-proof-700}),  we  next infer that: 
	$$\Pr[c \notin \palette_{i+1}(e)] =  2 \left(\eps \pm (4\eps\gamma_i+\eps^2)\right) - \left(\eps \pm (4\eps\gamma_i+\eps^2)\right)^2 \text{ for every color } c \in P_i(e).$$ 
Equivalently, for every color $c \in \palette_i(e)$, we have: 
$$\Pr[c \in \palette_{i+1}(e)] = 1 -  \Pr[c \notin \palette_{i+1}(e)]  = \left(1-\eps \pm (4\eps\gamma_i+\eps^2) \right)^2.$$  Applying linearity of expectation, we now get:
	\begin{equation}
	\label{eq:key-lem-1-proof-1}
		\E[|\palette_{i+1}(e)|] = \sum_{c \in \palette_i(e)} \Pr[c \in \palette_{i+1}(e)]  =  \left(1-\eps \pm (4\eps\gamma_i+\eps^2) \right)^2 \cdot |P_i(e)|.
	\end{equation}

\medskip
\noindent {\bf Deriving a concentration bound on $|P_{i+1}(e)|$ for a given edge $e = (u, v) \in E_{i+1}$:}

\smallskip
\noindent For each color $c \in P_i(e)$, let $X_c \in \{0,1 \}$ be an indicator random variable that is set to one iff $c \in P_{i+1}(e)$. Clearly, we have: $|P_{i+1}(e)| = \sum_{c \in P_i(e)} X_c$. We will now show that the random variables $\{X_c\}_{c \in P_i(e)}$ are negatively associated, and then apply Hoeffding bound.

\begin{cla}
\label{cl:key-lem-1-proof-1}
The random variables $\{X_c\}_{c \in P_i(e)}$ are negatively associated.
\end{cla}

\begin{proof}
For each color $c \in P_i(e)$ and each edge $e' \in (N_{i}(u) \cap S_i) \cup (N_{i}(v)) \cap S_i)$, define an indicator random variable $X_{c, e'} \in \{0,1\}$ that is set to one iff the edge $e'$ picks color $c$ in round $i$. Since each edge   picks at most one color in round $i$, \Cref{0-1-NA} implies that for each edge $e' \in (N_{i}(u) \cap S_i) \cup (N_{i}(v) \cap S_i)$, the random variables $\{ X_{c, e'}\}_{c}$ are negatively associated. Next, note that the color picked by any edge $e' \in S_i$ in round $i$ is independent of the color picked by a different edge $e'' \in S_i \setminus \{e'\}$ in round $i$. Hence, part (1) of \Cref{NA-closure} implies that the random variables $\{X_{c, e'}\}_{c, e'}$ are also negatively associated. Finally, note that  $X_c = 1 - \max_{e' \in (N_{i}(u) \cap S_i) \cup (N_{i}(v) \cap S_i)} \{ X_{c,e'} \}$ for all colors $c \in P_i(e)$. Accordingly, part (2) of \Cref{NA-closure} implies that the random variables $\{ X_c\}_{c \in P_i(e)}$ are negatively associated. This concludes the proof of the claim.
\end{proof}

\begin{cla}
\label{cl:key-lem-1-proof-2}
We have: $\Pr\left[ |P_{i+1}(e)| =  \E \left[ |P_{i+1}(e)| \right] \pm 50 \sqrt{\Delta \ln n}  \right] \geq 1 - 1/n^{2000}$.
\end{cla}

\begin{proof}
Note that $|P_{i+1}(e)| = \sum_{c \in P_i(e)} X_c$, where $|P_i(e)| \leq (1+\eps^2)\Delta$ and each $X_c$ is a $0/1$ random variable. Since the random variables $\{X_c\}$ are negatively associated according to~\Cref{cl:key-lem-1-proof-1}, from~\Cref{hoeffding-NA} we now infer that $\Pr\left[ |P_{i+1}(e)| =  \E \left[ |P_{i+1}(e)| \right] \pm 50 \sqrt{\Delta \ln n}  \right] \geq 1 - 1/n^{2000}$.
\end{proof}

\begin{cor} 
\label{cor:cl:key-lem-1-proof-1}
We have: 
$\Pr\left[ |P_{i+1}(e)| =  (1-\eps)^{2i} \cdot \left(1 \pm \gamma_{i+1}\right) \cdot \Delta  \right] \geq 1 - 1/n^{2000}$. 
\end{cor}

\begin{proof}
Consider any $M = \E\left[ |P_{i+1}(e)| \right] \pm 50 \sqrt{\Delta \ln n}$. Observe that:
\begin{align}
M = & \left(1-\eps \pm (4\eps \gamma_i + \eps^2)\right)^2 \cdot |P_{i}(e)| \pm 50 \sqrt{\Delta \ln n} && (\text{by~(\ref{eq:key-lem-1-proof-1})}) \nonumber \\
= & \left(\left(1-\eps \pm (4\eps \gamma_i + \eps^2)\right)^2 \pm \eps^2\right) \cdot |P_{i}(e)| && (\text{by~\Cref{cl:technical:1}}) \nonumber  \\
= & \left((1-\eps)^2 + (4\eps\gamma_i +\eps^2)^2 \pm 2 (4\eps\gamma_i + \eps^2) \pm \eps^2 \right) \cdot |P_{i}(e)| \nonumber \\
= & \left( (1-\eps)^2 \pm (12 \eps \gamma_i + 12 \eps^2) \right) \cdot |P_{i}(e)| && (\text{by~(\ref{eq:eps}) and~\Cref{cor:gamma}}) \nonumber \\
= & (1-\eps)^2 \cdot \left(1 \pm (24 \eps \gamma_i + 24 \eps^2)\right) \cdot (1-\eps)^{2(i-1)} \cdot (1\pm \gamma_i) \cdot \Delta && (\text{by~(\ref{eq:eps}) and~(\ref{eq:key-lem-1-proof-100})}) \nonumber \\
= & (1-\eps)^{2i} \cdot \left(1 \pm \left((1+48\eps) \gamma_i + 48\eps^2\right) \right)  \cdot \Delta && (\text{by~(\ref{eq:eps}) and~\Cref{cor:gamma}}) \nonumber \\
= & (1-\eps)^{2i} \cdot (1\pm \gamma_{i+1}) \cdot \Delta. &&  (\text{by~(\ref{eq:t}) and~(\ref{eq:gamma:2})}) \nonumber
\end{align}
To summarize, we have derived that if $|P_{i+1}(e)| =  \E\left[ |P_{i+1}(e)| \right] \pm 50 \sqrt{\Delta \ln n}$, then it must be the case that $|P_{i+1}(e)|  = (1-\eps)^{2i} \cdot (1 \pm \gamma_{i+1}) \cdot \Delta$. The corollary now follows from~\Cref{cl:key-lem-1-proof-2}.
\end{proof}

\medskip
\noindent {\bf Wrapping up the proof of \Cref{key-lem}:}

\smallskip
\noindent \Cref{key-lem}  follows from~\Cref{cor:cl:key-lem-1-proof-1} and a union bound over all the edges $e \in E_{i+1}$.

\subsection{Proof  of \Cref{key-lem-c}}
\label{sec:key-lem-c}

Throughout this section, we fix the bits $r_{< i} \cup r_i^{(\text{edges})}$ which ensure the occurrence of the event $\e_i \cap \C_i \cap \B_i$. To ease notations, henceforth we  refrain from repeatedly stating that we are conditioning on the bits $r_{< i} \cup r_i^{(\text{edges})}$. However, the reader should keep in mind that we are relying upon this conditioning for the rest of~\Cref{sec:key-lem-c}.

We  first bound the expected value of $|N_{i+1,c}(v)|$ for a given  $(c, v) \in [(1+\eps^2)\Delta] \times V$. Next, we  show that w.h.p.~the value of $|N_{i+1,c}(v)|$ does not deviate too far away from its expectation. Finally, we  take a union bound over all $(c, v) \in [(1+\eps^2)\Delta] \times V$ to argue that  $\C_{i+1}$ occurs w.h.p.

\medskip
\noindent {\bf Calculating the expected value of $|N_{i+1,c}(v)|$ for a given $(c, v) \in [(1+\eps^2)\Delta] \times V$:} 

\smallskip
\noindent  First, note that $N_{i+1, c}(v) \subseteq N_{i, c}(v) \setminus S_i$. Consider any edge $e'=(u,v) \in N_{i,c}(v) \setminus S_i$, and let $\Gamma_{e'}$ be the event that the edge   $e'$  belongs to the set $N_{i+1, c}(v)$. Our immediate goal is to calculate the value of $\Pr[\Gamma_{e'}]$. Towards this end, we first recall that the bits $r_{< i} \cup r_i^{(\text{edges})}$  we are conditioning upon ensure the occurrence of the event $\e_i \cap \C_i \cap \B_i$. Hence, we have:
\begin{eqnarray}
\label{eq:key-lem-2-proof-900}
|P_i(e)| & = &  (1-\epsilon)^{2(i-1)}\cdot (1\pm \gamma_i) \cdot \Delta \text{ for all edges } e \in E_i. \\
|N_{i, c}(x)|  & = &  (1-\epsilon)^{2(i-1)}\cdot (1\pm \gamma_i) \cdot \Delta \text{ for all nodes } x \in V. \label{eq:key-lem-2-proof-901} \\
\label{eq:key-lem-2-proof-902}
|N_{i, c}(x) \cap S_i|  & = & (\eps \pm \eps^2) \cdot (1-\epsilon)^{2(i-1)}\cdot (1\pm \gamma_i) \cdot \Delta \text{ for all nodes } x \in V.
\end{eqnarray}
The event $\Gamma_{e'}$ occurs iff  no edge $(u, w) \in N_{i, c}(u) \cap S_i$ picks the color $c$ in round $i$. 
Hence, from~(\ref{eq:key-lem-2-proof-900}),~(\ref{eq:key-lem-2-proof-902}) and~\Cref{cl:technical:2}, we now derive that:
\begin{align}
		 \Pr[\Gamma_{e'}] 
		 & = \prod_{w \in N_{i,c}(u) \cap S_i}\left(1-\frac{1}{|\palette_i(u,w)|}\right) \nonumber \\
		 & = \left(1-\frac{1}{(1-\epsilon)^{2(i-1)}\cdot (1\pm \gamma_i) \cdot \Delta}\right)^{(\eps\pm \eps^2) \cdot (1-\epsilon)^{2(i-1)} \cdot (1\pm \gamma_i) \cdot \Delta} \nonumber \\
		& = 1 - \eps \pm (4\eps\gamma_i +\eps^2)    \label{eq:key-lem-2-proof-903}
	\end{align}
Since we have already conditioned on the event $\B_i$, we get:
 \begin{equation}
 \label{eq:key-lem-2-proof-904}
  |N_{i, c}(v) \setminus S_i| = |N_{i, c}(v)| - |N_{i, c}(v) \cap S_i| = (1-\eps \pm \eps^2) \cdot |N_{i, c}(v)|.
  \end{equation} 
  From~(\ref{eq:key-lem-2-proof-903}) and~(\ref{eq:key-lem-2-proof-904}), together with linearity of expectation, we now derive that:
\begin{eqnarray}
\E\left[ |N_{i+1, c}(v)| \right] & = & \sum_{e' \in N_{i, c}(v) \setminus S_i} \Pr[\Gamma_{e'}]  \nonumber \\
& = & (1-\eps\pm (4\eps \gamma_i + \eps^2)) \cdot |N_{i, c}(v) \setminus S_i|   \nonumber \\
& = & (1-\eps \pm (4\eps \gamma_i + \eps^2)) \cdot (1-\eps\pm \eps^2) \cdot |N_{i,c}(v)| \nonumber \\
& = & \left(1-\eps \pm (4\eps \gamma_i + \eps^2)\right)^2 \cdot |N_{i,c}(v)|. \label{eq:key-lem-2-proof-1}
\end{eqnarray}
	
\medskip
\noindent {\bf Deriving a concentration bound on $|N_{i+1,c}(v)|$ for a given  $(c, v) \in [(1\pm \eps^2)\Delta] \times V$:}

\smallskip
\noindent We first identify the  sampled edges in round $i$ that are responsible for determining   which edges from $N_{i, c}(v) \setminus S_i$ will end up being included in $N_{i+1, c}(v)$. Towards this end, we define $T_i(v) := \bigcup_{u \in V : (u, v) \in N_{i, c}(v) \setminus S_i} N_{i, c}(u) \cap S_i$. Observe that if an edge $(u, u') \in T_i(v)$, where $(u, v) \in N_{i, c}(v) \setminus S_i$, picks the color $c$ in round $i$, then $(u, v) \notin N_{i+1, c}(v)$. Conversely, if an edge $(u, v) \in N_{i, c}(v) \setminus S_i$ ends up not being part of $N_{i+1, c}(v)$, then  some edge $(u, u') \in T_i(v)$ must pick the color $c$ in round $i$. For each edge $e \in T_i(v)$, define an indicator random variable $X_e \in \{0, 1\}$ that is set to one iff the edge $e$ picks color $c$ in round $i$. Clearly,  the random variables $\{X_e\}$ are mutually independent. Since  $|N_{i+1, c}(v)|$ is  completely determined by the random variables $\{ X_e \}, e \in T_i(v)$, we write  $|N_{i+1, c}(v)| := f(X)$, where $X \in \{0,1\}^{|T_i(v)|}$  follows the joint distribution of the random variables $\{X_e\}$.  We now prove a concentration bound on $f(X)$.

\begin{cla}
\label{cl:nodecolor:1}
The function $f$ satisfies the Lipschitz property (\Cref{def:bounded:differences}) with  $d_e = 2$, $e \in T_i(v)$. Furthermore, for  each edge $e \in T_i(v)$, let $X_{-\{e\}} \in \{0,1 \}^{|T_i(v)|-1}$ denote the values of all the remaining  variables $\{X_{e'}\}_{e' \in T_i(v) \setminus \{e\}}$. Then for all $e \in T_i(v)$ and  $X_{-\{e\}} \in \{0,1\}^{|T_i(v)| - 1}$, we have:
$$var\left[ f(X) \ \Big| \ X_{-\{e\}} \right] \leq  \lambda_e, \text{ where } \lambda_e := \frac{8}{(1-\eps)^{2(i-1)}  \Delta}.$$
\end{cla}

\begin{proof}
Throughout the proof, we fix any edge $e = (u, u') \in T_i(v)$ and  the color picked by every other edge $e' \in T_i(v) \setminus \{e\}$ in round $i$, which determine the value of $X_{-\{e\}}$. Let $Z_e = \{ (v, w) \in N_{i, c}(v) \setminus S_i : w \in \{u, u'\}\}$ denote the set of edges in $N_{i, c}(v) \setminus S_i$ that are adjacent to the edge $e$. 

When we are trying to figure out which edges from $N_{i+1, c}(v) \setminus S_i$ will end up being part of $N_{i+1, c}(v)$, observe that the color picked by  $e$  can change the {\em fate} of only the edges in $Z_e$. Indeed, if the edge $e$ picks the color $c$, then the edges in $Z_e$ will  {\em not} be included in $N_{i+1, c}(v)$. In contrast, if the edge $e$ picks some color $c' \in P_i(e) \setminus \{c\}$, then the edges in $Z_e$ can potentially be included in $N_{i+1, c}(v)$.\footnote{In this event, whether or not an edge in $Z_e$ is actually included in $N_{i+1, c}(v)$ will depend on $X_{-\{e\}}$.}   The fate of every other edge $e'' \in (N_{i, c}(v) \setminus S_i) \setminus Z_e$ is completely determined by $X_{-\{e\}}$. 

Since $|Z_e| \leq 2$,  the function $f$ satisfies the Lipschitz property with  $d_e = 2$,  $e \in T_i(v)$. Furthermore, since the edge $e$ picks the color $c \in P_i(e)$ with probability $1/|P_i(e)|$, we conclude that:
\begin{eqnarray*}
\Var \left( f(X) \ \Big| \ X_{-\{e\}} \right) & \leq & \frac{2^2}{|P_i(e)|}   \\
& = & \frac{4}{(1-\eps)^{2(i-1)} (1\pm \gamma_i) \Delta} \qquad \qquad \ \ \ (\text{by~(\ref{eq:key-lem-2-proof-900})}) \\
& \leq & \frac{8}{(1-\eps)^{2(i-1)}  \Delta}. \qquad \qquad \qquad \qquad (\text{by~(\ref{eq:eps}) and~\Cref{cor:gamma}})
\end{eqnarray*} 
\end{proof}

\begin{cla}
\label{cor:cl:key-lem-2-proof-1}
We have: 
$\Pr\left[ |N_{i+1, c}(v)| =  \E\left[ |N_{i+1, c}(v)| \right] \pm 50 \sqrt{\Delta \ln n} \right] \geq 1 - 1/n^{600}$.
\end{cla}

\begin{proof}
We derive an upper bound on the size of the set $T_i(v)$. 
\begin{align}
|T_i(v)| \leq & \sum_{u \in V : (u, v) \in N_{i, c}(v) \setminus S_i} |N_{i, c}(u) \cap S_i| \nonumber \\
= & |N_{i, c}(v) \setminus S_i| \cdot (\eps\pm \eps^2) \cdot (1-\eps)^{2(i-1)} \cdot (1\pm \gamma_i) \cdot \Delta &&  (\text{by~(\ref{eq:key-lem-2-proof-902})}) \nonumber \\
\leq & |N_{i, c}(v) \setminus S_i| \cdot 4 \eps \cdot (1-\eps)^{2(i-1)} \cdot \Delta && (\text{by~(\ref{eq:eps})}) \nonumber \\
\leq & 8\eps \cdot (1-\eps)^{2(i-1)} \cdot \Delta^2. \label{eq:upperbound:tiv}
\end{align}
From~(\ref{eq:eps}),~(\ref{eq:upperbound:tiv}) and \Cref{cl:nodecolor:1}, we infer that:
\begin{eqnarray}
\label{eq:nodecolor:500}
\sum_{e \in T_i(v)} \lambda_e \leq 64 \eps \cdot \Delta \leq \Delta.
\end{eqnarray}
Recall that $|N_{i+1, c}(v)| = f(X)$, where $X$ is drawn from the joint distribution of  mutually independent random variables $\{X_e\}_{e \in T_i(v)}$. Hence, from~(\ref{eq:nodecolor:500}), \Cref{cl:nodecolor:1} and \Cref{lm:bernstein:new}, we get:
\begin{eqnarray*}
\Pr\left[ |N_{i+1, c}(v)| =  \E\left[ |N_{i, c}(v)|\right] \pm 50 \sqrt{\Delta \ln n} \right] & \geq & 1 - 2 \cdot \exp\left( - \frac{2500 \cdot \Delta \log n}{2\Delta + (2/3) \cdot 100 \cdot \sqrt{\Delta \log n}}\right) \\ 
& \geq & 1 - 2\cdot \exp\left( - \frac{2500 \cdot \Delta \log n}{4\Delta} \right) \ \ \  (\text{by~(\ref{eq:eps})})\\
& \geq & 1 - 1/n^{600}.
\end{eqnarray*}
This concludes the proof of the claim.
\end{proof}

\begin{cor}
\label{cor:nodecolor:100}
We have: $\Pr\left[ |N_{i+1, c}(v)| =  (1-\eps)^{2i} \cdot (1 \pm \gamma_{i+1}) \cdot \Delta \right] \geq 1 - 1/n^{600}$.
\end{cor}

\begin{proof}
Consider any $M = \E\left[ |N_{i+1, c}(v)| \right] \pm 50 \sqrt{\Delta \ln n}$. Observe that:
\begin{align}
M = & \left(1-\eps \pm (4\eps \gamma_i + \eps^2)\right)^2 \cdot |N_{i,c}(v)| \pm 50 \sqrt{\Delta \ln n} && (\text{by~(\ref{eq:key-lem-2-proof-1})}) \nonumber \\
= & \left(\left(1-\eps \pm (4\eps \gamma_i + \eps^2)\right)^2 \pm \eps^2\right) \cdot |N_{i,c}(v)| &&  (\text{by~\Cref{cl:technical:1}}) \nonumber  \\
= & \left((1-\eps)^2 + (4\eps\gamma_i +\eps^2)^2 \pm 2 (4\eps\gamma_i + \eps^2) \pm \eps^2 \right) \cdot |N_{i,c}(v)| \nonumber \\
= & \left( (1-\eps)^2 \pm (12 \eps \gamma_i + 12 \eps^2) \right) \cdot |N_{i,c}(v)| && (\text{by~(\ref{eq:eps}) and~\Cref{cor:gamma}}) \nonumber \\
= & (1-\eps)^2 \cdot \left(1 \pm (24 \eps \gamma_i + 24 \eps^2)\right) \cdot (1-\eps)^{2(i-1)} \cdot (1\pm \gamma_i) \cdot \Delta \qquad  &&   (\text{by~(\ref{eq:eps}) and~(\ref{eq:key-lem-2-proof-901})}) \nonumber \\
 = & (1-\eps)^{2i} \cdot \left(1 \pm \left((1+48\eps) \gamma_i + 48\eps^2\right) \right)  \cdot \Delta &&  (\text{by~(\ref{eq:eps}) and~\Cref{cor:gamma}}) \nonumber \\
  = & (1-\eps)^{2i} \cdot (1\pm \gamma_{i+1}) \cdot \Delta. && 
  (\text{by~(\ref{eq:t}) and~(\ref{eq:gamma:2})}) \nonumber
\end{align}
To summarize, we have derived that if $|N_{i+1, c}(v)| =  \E\left[ |N_{i+1, c}(v)| \right] \pm 50 \sqrt{\Delta \ln n}$, then it must be the case that $|N_{i+1, c}(v)|  = (1-\eps)^{2i} \cdot (1 \pm \gamma_{i+1}) \cdot \Delta$. The corollary now follows from~\Cref{cor:cl:key-lem-2-proof-1}.
\end{proof}

\medskip
\noindent {\bf Wrapping up the proof of \Cref{key-lem-c}:}

\smallskip
\noindent \Cref{key-lem-c} follows from~\Cref{cor:nodecolor:100} and a union bound over all pairs $(c, v) \in  [(1+\eps^2)\Delta] \times V$.

\subsection{Proof of \Cref{degree-failures}}
\label{sec:degree-failures}

Recall the discussion on the random bits $r_{<i}, r_i^{(\text{edges})}$ and $r_i^{(\text{colors})}$  from~\Cref{sec:notations:key}.  We will prove the lemma stated below. \Cref{degree-failures}  follows from~\Cref{lm:degree-failures:key} and a union bound over all  $v \in V$.

\begin{lem}
\label{lm:degree-failures:key}
Fix any instantiation of the bits $r_{< i} \cup r_i^{(\text{edges})}$ which ensure the occurrence of  the event $\e_i \cap \C_i \cap \B_i$. Fix any node $v \in V$, and let $F_i(v) = N_i(v) \cap F_i$ denote the set of failed edges in round $i$ that are incident on $v$. Then we have:
$\Pr\left[ |F_i(v)| \leq 9\epsilon^2 \Delta \, \Big| \,  r_{< i} \cup r_i^{(\text{edges})} \right] \geq 1 - 1/n^{305}$.
\end{lem}

The rest of~\Cref{sec:degree-failures} is devoted to the proof of~\Cref{lm:degree-failures:key}. We fix the bits $r_{< i} \cup r_i^{(\text{edges})}$ which ensure the occurrence of the event $\e_i \cap \C_i \cap \B_i$. To ease notations, henceforth we  refrain from repeatedly stating that we are conditioning on the bits $r_{< i} \cup r_i^{(\text{edges})}$. However, the reader should keep in mind that we are relying upon this conditioning for the rest of~\Cref{sec:degree-failures}.

 \medskip
 \noindent {\bf A classification of failed edges:}
 
\smallskip
\noindent  Let $F_i^{(1)}(v) = \{ (u, v) \in F_i(v) : c(v', v) = c(u, v) \text{ for some } (v', v) \in N_i(v) \cap S_i\}$ denote the set of edges $(u,v)$ that fails in round $i$ because of the following reason: Some other edge  incident on $v$ picks the same color as $(u, v)$ in round $i$. Let $F_i^{(2)}(v) = F_i(v) \setminus F_i^{(1)}(v)$ denote the set of remaining edges incident on $v$ that fails in round $i$. An edge $(u, v) \in N_i(v) \cap S_i$ belongs to the set $F_i^{(2)}(v)$ iff no other edge  $e_v \in N_i(v) \cap S_i$  picks the same color as $(u,v)$ in round $i$, {\em and}  at least one edge  $e_u \in N_i(u) \cap S_i$  picks the same color as $(u, v)$ in round $i$.  We say that a failed edge $e \in F_i(v)$ is of type-(1) iff $e \in F_i^{(1)}(v)$ and it is of type-(2) iff $e \in F_i^{(2)}(v)$.  We will separately prove concentration bounds on the number of failed type-(1) and type-(2) edges incident on $v$. As  $|F_i(v)| = | F_i^{(1)}(v) | + | F_i^{(2)}(v) |$,  this will lead to the desired concentration bound on $|F_i(v)|$.

\medskip
\noindent {\bf Deriving a concentration bound on $| F_i^{(1)}(v) |$:} 

\smallskip 
\noindent \Cref{cl:type1:expectation} bounds  the expected value of $| F_i^{(1)}(v)|$.  \Cref{cl:type1:concentration} shows that w.h.p.~$|F_i^{(1)}(v)|$ does not deviate too far away from its expectation. \Cref{cl:type1:concentration} follows from \Cref{cl:type1:expectation} and \Cref{cl:type1:concentration}.
\begin{cla}
\label{cl:type1:expectation}
We have $\E\left[ |F_i^{(1)}(v)| \right]  \leq 4\eps^2 \Delta$.
\end{cla}

\begin{proof}
Note that $F_i^{(1)}(v) \subseteq N_i(v) \cap S_i$. Consider any edge $e = (u, v) \in N_i(v) \cap S_i$. Our immediate goal is to bound the probability that this edge $e$ does {\em not} belong to $F_i^{(1)}(v)$, conditioned on it picking a given color $c \in P_i(e)$ in round $i$. Towards this end, we first recall that the bits $r_{< i} \cup r_i^{(\text{edges})}$  we are conditioning upon ensure the occurrence of the event $\e_i \cap \C_i \cap \B_i$. Hence, we have: 
\begin{eqnarray}
\label{eq:type1:900}
 |P_i(e')| & = & (1-\epsilon)^{2(i-1)}\cdot (1\pm \gamma_i) \cdot \Delta \text{ for all edges } e' \in E_i. \\
 \label{eq:type1:901}
 |N_{i, c}(x)| & = &  (1-\epsilon)^{2(i-1)}\cdot (1\pm \gamma_i) \cdot \Delta \text{ for all nodes } x \in V. \\
 \label{eq:type1:902}
  |N_{i, c}(x) \cap S_i|  & = &  (\eps \pm \eps^2) \cdot (1-\epsilon)^{2(i-1)}\cdot (1\pm \gamma_i) \cdot \Delta \text{ for all nodes } x \in V.
 \end{eqnarray} 
Conditioned on the edge $e$ picking the color $c \in P_i(e)$, it does not belong to $F_i^{(1)}(v)$ iff none of the edges $e' \in (N_{i, c}(v) \cap S_i) \setminus \{e\}$ picks the  same color $c$ in round $i$. Hence, we  derive that:
\begin{eqnarray}
\Pr\left[e\not\in F_i^{(1)}(v) \ \Big | \  c(e)=c \right] & = & \prod_{e'\in (N_{i,c}(v)\cap S_i)  \setminus \{e\}} \left(1-\frac{1}{|\palette_i(e')|}\right) \nonumber \\
	& = & \left(1-\frac{1}{(1-\eps)^{2(i-1)} \cdot (1\pm \gamma_i) \cdot \Delta}\right)^{|N_{i,c}(v) \cap S_i| - 1} \nonumber \\
	& \geq & \left(1-\frac{1}{(1-\eps)^{2(i-1)} \cdot (1\pm \gamma_i) \cdot \Delta}\right)^{|N_{i,c}(v) \cap S_i| } \nonumber \\
	 & = & \left(1-\frac{1}{(1-\eps)^{2(i-1)} \cdot (1\pm \gamma_i) \cdot \Delta}\right)^{(\eps\pm \eps^2) \cdot (1-\eps)^{2(i-1)} \cdot (1\pm \gamma_i) \cdot \Delta} \nonumber \\
	& \geq & 1 - \eps - (4\eps \gamma_i + \eps^2) \nonumber \\
	& \geq & 1 - 2\eps. \label{eq:type1:903}
\end{eqnarray}
In the derivation above, the second step follows from~(\ref{eq:type1:900}), the fourth step follows from~(\ref{eq:type1:902}), the fifth step follows from~\Cref{cl:technical:2}, and the last step follows from~(\ref{eq:eps}) and~\Cref{cor:gamma}.
Since~(\ref{eq:type1:903}) holds for every color $c \in P_i(e)$, we conclude that:
\begin{eqnarray}
\label{eq:lm:degree-failures:key-1}
\Pr\left[e \notin F_i^{(1)}(v)\right] \geq 1 - 2\eps \text{ for every edge } e \in N_i(v) \cap S_i.
\end{eqnarray}
Now, applying linearity of expectation, we get:
\begin{align}
\E\left[ |F_i^{(1)}(v)| \right] = & \sum_{e \in N_i(v) \cap S_i} \Pr\left[e \in F_i^{(1)}(v)\right]  \nonumber \\
\leq & |N_i(v) \cap S_i| \cdot (2 \eps) && (\text{by~(\ref{eq:lm:degree-failures:key-1})}) \nonumber \\
\leq &  (2\eps  \Delta) \cdot (2\eps) && (\text{by~(\ref{eq:eps}),~(\ref{eq:type1:902}) and~\Cref{cor:gamma}}) \nonumber \\
= &  4\eps^2 \Delta. \nonumber 
\end{align}
This concludes the proof of the claim.
\end{proof}

\begin{cla}
\label{cl:type1:concentration}
We have: $\Pr\left[ | F_i^{(1)}(v) | \leq \E\left[ | F_i^{(1)}(v) | \right] + 50 \sqrt{\Delta \ln n}\right] \geq 1 - 1/n^{310}$.
\end{cla}

\begin{proof}
For each edge $e \in N_i(v) \cap S_i$, define a random variable $X_e \in P_i(e)$ whose value indicates the color picked by the edge $e$ in round $i$. The quantity $|F_i^{(1)}(v) |$ is a function of the random variables $\{X_e\}, e \in N_i(v) \cap S_i$, and the random variables $\{X_e\}$ themselves are mutually independent. 

We  claim  that the function $|F_i^{(1)}(v)|$ satisfies the Lipschitz property (see \Cref{def:bounded:differences}) with constants $d_e = 4$, $e \in N_i(v) \cap S_i$. To see why the claim holds, consider any given edge $e \in N_i(v) \cap S_i$ and fix the colors picked by every other edge $e' \in (N_{i}(v) \cap S_i) \setminus \{e\}$ in round $i$. Fix any two distinct colors $c_1, c_2 \in P_i(e)$. Let $n_{c_1}$ and $n_{c_2}$ respectively denote the number of edges $e' \in (N_i(v) \cap S_i) \setminus \{e\}$ that have picked color $c_1$ and color $c_2$ in round $i$.  Now, consider the following two scenarios: 
\begin{itemize}
\item (1) The edge $e$ picks  the color $c_1 \in P_i(e)$ in round $i$.
\item (2) The edge $e$ picks the color $c_2 \in P_i(e) \setminus \{c_1\}$ in round $i$. 
\end{itemize} 
As we switch from scenario (1) to scenario (2), the number of type-(1) failed edges in $N_i(v) \cap S_i$ that pick color $c_2$ increases by  $\phi(n_{c_2}+1) - \phi(n_{c_2})$; where  $\phi(y) = y$ if $y \geq 2$, and $\phi(y) = 0$ otherwise. Similarly, the number of type-(1) failed edges in $N_i(v) \cap S_i$ that pick color $c_1$ decreases by $\phi(n_{c_1}+1) - \phi(n_{c_1})$. In contrast, the number of type-(1) failed edges in $N_i(v) \cap S_i$ that pick any  color $c \notin [\Delta] \setminus \{c_1, c_2\}$ remains unchanged.  Thus, as we switch from scenario (1) to scenario (2), the absolute value of the change in $| F_i^{(1)}(v)|$ is given by:
\begin{eqnarray*}
\left| \left\{ \phi(n_{c_2}+1) - \phi(n_{c_2}) \right\} - \left\{ \phi(n_{c_1}+1) - \phi(n_{c_1}) \right\} \right| & \leq & \left| \phi(n_{c_2}+1) - \phi(n_{c_2}) \right| + \left|  \phi(n_{c_1}+1) - \phi(n_{c_1})  \right| \\
& \leq & 2 + 2 = 4.
\end{eqnarray*}

We therefore conclude that  $|F_i^{(1)}(v) |$ is a function of mutually independent random variables $\{X_e\}$ that satisfy the Lipschitz property (see \Cref{def:bounded:differences}) with constants $d_e = 4$, $e \in N_i(v) \cap S_i$. Since $\sum_{e \in N_i(v) \cap S_i} d_e^2 = 16 \cdot |N_i(v) \cap S_i| \leq 16 \Delta$, applying \Cref{lm:bounded:differences} we get:
$$\Pr\left[ | F_i^{(1)}(v) | \leq \E\left[ | F_i^{(1)}(v) | \right] + 50 \sqrt{\Delta \ln n}\right] \geq 1 - 1/n^{310}.$$ 
This concludes the proof of the claim.
\end{proof}

\begin{cor}
\label{cor:type1:concentration}
We have: $\Pr\left[ | F_i^{(1)}(v) | \leq 4 \eps^2 \Delta + 50 \sqrt{\Delta \ln n}\right] \geq 1 - 1/n^{310}$.
\end{cor}

\begin{proof}
Follows from \Cref{cl:type1:expectation} and \Cref{cl:type1:concentration}.
\end{proof}

\medskip
\noindent {\bf Deriving a concentration bound on $| F_i^{(2)}(v)|$:} 

\smallskip 
\noindent While analyzing this quantity, for technical reasons we first condition upon the colors picked by all the edges incident on $v$ that are sampled in round $i$. After this conditioning, we bound the expected value of $|F_i^{(2)}(v)|$ in \Cref{cl:type2:expectation}, and in \Cref{cl:type2:concentration} we show that w.h.p.~$|F_i^{(2)}(v)|$ does not deviate too far away from its expectation. \Cref{cor:type2:concentration} follows from \Cref{cl:type2:expectation} and \Cref{cl:type2:concentration}. 

\begin{cla}
\label{cl:type2:expectation}
Fix any color $c^*(e) \in P_i(e)$ for every edge  $e \in N_i(v) \cap S_i$. Let $\Gamma^*$ be the event which occurs iff every edge $e \in N_i(v) \cap S_i$ picks the color $c^*(e)$ in round $i$.  Then we have:  $$\E\left[ |F_i^{(2)}(v)| \ \Big| \ \Gamma^*  \right] \leq 4 \eps^2 \Delta.$$
\end{cla}

\begin{proof}
The proof is  analogous to the proof of \Cref{cl:type1:expectation}. Nevertheless, for the sake of completeness, we reproduce the same chain of reasoning in its entirety.

The event $\Gamma^*$ completely determines the set $F_i^{(1)}(v)$. Furthermore, we have  $F_i^{(2)}(v) \subseteq (N_i(v) \cap S_i) \setminus F_i^{(1)}(v)$. Consider any edge $e = (u, v) \in (N_i(v) \cap S_i) \setminus F_i^{(1)}(v)$, which picks the color $c^*(e)$ in round $i$. Our immediate goal is to bound the probability that this edge $e$ does {\em not} belong to the set $F_i^{(2)}(v)$. Towards this end, we first recall that the bits $r_{< i} \cup r_i^{(\text{edges})}$  we have already conditioned upon ensure the occurrence of the event $\e_i \cap \C_i \cap \B_i$. Hence, we have:
\begin{eqnarray}
\label{eq:type2:900}
 |P_i(e')| & = & (1-\epsilon)^{2(i-1)}\cdot (1\pm \gamma_i) \cdot \Delta \text{ for all edges } e' \in E_i. \\
 \label{eq:type2:901}
 |N_{i, c^*(e)}(x)| & = &  (1-\epsilon)^{2(i-1)}\cdot (1\pm \gamma_i) \cdot \Delta \text{ for all nodes } x \in V. \\
 \label{eq:type2:902}
  |N_{i, c^*(e)}(x) \cap S_i|  & = &  (\eps \pm \eps^2) \cdot (1-\epsilon)^{2(i-1)}\cdot (1\pm \gamma_i) \cdot \Delta \text{ for all nodes } x \in V.
 \end{eqnarray}
The  edge $e = (u, v)$  does not belong to the set $F_i^{(2)}(v)$ iff no edge $e' \in (N_{i, c^*(e)}(u) \cap S_i) \setminus \{e\}$ picks the  color $c(e') = c^*(e)$ in round $i$. Hence, we  derive that:
\begin{eqnarray}
\Pr\left[e\not\in F_i^{(2)}(v) \ \Big| \  \Gamma^* \right] & = & \prod_{e'\in (N_{i,c^*(e)}(u)\cap S_i)  \setminus \{e\}} \left(1-\frac{1}{|\palette_i(e')|}\right) \nonumber \\
	& = & \left(1-\frac{1}{(1-\eps)^{2(i-1)} \cdot (1\pm \gamma_i) \cdot \Delta}\right)^{|N_{i,c^*(e)}(u) \cap S_i| - 1} \nonumber \\
	& \geq & \left(1-\frac{1}{(1-\eps)^{2(i-1)} \cdot (1\pm \gamma_i) \cdot \Delta}\right)^{|N_{i,c^*(e)}(u) \cap S_i| } \nonumber \\
	 & = & \left(1-\frac{1}{(1-\eps)^{2(i-1)} \cdot (1\pm \gamma_i) \cdot \Delta}\right)^{(\eps\pm \eps^2) \cdot (1-\eps)^{2(i-1)} \cdot (1\pm \gamma_i) \cdot \Delta} \nonumber \\
	& \geq & 1 - \eps - (4\eps \gamma_i + \eps^2) \nonumber \\
	& \geq & 1 - 2\eps. \label{eq:type2:903}
\end{eqnarray}
In the derivation above, the second step follows from~(\ref{eq:type2:900}), the fourth step follows from~(\ref{eq:type2:902}), the fifth step follows from~\Cref{cl:technical:2},  the last step follows from~(\ref{eq:eps}) and~\Cref{cor:gamma}. Thus, we have:
\begin{eqnarray}
\label{eq:lm:degree-failures:key-1:type2}
\Pr\left[e \notin F_i^{(2)}(v) \ \Big| \ \Gamma^* \right] \geq 1 - 2\eps \text{ for every edge } e \in (N_i(v) \cap S_i) \setminus F_i^{(1)}(v).
\end{eqnarray}
Now, applying linearity of expectation, we get:
\begin{eqnarray}
\E\left[ |F_i^{(2)}(v)|  \Big| \ \Gamma^* \right] & = & \sum_{e \in (N_i(v) \cap S_i) \setminus F_i^{(1)}(v)} \Pr\left[e \in F_i^{(2)}(v) \ \Big| \ \Gamma^*  \right]  \nonumber \\
& \leq & |N_i(v) \cap S_i| \cdot (2 \eps) \qquad \qquad \qquad (\text{by~(\ref{eq:lm:degree-failures:key-1:type2})}) \nonumber \\
& \leq &  (2\eps  \Delta) \cdot (2\eps) \qquad \qquad \qquad \qquad \  (\text{by~(\ref{eq:eps}),~(\ref{eq:type2:902}) and~\Cref{cor:gamma}}) \nonumber \\
& = &  4\eps^2 \Delta. \nonumber 
\end{eqnarray}
This concludes the proof of the claim.
\end{proof}

\begin{cla}
\label{cl:type2:concentration}
Fix any color $c^*(e) \in P_i(e)$ for every edge  $e \in N_i(v) \cap S_i$. Let $\Gamma^*$ be the event which occurs iff every edge $e \in N_i(v) \cap S_i$ picks the color $c^*(e)$ in round $i$.  Then we have:  
$$\Pr\left[ |F_i^{(2)}(v)| \leq \E\left[ |F_i^{(2)}(v)| \right] + 50  \sqrt{\Delta \ln n}  \ \Big| \ \Gamma^*  \right] \geq 1 - 1/n^{2000}.$$
\end{cla}

\begin{proof}
Let $W_i(v) = (N_i(u) \cap S_i) \setminus F_i^{(1)}(v)$ denote the set of edges incident on $v$ that get sampled in round $i$ and do {\em not} end up being  type-(1) failures under the event $\Gamma^*$. By definition, all the edges in $W_i(v)$ receive distinct colors under the event $\Gamma^*$, and we have $F_i^{(2)}(v)  \subseteq W_i(v)$. Let $Z_i(v) = \{ (u, u') \in S_i \setminus N_i(v):  \text{either } u \in W_i(v) \text{ or } u' \in W_i(v)\}$ denote the set of  edges sampled in round $i$ that are not themselves incident on $v$, but are neighbors of at least one edge in $W_i(v)$. Note that the sets $W_i(v)$ and $Z_i(v)$, along with the colors picked by the edges in $W_i(v)$, are completely determined by the bits $r_{<i} \cup r_i^{\text{(edges)}}$ and the event $\Gamma^*$ we are conditioning upon. On the other hand, the colors picked by the edges in $Z_i(v)$ are yet to be determined. In particular, each edge in $Z_i(v)$ will pick a color uniformly at random from the set $P_i(e)$, independently of the other edges, and these colors will uniquely  determine the set $F_i^{(2)}(v)$. We can think of the edges  $e \in Z_i(v)$ picking their colors in round $i$ as a ``balls and bins" process, as described below. 

There is a  ball for each edge  $e \in Z_i(v)$, a  bin for each color $c \in [(1+\eps^2)\Delta]$, and an additional  dummy bin $\perp$. Consider any edge $e \in Z_i(v)$, and let $W_i(v, e) \subseteq W_i(v)$ be the set of edges $e' \in W_i(v)$ that share an endpoint with $e$. Note that $|W_i(v,e)| \in \{1, 2\}$. Suppose that the edge $e$ picks a color $c \in P_i(e)$ in round $i$. Then the ball for  $e$ gets thrown into the bin for   $c$ iff some edge $e' \in W_i(v, e)$ picked the same color $c$ under the event $\Gamma^*$; otherwise the ball for $e$ goes to the dummy bin. 

Let $\chi_i(v) = \bigcup_{e \in W_i(v)} \{ c^*(e)\}$ denote the set of colors picked by the edges $e \in W_i(v)$ in round $i$, under the event $\Gamma^*$.  For each color $c \in \chi_i(v)$, define an indicator random variable $Y_c \in \{0,1\}$ that is set to one iff  the bin for the color $c$ is nonempty (has at least one ball in it) at the end of the balls and bins process described above. As each edge $e \in Z_i(v)$ picks a color independently of the other edges in $Z_i(v)$, the balls get thrown into the bins independently of each other. Hence, \Cref{empty-bins} implies that the random variables $\{Y_c\}, c \in \chi_i(v),$ are negatively associated. Since $|\chi_i(v)| \leq (1+\eps^2)\Delta$, from \Cref{hoeffding-NA} we get:
\begin{equation}
\label{eq:type2:concentration:99}
\Pr\left[ \sum_{c \in \chi_i(v)} Y_c \leq \E\left[  \sum_{c \in \chi_i(v)} Y_c \right] + 50 \sqrt{\Delta \ln n}  \ \Big| \ \Gamma^*  \right] \geq 1 - 1/n^{2000}.
\end{equation}

Recall that  no two edges in $W_i(v)$ pick the same color under the event $\Gamma^*$. Accordingly, there is a natural bijective mapping $g : \chi_i(v) \rightarrow W_i(v)$, where $g(c)$ is the unique edge  in $W_i(v)$ that picked the color $c \in \chi_i(v)$ under the event $\Gamma^*$. For each color $c \in \chi_i(v)$, we have $Y_c = 1$ iff $g(c) \in F_i^{(2)}(v)$. Since $F_i^{(2)}(v) \subseteq W_i(v)$, we infer that $\sum_{c \in \chi_i(v)} Y_c = |F_i^{(2)}(v)|$.  The claim now follows from~(\ref{eq:type2:concentration:99}).
\end{proof}

\begin{cor}
\label{cor:type2:concentration}
We have: $\Pr\left[ | F_i^{(2)}(v) | \leq 4 \eps^2 \Delta + 50 \sqrt{\Delta \ln n}\right] \geq 1 - 1/n^{2000}$.
\end{cor}

\begin{proof}
Fix any color $c^*(e) \in P_i(e)$ for all  $e \in N_i(v) \cap S_i$. Let $\Gamma^*$ be the event which occurs iff every edge $e \in N_i(v) \cap S_i$ picks the color $c^*(e)$ in round $i$.  \Cref{cl:type2:expectation} and \Cref{cl:type2:concentration} imply that:
\begin{equation}
\label{eq:type2:concentration:100}
\Pr\left[ | F_i^{(2)}(v) | \leq 4 \eps^2 \Delta + 50\sqrt{\Delta \ln n} \ \Big| \ \Gamma^* \right] \geq 1 - 1/n^{2000}.
\end{equation}
Since the bound in~(\ref{eq:type2:concentration:100}) holds for every possible $\Gamma^*$, the corollary follows.
\end{proof}

\medskip
\noindent {\bf Wrapping up the proof of \Cref{lm:degree-failures:key}:} 

\smallskip
\noindent
Applying  a union bound over \Cref{cor:type1:concentration} and \Cref{cor:type2:concentration}, we get:
$$\Pr\left[ | F_i^{(1)}(v) | + | F_i^{(2)}(v) | \leq 8 \eps^2 \Delta + 100 \sqrt{\Delta \ln n} \right] \geq 1 - 1/n^{305}.$$
Now, \Cref{lm:degree-failures:key} follows from the following two observations: (a) $|F_i(v)| = |F_i^{(1)}(v)| + |F_i^{(2)}(v)|$, and (b) 
$100  \sqrt{\Delta \ln n} \leq \eps^2 \Delta$ according to~(\ref{eq:eps}).

\section{Missing Proofs from \Cref{sec:recourse}}
\label{app:recourse}

\subsection{Proof of \Cref{main:recourse:analysis:100}}
\label{app:recourse:1}

We start by proving the key property for our dynamic algorithm, which implies that it implements the first phase of \Cref{alg:model-agnostic}. We re-state this lemma here for ease of reference.

\recourseimplement*

To prove the above, we must prove that each edge $e\in S^{(t)}_i$ has a tentative color $c^{(t)}$ sampled u.a.r.~from $\palette_{i}^{(t)}(e)$ (or null if $\palette_{i}^{(t)}(e)$), independently of all other edges in $S^{(t)}_i$. To prove this, we will prove a stronger claim, whereby all edges in $S_i$ (including those not in $E^{(t)}$) have a tentative color $c^{(t)}$ sampled u.a.r.~from from $\palette_{i}^{(t)}$, independent of all edges in $S_i$.

To prove the above, we will think of $c^{(t)}$ and $\palette_{i}^{(t)}$ restricted to $S_i$ as $|S_i|$-dimensional vectors, indexed by (fake or real) edges $e\in S_i$. Moreover, we let $P(e) = \{0\}$ indicate $P(e)=\emptyset$, using $c^{(t)}(e)=0$ to mean $c^{(t)}(e) = null$. With this notation, we are ready to state our key lemma.

\begin{lem}\label{independent-colors-in-_S_i}
	For all times $t$, real vector $\vec{c}$ and vector of real sets $\vec{P}$ for which 	$\Pr[{P^{(t)}_i} = \vec{P}] \neq 0$, 
	$$\Pr[c^{(t)} = \vec{c} \mid {P^{(t)}_i}=\vec{P}] = \begin{cases}
	\prod_{e\in S_i} \frac{1}{|P(e)|} & c(e)\in P(e)\,\,\, \forall e\in S_i \\
	0 & \textrm{else}.
	\end{cases}$$
\end{lem}

\Cref{independent-colors-in-_S_i} implies  that the coloring of edges in $E^{(t)}\cap \bigcup_i (S_i \setminus F_i)$ shares the same joint distribution as its counterpart in $G^{(t)}$. This immediately leads to \Cref{main:recourse:analysis:100}. 

\begin{proof}
	We prove this lemma by induction on $t\geq 0$, referring to the hypothesis we wish to prove for time $t$ by $\calH^t$. Hypothesis $\calH^0$ holds trivially. We now prove that $\calH^{t-1}$ implies $\calH^t$.
	
	First, we consider the same probability when further conditioning on the previous palette, $\palette^{(t-1)}_i$, for which we have that for all vector of real sets $\vec{Q}$,
	\begin{align}
	& \Pr[c^{(t)} = \vec{c} \mid {P^{(t)}_i}=\vec{P}, {P^{(t-1)}_i}=\vec{Q}] \nonumber \\
	= & \sum_{\vec{d}}\Pr[c^{(t)} = \vec{c} \mid c^{(t-1)} = \vec{d}, {P^{(t)}_i}=\vec{P}, {P^{(t-1)}_i}=\vec{Q}]\cdot \Pr[c^{(t-1)} = \vec{d}\mid {P^{(t)}_i}=\vec{P}, {P^{(t-1)}_i}=\vec{Q}]. \label{}
	\end{align}
	We begin by simplifying the term $\Pr[c^{(t-1)} = \vec{d}\mid {P^{(t)}_i}=\vec{P}, {P^{(t-1)}_i}=\vec{Q}]$, in the following claim.
	\begin{cla}
	For any real vector $\vec{d}$ and vector of real set $\vec{Q}$, we have 
	$$\Pr[c^{(t-1)} = \vec{d} \mid {P^{(t)}_i}=\vec{P}, {P^{(t-1)}_i}=\vec{Q}] = \begin{cases}
	\prod_{e\in S_i} \frac{1}{|Q(e)|} & d(e)\in Q(e)\,\,\, \forall e\in S_i \\
		0 & \textrm{else},
	\end{cases}$$
	\end{cla}
	\begin{proof}
		By the oblivious adversary assumption, the changes between $P^{(t)}_i$ and $P^{(t-1)}_i$---the color palettes of round $i$ of times $t$ and $t-1$---are independent of the candidate colors of $S_i$ edges in time $t-1$, namely values $c^{(t-1)}$.
		Therefore, the joint distributions  
		$\left(c^{(t-1)} \mid P^{(t)}_i = \vec{P}, P^{(t-1)}_i = \vec{Q}\right)$
		and 
		$\left(c^{(t-1)} \mid P^{(t-1)}_i = \vec{Q} \right)$
		are identically distributed. The claim thus follows by Hypothesis $\calH^{t-1}$.
	\end{proof}
	
	By the above claim and the preceding equation, we have that 
	\begin{align}
	& \Pr[c^{(t)} = \vec{c} \mid {P^{(t)}_i}=\vec{P}, {P^{(t-1)}_i}=\vec{Q}] \nonumber \\
	= & \sum_{\vec{d}}\Pr[c^{(t)} = \vec{c} \mid c^{(t-1)} = \vec{d}, {P^{(t)}_i}=\vec{P}, {P^{(t-1)}_i}=\vec{Q}]\cdot \prod_{e\in S_i} \frac{1}{|Q(e)|}.\label{conditional-colors}
	\end{align}
	On the other hand, if we denote for each edge $e\in S_i$ the color sets $Q_\cap(e) := Q(e)\cap P(e)$ and $Q_\setminus(e) := Q(e)\setminus P(e)$, then by the definition of \Cref{alg:recourse-change-color}, we have that
	\begin{align}
	\Pr[c^{(t)} = \vec{c} \mid c^{(t-1)} = \vec{d}, {P^{(t)}_i}=\vec{P}, {P^{(t-1)}_i}=\vec{Q}] = \prod_{e:\, d(e) \in Q_\setminus(e)} \frac{1}{|P(e)|}\prod_{e:\,  c(e)=d(e)\in Q_\cap(e)}\frac{|Q_\cap(e)|}{|P(e)|}.\label{conditional-colors-summand}
		\end{align}
	Summing \Cref{conditional-colors-summand} over all vectors $\vec{d}$, \Cref{conditional-colors} then implies that, if we let every edge in the following equation implicitly belong to $S_i$, then
	\begin{align*}
	\Pr[c^{(t)} = \vec{c} \mid {P^{(t)}_i}=\vec{P}, {P^{(t-1)}_i}=\vec{Q}] & = \sum_{\vec{d}}\prod_{e:\, d(e) \in Q_\setminus(e)} \frac{1}{|P(e)|}\prod_{e:\, c(e)=d(e)\in Q_\cap(e)}\frac{|Q_\cap(e)|}{|P(e)|} \prod_{e} \frac{1}{|Q(e)|} \\
	& = \prod_{e} \frac{|Q_{\setminus} (e)| + |Q_\cap(e)|}{|P(e)|} \prod_{e\in S_i} \frac{1}{|Q(e)|} \\
	& = \prod_{e} \frac{1}{|P(e)|},
	\end{align*}
	where the last step relied on $Q_\cap(e)$ and $Q_\setminus(e)$ being a partition of $Q(e)$, which  therefore implies $|Q_{\setminus}(e)| + |Q_\cap(e)| = |Q(e)|$.
	
	But then, by total probability over $\palette^{(t-1)}_i$, we therefore have that Hypothesis $\calH^t$ holds, as
	\begin{align*}
	\Pr[c^{(t)} = \vec{c} \mid {P^{(t)}_i}=\vec{P}] & = \sum_{\vec{Q}} 	\Pr[c^{(t)} = \vec{c} \mid {P^{(t)}_i}=\vec{P}, {P^{(t-1)}_i}=\vec{Q}]\cdot \Pr[{P^{(t-1)}_i}=\vec{Q}] = \prod_{e\in S_i}\frac{1}{|P(e)|}.
	\qedhere
	\end{align*}
\end{proof}

\subsection{Proof of \Cref{main:lm:recourse:2}}
\label{app:recourse:2}
In this section, we will need to use some key concepts  from the analysis of  \Cref{alg:model-agnostic}. In particular, before proceeding any further, the reader will find it useful to review all of \Cref{sec:notations:key}, and the statement of \Cref{cor:key-lem:key} from \Cref{app:model:agnostic:main:proof}.

Throughout this section, we use the symbol $\e_i^{(t)}$ to denote the random event $\e_i$ for the output of the dynamic algorithm at the end of the $t^{th}$ update. To be very specific, the event $\e_i^{(t)}$ occurs iff $|P_i^{(t)}(e)| = (1-\eps)^{2(i-1)} \cdot (1\pm \gamma_i) \cdot \Delta$ for all edges  $e \in E_i^{(t)}$. Similarly, we will use the notations $\C_i^{(t)}$ and $\e_i^{(t)}$ to respectively denote the corresponding events $\e_i$ and $\C_i$ at the end of the $t^{th}$ update. 

We now  define the following random event that will play an important role in this section.
$$\mathcal{Z} := \e_i^{(t-1)} \cap \C_i^{(t-1)} \cap \B_i^{(t-1)} \cap \e_i^{(t)} \cap \C_i^{(t)} \cap \B_i^{(t)}.$$

\begin{lem}
\label{main:lm:recourse:3}
We have $\E[|D_i| \mid \mathcal{Z}] \leq 5\eps \cdot \left(1+ \E[|D_{< i}| \mid \mathcal{Z}]\right)$.
\end{lem}

The complete proof of \Cref{main:lm:recourse:3} appears in \Cref{app:lm:recourse:3}. For now, we focus on showing how  \Cref{main:lm:recourse:3} almost immediately leads us to the proof of \Cref{main:lm:recourse:2}.

\begin{lem}
\label{main:lm:recourse:4}
We have $\E[|D_i| \mid \neg \mathcal{Z}]  \cdot \Pr[\neg \mathcal{Z}] \leq \eps$.
\end{lem}

\begin{proof}
\Cref{cor:key-lem:key} and \Cref{main:recourse:analysis:100}  imply that:
\begin{eqnarray}
\label{eq:recourse:500}
\Pr[\e_i^{(t-1)} \cap \C_i^{(t-1)} \cap \B_i^{(t-1)}] & \geq & 1 - 1/n^{400}. \\
\label{eq:recourse:501}
\Pr[\e_i^{(t)} \cap \C_i^{(t)} \cap \B_i^{(t)}] & \geq & 1 - 1/n^{400}. 
\end{eqnarray}
Taking a union bound over~(\ref{eq:recourse:500}) and~(\ref{eq:recourse:501}), we get:
\begin{equation}
\label{eq:recourse:502}
\Pr[\neg \mathcal{Z}] \leq 2/n^{400}.
\end{equation}
Since $|D_i| \leq n^2$ with probability one,~(\ref{eq:recourse:502}) implies that: $\E[|D_i| \mid \neg \mathcal{Z}]  \cdot \Pr[\neg \mathcal{Z}] \leq n^2 \cdot (2/n^{400}) \leq \eps$.
\end{proof}

\noindent{\bf Proof of \Cref{main:lm:recourse:2}:} From \Cref{main:lm:recourse:3} and \Cref{main:lm:recourse:4}, we now derive that:
\begin{eqnarray*}
\E[|D_{i}|] & = & \E[|D_{i}| \mid \mathcal{Z}] \cdot \Pr[\mathcal{Z}] + \E[|D_{i}| \mid \neg \mathcal{Z}] \cdot \Pr[\neg \mathcal{Z}] \\
& \leq & 5 \eps \cdot \left(1+ \E[|D_{< i}| \mid \mathcal{Z}]\right) \cdot \Pr[\mathcal{Z}] + \eps \\ 
& = & 5 \eps \cdot \Pr[\mathcal{Z}] + 5 \eps \cdot \E[|D_{< i}| \mid \mathcal{Z}] \cdot \Pr[\mathcal{Z}]  + \eps \\
& \leq & 6\eps + 6\eps \cdot \E[|D_{<i}|].
\end{eqnarray*}
In the above derivation, the first inequality follows from \Cref{main:lm:recourse:3} and \Cref{main:lm:recourse:4}.

\subsubsection{Proof of \Cref{main:lm:recourse:3}}
\label{app:lm:recourse:3}

Let $e^*$ be the edge that gets inserted/deleted during the $t^{th}$ update, that is, the sets $E^{(t)}$ and $E^{(t-1)}$ differ only in the edge $e^*$. Now, define the set $D^*_{< i}$ as follows. 
\begin{eqnarray*}
D^*_{< i }  := \begin{cases}  D_{< i} \cup \{e^*\} \text{ if } i(e^*) <  i; \\
 D_{< i} \text{ otherwise.}
 \end{cases}
\end{eqnarray*}

\begin{wrapper}
For  the rest of the proof, we will fix (condition upon)  the following {\em ``critical"} random variables:
\begin{itemize}
\item The set of edges $S_{\leq i}^* = \{ e \in E^{(t)} \cup E^{(t-1)} : i(e) \leq i\}$ that get assigned to levels at most $i$.
\item The tentative colors $c^{(t-1)}(e), c^{(t)}(e)$ for every edge $e \in S_{\leq i-1}^* = \{ e' \in S_{\leq i}^* : i(e') \leq i-1\}$.
\end{itemize}
\end{wrapper}

Note that once we fix the critical random variables, the occurrence of the event $\mathcal{Z}$ and  the contents of the sets $D_{<i}^*, D_{<i}$ are completely determined. Our main goal will be to show that:
\begin{wrapper}
\begin{equation}
\label{eq:recourse:key:100}
\E[|D_i| ] \leq 5\eps \cdot  |D^*_{<i}|, \text{ if we fix the critical random variables in such a way that event } \mathcal{Z} \text{ occurs.}
\end{equation}
\end{wrapper}
Since $|D_{<i}^*| \leq 1 + |D_{<i}|$, \Cref{main:lm:recourse:3} follows from~(\ref{eq:recourse:key:100}) if we take expectations on both sides  while still conditioning upon the event $\mathcal{Z}$.  Accordingly, for the rest of the proof,  we fix the critical random variables in such a way that the event $\mathcal{Z}$ occurs, and focus on proving~(\ref{eq:recourse:key:100}). For ease of exposition, however, from now on we will refrain from explicitly stating again and again that we have conditioned on the critical random variables.

\medskip
Our first task is to identify, for each edge $e \in D_i$, the set of edges in $D^*_{< i}$ that are  responsible for $e$ changing its tentative color during the $t^{th}$ update. Recall that an edge $e \in D_i$ changes its tentative color either in line (2) of \Cref{alg:recourse-change-color} (in this case we say that the edge $e$ is of  type-(a)), or in line (6) of \Cref{alg:recourse-change-color} (in this case we say that the edge $e$ is of type-(b)).  We now define:
 \begin{eqnarray}
 R(e) & := & \{ e' \in N(e) \cap D^*_{< i} : c^{(t)}(e') = c^{(t-1)}(e)\} \text{ for every edge } e \in D_i \text{ of type-(a)}. \label{eq:recourse:type:100} \\
 R(e) & := & \{ e' \in N(e) \cap D^*_{< i} : c^{(t-1)}(e') = c^{(t)}(e)\} \text{ for every edge } e \in D_i \text{ of type-(b).} \label{eq:recourse:type:101}
 \end{eqnarray}
 
 To appreciate the rationale behind~(\ref{eq:recourse:type:100}), consider any edge $e \in D_i$ of type-(a). For such an edge $e$,  line (1) of \Cref{alg:recourse-change-color} implies that $c^{(t-1)}(e) \in P_i^{(t-1)}(e) \setminus P_i^{(t)}(e)$. This  happens only if  $e$ has some neighboring edges $e'$ in $D_{<i}^*$ with $c^{(t)}(e') = c^{(t-1)}(e)$. These neighboring edges are precisely the ones that are {\em responsible} for $e$ becoming part of the set $D_i$. 
 
 Similarly, to appreciate the rationale behind~(\ref{eq:recourse:type:101}), consider any edge $e \in D_i$ of type-(b). For such an edge $e$,  line (5) of \Cref{alg:recourse-change-color} implies that $c^{(t)}(e) \in P_i^{(t)}(e) \setminus P_i^{(t-1)}(e)$.\footnote{Since conditioned on $\mathcal{Z}$, the event $\e_i^{(t)}$ occurs with probability one, we have $P_i^{(t)}(e) \neq \emptyset$ and hence  $c^{(t)}(e) \neq \text{null}$.} This happens only if  $e$ has some neighboring edges $e'$ in $D_{<i}^*$ with $c^{(t-1)}(e') = c^{(t)}(e)$. These neighboring edges are precisely the ones that are {\em responsible} for $e$ becoming part of the set $D_i$. 
 
 The discussion above leads us to the following observation.  
 
\begin{obs}
\label{cor:recourse:nonempty}
We have $R(e)  \neq \emptyset$ for all edges $e \in D_i$.
\end{obs}

The sets $\{R(e)\}_{e\in D_i}$ tell us, for every $e \in D_i$, which edges in $D_{<i}^*$ are responsible for $e$ changing its tentative color during the $t^{th}$ update. We now want to look at the same picture, but from the point of view of the edges in $D_{<i}^*$. Accordingly, we define:
\begin{eqnarray}
\label{eq:recourse:type:102}
R^{-1}(e') & = &  \{ e \in D_i : e' \in R(e) \} \text{ for all } e' \in D_{<i}^*. \\
R^{-1}_a(e') & = &  \{ e \in D_i  :e \text{ is of  type-(a)  and } e' \in R(e) \} \text{ for all } e' \in D_{<i}^*. \label{eq:recourse:type:103} \\
R^{-1}_b(e') & = &  \{ e \in D_i  :e \text{ is of  type-(b)  and } e' \in R(e) \} \text{ for all } e' \in D_{<i}^*. \label{eq:recourse:type:104}
\end{eqnarray}
We are  now  derive a useful claim.
\begin{cla}
\label{cla:recourse:nonempty:101}
We have: $\E[|D_i|] \leq  \sum_{e' \in D_{< i}^*} \left( \E[|R^{-1}_a(e')|] + \E[|R^{-1}_b(e')|] \right)$.
\end{cla}

\begin{proof}
\Cref{cor:recourse:nonempty} implies that $D_i = \bigcup_{e' \in D_{<i}^*} R^{-1}(e')$. Since each $R^{-1}(e')$ is partitioned into two subsets $R^{-1}_a(e')$ and $R^{-1}_b(e')$, we get:
$$|D_i| \leq  \sum_{e' \in D_{< i}^*} \left( |R^{-1}_a(e')| + |R^{-1}_b(e')|\right).$$
The claim follows if we take expectations on both sides of the above inequality. 
\end{proof}

\Cref{cla:recourse:nonempty:101} suggests the following  natural approach for coming up an upper bound on $\E[|D_i|]$: Separately upper bound $\E[|R^{-1}_a(e')|]$ and $\E[|R^{-1}_b(e')|]$ for every edge $e' \in D_{<i}^*$, and then sum up the resulting inequalities. We implement this approach in the next two claims.

\begin{cla}
\label{cla:recourse:1}
We have $\E[|R_a^{-1}(e')| ] \leq (5/2) \cdot \eps$ for every edge $e' \in D_{< i}^*$. 
\end{cla}

\begin{proof}
Fix any edge $e' = (u, v) \in D_{< i}^*$, and let $c = c^{(t)}(e')$ denote its tentative color in $G^{(t)}$. Note that for an edge $e \in E^{(t-1)} \cap E^{(t)}$ to be included in the set $R^{-1}_a(e')$, it  needs to: (1) be a neighbor of $e'$ with $i(e) = i$ and (2)  have the color $c$ in its palette for round $i$ in $G^{(t-1)}$. This happens iff   $e \in \left( (N_{i, c}^{(t-1)}(u) \cap S_i^{(t-1)}) \cup (N_{i, c}^{(t-1)}(v) \cap S_i^{(t-1)}) \right) \setminus \{e^*\}$.\footnote{Note that this specific set is completely determined once we condition upon the critical random variables.} For such an edge $e$, we have:
\begin{eqnarray}
\Pr[e \in R^{-1}_a(e') ]  & = & \Pr[ c^{(t-1)}(e) = c ]  =  \frac{1}{|P_{i}^{(t-1)}(e)|} \nonumber \\
& \leq & \frac{1}{(1-\eps)^{2(i-1)} \cdot (1-\gamma_i) \cdot \Delta} \label{eq:main:recourse:2}
\end{eqnarray}
In the above derivation, the second equality holds because of \Cref{main:recourse:analysis:100}. The last inquality holds since conditioned on  $\mathcal{Z}$, the event $\e_i^{(t-1)}$ occurs with probability one. Now, summing~(\ref{eq:main:recourse:2}) over all the relevant edges $e$, we get:
\begin{eqnarray}
\E[|R^{-1}_a(e')| ] & = & \sum_{e} \Pr[e \in R^{-1}_a(e') ] \nonumber \\
& \leq & \frac{|N_{i, c}^{(t-1)}(u) \cap S_i^{(t-1)}| + |N_{i, c}^{(t-1)}(v) \cap S_i^{(t-1)}|}{(1-\eps)^{2(i-1)} \cdot (1-\gamma_i) \cdot \Delta} \nonumber  \\
& \leq & \frac{2 \cdot (\eps+\eps^2) \cdot (1-\eps)^{2(i-1)} \cdot (1+\gamma_i) \cdot \Delta}{(1-\eps)^{2(i-1)} \cdot (1-\gamma_i) \cdot \Delta} \nonumber \\
& \leq & 2(\eps+\eps^2) (1+3\gamma_i) \leq (5/2) \cdot \eps. \nonumber
\end{eqnarray}
The second inequality follows from the fact that conditioned on $\mathcal{Z}$,  the event $\B_i^{(t-1)} \cap \C_i^{(t-1)}$ occurs with probability one. The last two inequalities follow from~(\ref{eq:eps}) and~\Cref{cor:gamma}. 
\end{proof}

\begin{cla}
\label{cla:recourse:2}
We have $\E[|R^{-1}_b(e')| ] \leq (5/2) \cdot \eps$ for every edge $e' \in D_{< i}^*$. 
\end{cla}

\begin{proof}
Fix any edge $e' = (u, v) \in D_{< i}^*$, and let $c = c^{(t-1)}(e')$ denote its tentative color in $G^{(t-1)}$. For an edge $e \in E^{(t-1)} \cap E^{(t)}$ to be included in the set $R^{-1}_b(e')$, it needs to be a neighbor of $e'$ with $i(e) = i$ and have the color $c$ in its palette for round $i$ in $G^{(t)}$. This happens iff   $e \in \left( (N_{i, c}^{(t)}(u) \cap S_i^{(t)}) \cup (N_{i, c}^{(t)}(v) \cap S_i^{(t)}) \right) \setminus \{e^*\}$.\footnote{Note that this specific set is completely determined once we condition upon the critical random variables.} For such an edge $e$, we have:
\begin{eqnarray}
\Pr[e \in R^{-1}_b(e') ] & = &  \Pr[ c^{(t)}(e) = c ]  = \frac{1}{|P_{i}^{(t)}(e)|}  \nonumber \\
& \leq & \frac{1}{(1-\eps)^{2(i-1)} \cdot (1-\gamma_i) \cdot \Delta}. \label{eq:main:recourse:3}
\end{eqnarray}
In the above derivation, the second equality holds because of \Cref{main:recourse:analysis:100}. The last inequality holds since conditioned on  $\mathcal{Z}$, the event $\e_i^{(t)}$ occurs with probability one. Now, summing~(\ref{eq:main:recourse:3}) over all the relevant edges $e'$, we get:
\begin{eqnarray}
\E[|R^{-1}_b(e')| ] & = & \sum_{e} \Pr[e \in R^{-1}_b(e') ] \nonumber \\
& \leq & \frac{|N_{i, c}^{(t)}(u) \cap S_i^{(t)}| + |N_{i, c}^{(t)}(v) \cap S_i^{(t)}|}{(1-\eps)^{2(i-1)} \cdot (1-\gamma_i) \cdot \Delta} \nonumber  \\
& \leq & \frac{2 \cdot (\eps+\eps^2) \cdot (1-\eps)^{2(i-1)} \cdot (1+\gamma_i) \cdot \Delta}{(1-\eps)^{2(i-1)} \cdot (1-\gamma_i) \cdot \Delta} \nonumber \\
& \leq & 2 (\eps+\eps^2) (1+3\gamma_i) \leq (5/2) \cdot \eps. \nonumber
\end{eqnarray}
The second inequality follows from the fact that conditioned on $\mathcal{Z}$,  the event $\B_i^{(t)} \cap \C_i^{(t)}$ occurs with probability one. The last two inequalities follow from~(\ref{eq:eps}) and~\Cref{cor:gamma}. 
\end{proof}

\noindent {\bf Proof of \Cref{main:lm:recourse:3}:} Recall that in order to prove the lemma, all we needed to do was to prove~(\ref{eq:recourse:key:100}). Now,~(\ref{eq:recourse:key:100}) follows from \Cref{cla:recourse:nonempty:101}, \Cref{cla:recourse:1} and \Cref{cla:recourse:2}.

\section{Useful concentration inequalities}\label{sec:prelims}

 We start by describing the standard Hoeffding bounds for mutually independent random variables. 

\begin{lem}[Hoeffding bounds]\label{hoeffding} 
	Let $X$ be the sum of $m$ mutually independent random variables $X_1,\dots,X_m$ with $X_i \in[a_i, b_i]$ for each $i\in [m]$. Then for all $t>0$, we have:
	$$\Pr[X \geq  \E[X] + t] \leq \exp\left(-\frac{2t^2}{\sum_i (b_i-a_i)^2}\right),$$		
	$$\Pr[X \leq  \E[X] - t] \leq \exp\left(-\frac{2t^2}{\sum_i (b_i-a_i)^2}\right).$$			
\end{lem}

Throughout this paper, we denote by $x\sim_R S$ a (uniformly) random sample x from the set $S$. One tool we will rely on crucially for our analysis is concentration inequalities of \emph{dependent} random variables. Specifically, we will study concentration of sums of \emph{negatively associated} variables.

\begin{Def}[\cite{khursheed1981positive,joag1983negative}]\label{def:NA}
	We say a joint distribution $(X_1,\dots,X_n)$ is \emph{negatively associated (NA)}, or that the variables $X_1,\dots,X_n$ are NA, if every two monotone increasing functions $f$ and $g$ defined on disjoint subsets of the variables in $\vec{X}$ are negatively correlated. That is,
	\begin{equation}\label{eq:NA}
		\E[f\cdot g] \leq \E[f]\cdot \E[g].
	\end{equation}
\end{Def}

A trivial example of NA is \emph{independent} random variables (for which \eqref{eq:NA} holds with equality.)
A more interesting, useful example of NA for our use is given by the following two propositions.

\begin{prop}[0-1 Principle \cite{dubhashi1996balls}]\label{0-1-NA}
	Let $X_1,\dots,X_n\in \{0,1\}$ be binary random variables such that $\sum_i X_i\leq 1$ always. Then, the joint distribution $(X_1,\dots,X_n)$ is NA.
\end{prop}

\begin{prop}[Permutation Distributions are NA (\cite{joag1983negative})]\label{lem:permutations}
	Let $x_1, \dots, x_n$ be $n$ values and let $X_1,\dots,X_n$ be random variables taking on all permutations of $(x_1,\dots,x_n)$ with equal probability. Then $X_1,\dots,X_n$ are NA.
\end{prop}

Permutation distributions prove useful in our context due to our study of random-order streams.

More elaborate NA distributions can be obtained from simple NA distributions as those given by propositions \ref{0-1-NA} and \ref{lem:permutations} via the following closure properties.
\begin{prop}[NA Closure Properties \cite{khursheed1981positive,joag1983negative,dubhashi1996balls}]\label{NA-closure}
	$\phantom{a}$
	\begin{enumerate}
		\item \label{P7_union} \underline{Independent union.}
		If the joint distributions $(X_1,\dots,X_n)$ and $(Y_1,\dots,Y_m)$ are both NA and independent of each other, then the joint distribution  $(X_1,\dots,X_n,Y_1,\dots,Y_m)$ is also NA.
		\item \label{P6_inc_funs} \underline{Function composition.}
		Let $f_1,\dots,f_k$ be monotone (all increasing or all decreasing) functions defined on disjoint subsets of the variables in $\vec{X}$. Then the joint distribution  $(f_1(\vec{X}),\dots,f_k(\vec{X}))$ is NA.
	\end{enumerate}
\end{prop}

An example NA distribution obtained using these closure properties are balls and bins processes.

\begin{prop}[Balls and Bins is NA]\label{balls-and-bins-NA}
	Suppose $m$ balls are thrown independently into one of $n$ bins (not necessarily u.a.r., and not necessarily i.i.d).
	Let $B_i$ be the number of balls placed in bin $i$ in this process. Then the joint distribution $(B_1,B_2,\dots,B_n)$ is NA.
\end{prop}
\begin{proof}
	For each $b\in [m]$ and $i\in [n]$, let $X_{b,i}$ be an indicator variable for ball $b$ landing in bin $i$. By the 0-1 Principle (\Cref{0-1-NA}), the variables $\{X_{b,i} \mid i\in [n]\}$ are NA. By closure of NA under independent union (\Cref{NA-closure}.\ref{P7_union}), as each ball is placed independently of all other balls, $\{X_{b,i} \mid b\in [m], i\in [n]\}$ are NA. Finally, by closure of NA under monotone increasing functions on disjoint subsets (\Cref{NA-closure}.\ref{P6_inc_funs}), the variables $B_i=\sum_{b} X_{b,i}$ are indeed NA.
\end{proof}

As the variables $N_i = \min\{1,B_i\}$, indicating whether bin $i$ is non-empty, are monotone increasing functions depending on disjoint subsets of the $B_i$ variables (specifically, singletons), we obtain the following corollary.

\begin{cor}\label{empty-bins}
	The indicator variables $N_i$ for bins being non-empty in a balls and bins process as in \Cref{balls-and-bins-NA}, are NA.
\end{cor}

A particularly useful property of NA variables is the applicability of Chernoff-Hoeffding type concentration inequalities to sums of NA variables $X_1,X_2,\dots,X_n$. This follows from monotonicity of the exponential function implying (by induction, using \Cref{def:NA}) that $\E[\exp(\lambda\cdot \sum_i X_i)]\leq \prod_i \E[\exp(\lambda\cdot x_i)]$ (see \cite{dubhashi1996balls}). This is the crucial first step of proofs of such tail bounds. 
In particular, we will make use of the following  tail bounds.

\begin{lem}[Chernoff bounds for NA variables \cite{dubhashi1996balls}]\label{chernoff-NA} 
	Let $X$ be the sum of NA random variables $X_1,\dots,X_m$ with $X_i \in[0, 1]$ for each $i\in [m]$. Then for all $\delta\in (0,1)$, and $\kappa\geq \E[X]$, 
	$$\Pr[X\leq (1-\delta) \cdot \E[X]] \leq \exp\left(-\frac{\E[X] \cdot \delta^2}{2}\right),$$		
	$$\Pr[X\geq (1+\delta) \cdot \kappa] \leq \exp\left(-\frac{\kappa \cdot \delta^2}{3}\right).
	$$	
\end{lem}

\begin{lem}[Hoeffding bounds for NA variables \cite{dubhashi1996balls}]\label{hoeffding-NA} 
	Let $X$ be the sum of $m$ NA random variables $X_1,\dots,X_m$ with $X_i \in[a_i, b_i]$ for each $i\in [m]$. Then for all $t>0$, we have:
	$$\Pr[X \geq  \E[X] + t] \leq \exp\left(-\frac{2t^2}{\sum_i (b_i-a_i)^2}\right),$$		
	$$\Pr[X \leq  \E[X] - t] \leq \exp\left(-\frac{2t^2}{\sum_i (b_i-a_i)^2}\right).$$			
\end{lem}

We next describe a  concentration inequality which is known as the {\em method of bounded differences}.

\begin{Def}\cite{Dubhashi2009}
\label{def:bounded:differences}
Consider $n$ sets $A_1, \ldots, A_n$ and a real-valued function $f : A_1 \times \cdots \times A_n \rightarrow \mathbf{R}$. The function $f$ satisfies the {\em Lipschitz property} with constants $\{d_i\}, i \in [n],$ iff $|f(\mathbf{a}) - f(\mathbf{a'})| \leq d_i$ whenever $\mathbf{a}$ and $\mathbf{a'}$ differ only in the $i^{th}$ co-ordinate, for all $i \in [n]$.
\end{Def}

\begin{lem}\cite{Dubhashi2009}
\label{lm:bounded:differences}
Consider a  function $f(X_1, \ldots, X_n)$ of $n$ mutually independent random variables $X_1, \ldots, X_n$ that   satisfy the Lipschitz property with  $\{d_i\mid i \in [n]\}$. Then for all $t > 0$, 
\begin{eqnarray*}
\Pr\left[f \geq \E[f] + t \right] & \leq &  \exp\left(-\frac{2t^2}{\sum_{i=1}^n d_i^2}\right), \\
\Pr\left[f \leq \E[f] - t \right] & \leq &  \exp\left(-\frac{2t^2}{\sum_{i=1}^n d_i^2}\right).
\end{eqnarray*}
\end{lem}

\begin{lem}
\label{lm:bernstein:new}
Consider a  function $f(X_1, \ldots, X_n)$ of $n$ mutually independent $0/1$ random variables $X_1, \ldots, X_n$ that   satisfy the Lipschitz property with  $\{d_i\}, i \in [n]$. For each $i \in [n]$, let $X_{-i} \in \{0,1\}^{n-1}$ denote the values taken up by every other random variable $\{X_j\}, j \in [n] \setminus \{i\}$. Furthermore, suppose that  $var\left[ f \mid X_{-i}\right] \leq \lambda_i$ for all $i \in [n]$ and all $X_{-i} \in \{0,1 \}^{n-1}$. Let $\lambda = \sum_{i=1}^n \lambda_i$, and $d = \max_{i \in [n]} \{d_i\}$. Then for all $t > 0$, we have: 
\begin{eqnarray*}
\Pr\left[f \geq \E[f] + t \right] & \leq &  \exp\left(-\frac{t^2}{2\lambda + (2/3) \cdot  t  d }\right), \\
\Pr\left[f \leq \E[f] - t \right] & \leq &  \exp\left(-\frac{t^2}{2\lambda + (2/3) \cdot  t  d }\right).
\end{eqnarray*}
\end{lem}

\begin{proof}
 The lemma follows from the method of bounded variances, as explained in Chapter 8.1 of \cite{Dubhashi2009} (in particular, the lemma follows from equation (8.5) in this chapter).
 \end{proof}

\bibliographystyle{acmsmall}
\bibliography{abb,ultimate}

\begin{thebibliography}{}

\bibitem[\protect\citeauthoryear{Aggarwal, Motwani, Shah, and Zhu}{Aggarwal
  et~al\mbox{.}}{2003}]{aggarwal2003switch}
{\sc Aggarwal, G.}, {\sc Motwani, R.}, {\sc Shah, D.}, {\sc and} {\sc Zhu, A.}
  2003.
\newblock Switch scheduling via randomized edge coloring.
\newblock In {\em Proceedings of the 44th Symposium on Foundations of Computer
  Science (FOCS)}. 502--512.

\bibitem[\protect\citeauthoryear{Alon and Spencer}{Alon and
  Spencer}{2004}]{alon2004probabilistic}
{\sc Alon, N.} {\sc and} {\sc Spencer, J.~H.} 2004.
\newblock {\em The probabilistic method}.
\newblock John Wiley \& Sons.

\bibitem[\protect\citeauthoryear{Bahmani, Mehta, and Motwani}{Bahmani
  et~al\mbox{.}}{2012}]{bahmani2012online}
{\sc Bahmani, B.}, {\sc Mehta, A.}, {\sc and} {\sc Motwani, R.} 2012.
\newblock Online graph edge-coloring in the random-order arrival model.
\newblock {\em Theory of Computing (conference version appeared in SODA
  2010)\/}~{\em 8,\/}~1, 567--595.

\bibitem[\protect\citeauthoryear{Bar-Noy, Motwani, and Naor}{Bar-Noy
  et~al\mbox{.}}{1992}]{bar1992greedy}
{\sc Bar-Noy, A.}, {\sc Motwani, R.}, {\sc and} {\sc Naor, J.} 1992.
\newblock The greedy algorithm is optimal for on-line edge coloring.
\newblock {\em Information Processing Letters (IPL)\/}~{\em 44,\/}~5, 251--253.

\bibitem[\protect\citeauthoryear{Barenboim and Maimon}{Barenboim and
  Maimon}{2017}]{barenboim2017fully}
{\sc Barenboim, L.} {\sc and} {\sc Maimon, T.} 2017.
\newblock Fully-dynamic graph algorithms with sublinear time inspired by
  distributed computing.
\newblock {\em Procedia Computer Science\/}~{\em 108}, 89--98.

\bibitem[\protect\citeauthoryear{Bernstein, Holm, and Rotenberg}{Bernstein
  et~al\mbox{.}}{2018}]{BernsteinHR18}
{\sc Bernstein, A.}, {\sc Holm, J.}, {\sc and} {\sc Rotenberg, E.} 2018.
\newblock Online bipartite matching with amortized replacements.
\newblock In {\em Proceedings of the Twenty-Ninth Annual {ACM-SIAM} Symposium
  on Discrete Algorithms (SODA)}, {A.~Czumaj}, Ed. 947--959.

\bibitem[\protect\citeauthoryear{Bhattacharya, Chakrabarty, Henzinger, and
  Nanongkai}{Bhattacharya et~al\mbox{.}}{2018}]{bhattacharya2018dynamic}
{\sc Bhattacharya, S.}, {\sc Chakrabarty, D.}, {\sc Henzinger, M.}, {\sc and}
  {\sc Nanongkai, D.} 2018.
\newblock Dynamic algorithms for graph coloring.
\newblock In {\em Proceedings of the 29th Annual ACM-SIAM Symposium on Discrete
  Algorithms (SODA)}. 1--20.

\bibitem[\protect\citeauthoryear{Bosek, Leniowski, Sankowski, and Zych}{Bosek
  et~al\mbox{.}}{2014}]{BosekLSZ14}
{\sc Bosek, B.}, {\sc Leniowski, D.}, {\sc Sankowski, P.}, {\sc and} {\sc Zych,
  A.} 2014.
\newblock Online bipartite matching in offline time.
\newblock In {\em 55th {IEEE} Annual Symposium on Foundations of Computer
  Science (FOCS)}. 384--393.

\bibitem[\protect\citeauthoryear{Chang, He, Li, Pettie, and Uitto}{Chang
  et~al\mbox{.}}{2018}]{chang2018complexity}
{\sc Chang, Y.-J.}, {\sc He, Q.}, {\sc Li, W.}, {\sc Pettie, S.}, {\sc and}
  {\sc Uitto, J.} 2018.
\newblock The complexity of distributed edge coloring with small palettes.
\newblock In {\em Proceedings of the 29th Annual ACM-SIAM Symposium on Discrete
  Algorithms (SODA)}. 2633--2652.

\bibitem[\protect\citeauthoryear{Charikar and Liu}{Charikar and
  Liu}{2021}]{charikar2021improved}
{\sc Charikar, M.} {\sc and} {\sc Liu, P.} 2021.
\newblock Improved algorithms for edge colouring in the w-streaming model.
\newblock In {\em Proceedings of the 4th Symposium on Simplicity in Algorithms
  (SOSA)}. To appear.

\bibitem[\protect\citeauthoryear{Cohen, Peng, and Wajc}{Cohen
  et~al\mbox{.}}{2019}]{cohen2019tight}
{\sc Cohen, I.~R.}, {\sc Peng, B.}, {\sc and} {\sc Wajc, D.} 2019.
\newblock Tight bounds for online edge coloring.
\newblock In {\em Proceedings of the 60th Symposium on Foundations of Computer
  Science (FOCS)}. 1--25.

\bibitem[\protect\citeauthoryear{Cohen and Wajc}{Cohen and
  Wajc}{2018}]{cohen2018randomized}
{\sc Cohen, I.~R.} {\sc and} {\sc Wajc, D.} 2018.
\newblock Randomized online matching in regular graphs.
\newblock In {\em Proceedings of the 29th Annual ACM-SIAM Symposium on Discrete
  Algorithms (SODA)}. 960--979.

\bibitem[\protect\citeauthoryear{Cole, Ost, and Schirra}{Cole
  et~al\mbox{.}}{2001}]{cole2001edge}
{\sc Cole, R.}, {\sc Ost, K.}, {\sc and} {\sc Schirra, S.} 2001.
\newblock Edge-coloring bipartite multigraphs in ${O}({E} \log {D})$ time.
\newblock {\em Combinatorica\/}~{\em 21,\/}~1, 5--12.

\bibitem[\protect\citeauthoryear{Duan, He, and Zhang}{Duan
  et~al\mbox{.}}{2019}]{duan2019dynamic}
{\sc Duan, R.}, {\sc He, H.}, {\sc and} {\sc Zhang, T.} 2019.
\newblock Dynamic edge coloring with improved approximation.
\newblock In {\em Proceedings of the 30th Annual ACM-SIAM Symposium on Discrete
  Algorithms (SODA)}. 1937--1945.

\bibitem[\protect\citeauthoryear{Dubhashi, Grable, and Panconesi}{Dubhashi
  et~al\mbox{.}}{1998}]{dubhashi1998near}
{\sc Dubhashi, D.}, {\sc Grable, D.~A.}, {\sc and} {\sc Panconesi, A.} 1998.
\newblock Near-optimal, distributed edge colouring via the nibble method.
\newblock {\em Theor. Comput. Sci.\/}~{\em 203,\/}~2, 225--252.

\bibitem[\protect\citeauthoryear{Dubhashi and Ranjan}{Dubhashi and
  Ranjan}{1996}]{dubhashi1996balls}
{\sc Dubhashi, D.} {\sc and} {\sc Ranjan, D.} 1996.
\newblock Balls and bins: A study in negative dependence.
\newblock {\em BRICS Report Series\/}~{\em 3,\/}~25.

\bibitem[\protect\citeauthoryear{Dubhashi and Panconesi}{Dubhashi and
  Panconesi}{2009}]{Dubhashi2009}
{\sc Dubhashi, D.~P.} {\sc and} {\sc Panconesi, A.} 2009.
\newblock {\em Concentration of Measure for the Analysis of Randomized
  Algorithms}.

\bibitem[\protect\citeauthoryear{Elkin, Pettie, and Su}{Elkin
  et~al\mbox{.}}{2014}]{elkin20142delta}
{\sc Elkin, M.}, {\sc Pettie, S.}, {\sc and} {\sc Su, H.-H.} 2014.
\newblock (2$\delta$—l)-edge-coloring is much easier than maximal matching in
  the distributed setting.
\newblock In {\em Proceedings of the 26th Annual ACM-SIAM Symposium on Discrete
  Algorithms (SODA)}. 355--370.

\bibitem[\protect\citeauthoryear{Feldkord, Feldotto, Gupta, Guruganesh, Kumar,
  Riechers, and Wajc}{Feldkord et~al\mbox{.}}{2018}]{feldkord2018fully}
{\sc Feldkord, B.}, {\sc Feldotto, M.}, {\sc Gupta, A.}, {\sc Guruganesh, G.},
  {\sc Kumar, A.}, {\sc Riechers, S.}, {\sc and} {\sc Wajc, D.} 2018.
\newblock Fully-dynamic bin packing with little repacking.
\newblock In {\em Proceedings of the 45th International Colloquium on Automata,
  Languages and Programming (ICALP)}. 51:1--51:24.

\bibitem[\protect\citeauthoryear{Gamlath, Kapralov, Maggiori, Svensson, and
  Wajc}{Gamlath et~al\mbox{.}}{2019}]{gamlath2019online}
{\sc Gamlath, B.}, {\sc Kapralov, M.}, {\sc Maggiori, A.}, {\sc Svensson, O.},
  {\sc and} {\sc Wajc, D.} 2019.
\newblock Online matching with general arrivals.
\newblock In {\em Proceedings of the 60th Symposium on Foundations of Computer
  Science (FOCS)}. 26--37.

\bibitem[\protect\citeauthoryear{Grable}{Grable}{1998}]{grable1998large}
{\sc Grable, D.~A.} 1998.
\newblock A large deviation inequality for functions of independent, multi-way
  choices.
\newblock {\em Combinatorics Probability and Computing\/}~{\em 7,\/}~1, 57--63.

\bibitem[\protect\citeauthoryear{Gu, Gupta, and Kumar}{Gu
  et~al\mbox{.}}{2013}]{GuGK13-steiner-tree}
{\sc Gu, A.}, {\sc Gupta, A.}, {\sc and} {\sc Kumar, A.} 2013.
\newblock The power of deferral: maintaining a constant-competitive steiner
  tree online.
\newblock In {\em ACM Symposium on Theory of Computing (STOC)}. 525--534.

\bibitem[\protect\citeauthoryear{Gupta, Krishnaswamy, Kumar, and
  Panigrahi}{Gupta et~al\mbox{.}}{2017}]{gupta2017online}
{\sc Gupta, A.}, {\sc Krishnaswamy, R.}, {\sc Kumar, A.}, {\sc and} {\sc
  Panigrahi, D.} 2017.
\newblock Online and dynamic algorithms for set cover.
\newblock In {\em Proceedings of the 49th Annual ACM Symposium on Theory of
  Computing (STOC)}. 537--550.

\bibitem[\protect\citeauthoryear{Gupta and Kumar}{Gupta and
  Kumar}{2014}]{GuptaK14}
{\sc Gupta, A.} {\sc and} {\sc Kumar, A.} 2014.
\newblock Online steiner tree with deletions.
\newblock In {\em Proceedings of the Twenty-Fifth Annual {ACM-SIAM} Symposium
  on Discrete Algorithms (SODA)}, {C.~Chekuri}, Ed. 455--467.

\bibitem[\protect\citeauthoryear{Gupta, Kumar, and Stein}{Gupta
  et~al\mbox{.}}{2014}]{GuptaKS14}
{\sc Gupta, A.}, {\sc Kumar, A.}, {\sc and} {\sc Stein, C.} 2014.
\newblock Maintaining assignments online: Matching, scheduling, and flows.
\newblock In {\em Proceedings of the Twenty-Fifth Annual {ACM-SIAM} Symposium
  on Discrete Algorithms, (SODA)}. 468--479.

\bibitem[\protect\citeauthoryear{Gupta and Singla}{Gupta and
  Singla}{2020}]{gupta2020random}
{\sc Gupta, A.} {\sc and} {\sc Singla, S.} 2020.
\newblock Random-order models.
\newblock {\em arXiv preprint arXiv:2002.12159\/}.

\bibitem[\protect\citeauthoryear{Holyer}{Holyer}{1981}]{holyer1981np}
{\sc Holyer, I.} 1981.
\newblock The np-completeness of edge-coloring.
\newblock {\em SIAM Journal on Computing (SICOMP)\/}~{\em 10,\/}~4, 718--720.

\bibitem[\protect\citeauthoryear{Joag-Dev and Proschan}{Joag-Dev and
  Proschan}{1983}]{joag1983negative}
{\sc Joag-Dev, K.} {\sc and} {\sc Proschan, F.} 1983.
\newblock Negative association of random variables with applications.
\newblock {\em The Annals of Statistics\/}, 286--295.

\bibitem[\protect\citeauthoryear{Karande, Mehta, and Tripathi}{Karande
  et~al\mbox{.}}{2011}]{karande2011online}
{\sc Karande, C.}, {\sc Mehta, A.}, {\sc and} {\sc Tripathi, P.} 2011.
\newblock Online bipartite matching with unknown distributions.
\newblock In {\em Proceedings of the 43rd Annual ACM Symposium on Theory of
  Computing (STOC)}. 587--596.

\bibitem[\protect\citeauthoryear{Karloff and Shmoys}{Karloff and
  Shmoys}{1987}]{karloff1987efficient}
{\sc Karloff, H.~J.} {\sc and} {\sc Shmoys, D.~B.} 1987.
\newblock Efficient parallel algorithms for edge coloring problems.
\newblock {\em J. Algorithms\/}~{\em 8,\/}~1, 39--52.

\bibitem[\protect\citeauthoryear{Kesselheim, T{\"o}nnis, Radke, and
  V{\"o}cking}{Kesselheim et~al\mbox{.}}{2014}]{kesselheim2014primal}
{\sc Kesselheim, T.}, {\sc T{\"o}nnis, A.}, {\sc Radke, K.}, {\sc and} {\sc
  V{\"o}cking, B.} 2014.
\newblock Primal beats dual on online packing lps in the random-order model.
\newblock In {\em Proceedings of the 46th Annual ACM Symposium on Theory of
  Computing (STOC)}. 303--312.

\bibitem[\protect\citeauthoryear{Khursheed and Lai~Saxena}{Khursheed and
  Lai~Saxena}{1981}]{khursheed1981positive}
{\sc Khursheed, A.} {\sc and} {\sc Lai~Saxena, K.} 1981.
\newblock Positive dependence in multivariate distributions.
\newblock {\em Communications in Statistics - Theory and Methods\/}~{\em
  10,\/}~12, 1183--1196.

\bibitem[\protect\citeauthoryear{K{\"o}nig}{K{\"o}nig}{1916}]{konig1916graphen}
{\sc K{\"o}nig, D.} 1916.
\newblock {\"U}ber graphen und ihre anwendung auf determinantentheorie und
  mengenlehre.
\newblock {\em Mathematische Annalen\/}~{\em 77,\/}~4, 453--465.

\bibitem[\protect\citeauthoryear{Korula, Mirrokni, and Zadimoghaddam}{Korula
  et~al\mbox{.}}{2018}]{korula2018online}
{\sc Korula, N.}, {\sc Mirrokni, V.}, {\sc and} {\sc Zadimoghaddam, M.} 2018.
\newblock Online submodular welfare maximization: Greedy beats 1/2 in random
  order.
\newblock {\em SIAM Journal on Computing (SICOMP)\/}~{\em 47,\/}~3, 1056--1086.

\bibitem[\protect\citeauthoryear{Mahdian and Yan}{Mahdian and
  Yan}{2011}]{mahdian2011online}
{\sc Mahdian, M.} {\sc and} {\sc Yan, Q.} 2011.
\newblock Online bipartite matching with random arrivals: an approach based on
  strongly factor-revealing lps.
\newblock In {\em Proceedings of the 43rd Annual ACM Symposium on Theory of
  Computing (STOC)}. 597--606.

\bibitem[\protect\citeauthoryear{Megow, Skutella, Verschae, and Wiese}{Megow
  et~al\mbox{.}}{}]{MegowSVW12}
{\sc Megow, N.}, {\sc Skutella, M.}, {\sc Verschae, J.}, {\sc and} {\sc Wiese,
  A.}
\newblock The power of recourse for online {MST} and {TSP}.
\newblock In {\em Proceedings of 39th International Colloquium on Automata,
  Languages, and Programming, (ICALP)}. Vol. 7391. 689--700.

\bibitem[\protect\citeauthoryear{Meyerson}{Meyerson}{2001}]{meyerson2001online}
{\sc Meyerson, A.} 2001.
\newblock Online facility location.
\newblock In {\em Proceedings of the 42nd Symposium on Foundations of Computer
  Science (FOCS)}. 426--431.

\bibitem[\protect\citeauthoryear{Motwani, Naor, and Naor}{Motwani
  et~al\mbox{.}}{1994}]{motwani1994probabilistic}
{\sc Motwani, R.}, {\sc Naor, J.~S.}, {\sc and} {\sc Naor, M.} 1994.
\newblock The probabilistic method yields deterministic parallel algorithms.
\newblock {\em Journal of Computer and System Sciences\/}~{\em 49,\/}~3,
  478--516.

\bibitem[\protect\citeauthoryear{Petersen}{Petersen}{1898}]{petersen1898theoreme}
{\sc Petersen, J.} 1898.
\newblock Sur le th{\'e}oreme de tait.
\newblock {\em L'interm{\'e}diaire des Math{\'e}maticiens\/}~{\em 5}, 225--227.

\bibitem[\protect\citeauthoryear{Su and Vu}{Su and Vu}{2019}]{su2019towards}
{\sc Su, H.} {\sc and} {\sc Vu, H.~T.} 2019.
\newblock Towards the locality of vizing's theorem.
\newblock In {\em Proceedings of the 51st Annual ACM Symposium on Theory of
  Computing (STOC)}. 355--364.

\bibitem[\protect\citeauthoryear{Tait}{Tait}{1880}]{tait1880remarks}
{\sc Tait, P.} 1880.
\newblock Remarks on the colourings of maps.
\newblock {\em Proc. R. Soc. Edinburgh\/}~{\em 10}, 729.

\bibitem[\protect\citeauthoryear{Vizing}{Vizing}{1964}]{vizing1964estimate}
{\sc Vizing, V.~G.} 1964.
\newblock On an estimate of the chromatic class of a p-graph.
\newblock {\em Diskret analiz\/}~{\em 3}, 25--30.

\bibitem[\protect\citeauthoryear{Wajc}{Wajc}{2020}]{wajc2020rounding}
{\sc Wajc, D.} 2020.
\newblock Rounding dynamic matchings against an adaptive adversary.
\newblock In {\em Proceedings of the 52nd Annual ACM Symposium on Theory of
  Computing (STOC)}. 194--207.

\end{thebibliography}

\end{document}